\documentclass{article}

\usepackage[utf8]{inputenc}
\usepackage[english]{babel}
\usepackage{geometry}
\usepackage{lmodern}
\usepackage{microtype}
\usepackage{authblk}
\usepackage{titlesec}
\usepackage{ragged2e}
\usepackage{relsize}
\usepackage{amsmath,amsfonts,amssymb,amsthm}
\usepackage{mathtools}
\mathtoolsset{showonlyrefs}
\allowdisplaybreaks
\usepackage{bm}
\usepackage{physics}
\usepackage{dsfont}
\usepackage{bbm}
\usepackage{nicefrac}
\usepackage{stackrel}
\usepackage{scalerel}
\usepackage{letltxmacro}

\usepackage{graphicx}
\graphicspath{{images/}}
\usepackage{xcolor}
\usepackage{float}
\usepackage{wrapfig}
\usepackage[font=footnotesize]{caption}
\usepackage{subcaption}
\usepackage{alphalph}

\usepackage[export]{adjustbox}
\usepackage{multirow}
\usepackage{makecell}
\usepackage{booktabs}
\usepackage{enumitem}
\usepackage{algorithm,algpseudocode}
\usepackage{tikz}
\usetikzlibrary{patterns,decorations.pathreplacing}
\usepackage{etoolbox}
\usepackage{needspace}
\pretocmd{\subsection}{\Needspace{5\baselineskip}}{}{}
\usepackage{csquotes}
\usepackage{epigraph}

\newtheorem{theorem}{Theorem}
\newtheorem{proposition}{Proposition}
\newtheorem{lemma}{Lemma}

\newtheorem{assumption}{Assumption}

\newtheorem{corollary}{Corollary}

\DeclareMathOperator*{\argmin}{arg\,min}
\DeclareMathOperator*{\argmax}{arg\,max}

\DeclareRobustCommand{\Rb}{\mathbb{R}}
\DeclareRobustCommand{\bD}{\boldsymbol{D}}

\DeclareRobustCommand{\bA}{\boldsymbol{A}}

\DeclareRobustCommand{\bE}{\boldsymbol{E}}
\DeclareRobustCommand{\bU}{\boldsymbol{U}}
\DeclareRobustCommand{\bV}{\boldsymbol{V}}
\DeclareRobustCommand{\bS}{\boldsymbol{S}}
\DeclareRobustCommand{\bM}{\boldsymbol{M}}
\DeclareRobustCommand{\bQ}{\boldsymbol{Q}}
\DeclareRobustCommand{\bP}{\boldsymbol{P}}
\DeclareRobustCommand{\bX}{\boldsymbol{X}}

\DeclareRobustCommand{\bB}{\boldsymbol{B}}
\DeclareRobustCommand{\bW}{\boldsymbol{W}}

\newcommand{\bI}{\boldsymbol{I}}

\newcommand{\bY}{\boldsymbol{Y}}

\DeclareRobustCommand{\bSigma}{\boldsymbol{\Sigma}}
\DeclareRobustCommand{\hA}{\widehat{{A}}}

\DeclareRobustCommand{\bhA}{\widehat{\boldsymbol{A}}}
\DeclareRobustCommand{\bhU}{\widehat{\boldsymbol{U}}}
\DeclareRobustCommand{\bhP}{\widehat{\boldsymbol{P}}}
\DeclareRobustCommand{\bhV}{\widehat{\boldsymbol{V}}}
\DeclareRobustCommand{\bhS}{\widehat{\boldsymbol{S}}}

\DeclareRobustCommand{\Hc}{\mathcal{H}}
\DeclareRobustCommand{\Jc}{\mathcal{J}}
\DeclareRobustCommand{\Ic}{\mathcal{I}}
\DeclareRobustCommand{\Ec}{\mathcal{E}}

\DeclareRobustCommand{\Rc}{\mathcal{R}}

\DeclareRobustCommand{\Oc}{\mathcal{O}}
\newcommand{\Gc}{\mathcal{G}}

\newcommand{\Vc}{\mathcal{V}}
\DeclareRobustCommand{\Sc}{\mathcal{S}}

\DeclareRobustCommand{\Rc}{\mathcal{R}}
\DeclareRobustCommand{\BCc}{\mathcal{BC}}

\DeclareRobustCommand{\Rb}{\mathbb R}

\DeclareRobustCommand{\Pb}{\mathbb P}

\DeclareMathOperator{\Cov}{Cov}
\DeclareMathOperator{\Var}{Var}

\newcommand{\bz}{\boldsymbol{z}} 
\newcommand{\bDelta}{\boldsymbol{\Delta}} 
\newcommand{\btheta}{\boldsymbol{\theta}} 
\newcommand{\bXi}{\boldsymbol{\Xi}}
\newcommand{\bomega}{\boldsymbol{\omega}} 
\newcommand{\cE}{\mathcal{E}}

\newcommand{\bvarphi}{\boldsymbol{\varphi}}

\newcommand{\tY}{\widetilde{Y}}

\newcommand{\hsigma}{\widehat{\sigma}}

\newcommand{\btY}{\widetilde{\bY}}

\newcommand{\halpha}{\widehat{\alpha}}

\newcommand{\bhQ}{\widehat{\bQ}}

\newcommand{\distas}[1]{\mathbin{\overset{#1}{\kern\z@\sim}}}%
\newsavebox{\mybox}\newsavebox{\mysim}
\newcommand{\distras}[1]{%
  \savebox{\mybox}{\hbox{\kern3pt$\scriptstyle#1$\kern3pt}}%
  \savebox{\mysim}{\hbox{$\sim$}}%
  \mathbin{\overset{#1}{\kern\z@\resizebox{\wd\mybox}{\ht\mysim}{$\sim$}}}%
}

\newcommand{\htheta}{\widehat{\theta}}

\newcommand{\hbeta}{\widehat{\beta}}

\newcommand{\Ex}{\mathbb{E}}

\newcommand{\hvarphi}{\widehat{\varphi}}

\newcommand{\hs}{\widehat{s}}
\newcommand{\bv}{\boldsymbol{v}}
\newcommand{\bx}{\boldsymbol{x}}
\newcommand{\bu}{\boldsymbol{u}}
\newcommand{\hbu}{\widehat{\bu}}
\newcommand{\hbv}{\widehat{\bv}}

\newcommand{\bbeta}{\boldsymbol{\beta}} 
\newcommand{\balpha}{\boldsymbol{\alpha}} 
\newcommand{\hbbeta}{\widehat{\bbeta}}
\newcommand{\bzero}{\boldsymbol{0}}

\newcommand{\hbalpha}{\widehat{\balpha}}
\newcommand{\nAR}{\normalfont{\AR}}
\newcommand{\nAC}{\normalfont{\AC}} 

\newcommand{\hB}{\widehat{B}}

\newcommand{\txtop}{\text{op}}
\newcommand{\eop}{\emph{op}}

\newcommand{\clip}{\text{clip}}
\newcommand{\eclip}{\emph{clip}} 

\newcommand{\ARC}{\textsf{ARC}} 
\newcommand{\nARC}{\normalfont{\ARC}}
\newcommand{\AR}{\textsf{AR}} 
\newcommand{\AC}{\textsf{AC}}

\newcommand{\NR}{\textsf{NR}} 
\newcommand{\NC}{\textsf{NC}}

\newcommand{\USVT}{\texttt{USVT}}

\newcommand{\SNN}{\texttt{SNN}}
\newcommand{\PCR}{\texttt{PCR}}

\newcommand{\anchor}{\texttt{AnchorSubMatrix}}
\newcommand{\softimpute}{\texttt{softImpute}}

\newcommand{\core}{\text{core}}
\newcommand{\tuser}{\text{user}}

\newcommand{\tmovie}{\text{movie}}
\newcommand{\standard}{\text{standard}}
\newcommand{\Cc}{\mathcal{C}}

\newcommand{\by}{\boldsymbol{y}}
\newcommand{\bbX}{\bar{\bX}}
\newcommand{\bbY}{\bar{\bY}}
\newcommand{\bby}{\bar{\by}}
\newcommand{\bbz}{\bar{\bz}} 
 
\newcommand{\hbX}{\widehat{\bX}} 

\newcommand{\bbW}{\bar{\bW}}

\newcommand{\bxi}{\boldsymbol{\xi}}
\newcommand{\br}{\bar{r}}

\newcommand{\btau}{\bar{\tau}}
\newcommand{\noise}{\texttt{noise}}

\newcommand{\boldb}{\boldsymbol{b}}
\newcommand{\bolda}{\boldsymbol{a}}

\newcommand{\bdelta}{\boldsymbol{\delta}}
\newcommand{\ba}{\boldsymbol{a}}
\newcommand{\Nc}{\mathcal{N}}
\newcommand{\hupsilon}{\widehat{\upsilon}}

\newcommand{\talpha}{\tilde{\alpha}}
\newcommand{\tbalpha}{\tilde{\balpha}}
\newcommand{\homega}{\widehat{\omega}} 
\newcommand{\hbomega}{\widehat{\bomega}}
\newcommand{\hx}{\widehat{x}}
\newcommand{\hbx}{\widehat{\bx}}
\newcommand{\Tc}{\mathcal{T}}
\newcommand{\heta}{\widehat{\eta}}
\newcommand{\YW}{\texttt{Yan-Wainwright}}
\newcommand{\bGamma}{\boldsymbol{\Gamma}}

\usepackage[numbers,square]{natbib}
\usepackage{hyperref}
\hypersetup{
    colorlinks=true,
    linkcolor=blue,
    citecolor=magenta,
    urlcolor=cyan,
    pdfpagemode=FullScreen,
    pdftitle={Synthetic Nearest Neighbors: Extending Synthetic Controls for Matrix Completion with Missing Not at Random Data},
    pdfauthor={Anish Agarwal, Munther Dahleh, Devavrat Shah, Dennis Shen}
}

\title{\scshape Synthetic Nearest Neighbors:\\
Extending Synthetic Controls for Matrix Completion with Missing Not at Random Data}
\author[1]{Anish Agarwal}
\author[2]{Munther Dahleh}
\author[2]{Devavrat Shah}
\author[3]{Dennis Shen}
\affil[1]{\small Department of Industrial Engineering \& Operations Research\\
Columbia University, \texttt{aa5194@columbia.edu}}
\affil[2]{\small Department of Electrical Engineering \& Computer Science\\
Massachusetts Institute of Technology, \texttt{munther@mit.edu}, \texttt{devavrat@mit.edu}}
\affil[3]{\small Department of Data Sciences \& Operations\\
University of Southern California, \texttt{dennis.shen@marshall.usc.edu}}
\date{}

\begin{document}
\maketitle

\vspace{-15pt}
\begin{abstract}
We develop a causal framework for matrix completion under missing not at random (MNAR) data. Drawing on synthetic controls from the econometric panel data literature, our approach relaxes two assumptions common in MNAR matrix completion: positivity and independence of observation indicators. 
Unlike traditional panel data models, which often require prescribed block-sparse geometries, our framework accommodates flexible, heterogeneous observation patterns through target-specific local information structures.
We propose synthetic nearest neighbors (\SNN), a local synthetic-controls-inspired estimator, and establish finite-sample entrywise error bounds and consistency for mean recovery under suitable conditions. 
We further derive asymptotic normality under heteroskedastic noise and develop feasible entrywise inference.
To estimate entry-specific noise variances, we apply the same local principle to squared outcomes, obtaining consistency under bounded noise and asymptotic unbiasedness under general subgaussian noise. 
Simulation studies corroborate the theoretical findings across a range of missingness designs and observation patterns.

\end{abstract}

\medskip
\noindent
{\smaller
}

%
\section{Introduction} \label{sec:intro} 
Matrix completion seeks to recover an underlying matrix from partial, often noisy, observations. 
Owing to its broad applicability, this field has witnessed rapid growth. 
Traditionally, theoretical guarantees in matrix completion are established with respect to global error metrics, such as the Frobenius norm, which quantify the average performance across all entries. 
While these metrics are important, they can obscure localized errors and fail to provide actionable insights in settings where individual predictions are consequential. 
For instance, in personalized recommendations, accuracy at the level of specific entries---not just on average---is often of chief concern. 
Moreover, these recovery guarantees typically rest on two core assumptions: 
(i) the underlying matrix has a low-complexity structure (e.g., low-rank), and
(ii) the entries are {\em missing completely at random} (MCAR), i.e., each entry is observed independently with uniform probability $p > 0$, regardless of its value. 
In many practical applications, however, the MCAR assumption is frequently violated. 

One motivating example arises in recommender systems, a canonical application of matrix completion. 
Here, user-item interactions are represented as a matrix, where rows correspond to users, columns to items, 
and the $(i,j)$th entry records user $i$'s rating of item $j$. 
Observations in such systems are often subject to {\em selection bias}. 
Sports enthusiasts are more likely than non-fans to watch {\em The Last Dance}; meat-lovers are more likely than vegetarians to review steakhouses. Platform-side interventions can amplify these patterns: users searching for Grand Canyon trails are more likely to see ads for hiking boots than for wedding heels.
Across these cases, the user's preferences and the system's inferred beliefs about those preferences shape the sparsity pattern in such a way that the observed entries are influenced by the underlying values of the matrix itself---a phenomenon known as {\em missing not at random} (MNAR). 

Recent works have made significant strides in addressing MNAR matrix completion, demonstrating that algorithms explicitly accounting for MNAR mechanisms often outperform standard methods designed under the MCAR assumption.
Yet, on the theoretical front, key questions remain open. 
As highlighted by \cite{ma2019missing}, two major limitations persist in the existing literature: 
(i) {\em positivity}---every entry must have a strictly positive probability of being observed; 
(ii) {\em independence}---entries are assumed to be observed independently of one another. 
This work aims to challenge both assumptions, which can prohibit applicability in many real-world settings. 

To move beyond these constraints, we build on a growing line of research that connects matrix completion with the panel data framework in econometrics. 
In panel data, rows index units (e.g., states) and columns index time-treatment pairs (e.g., economic policies over time) such that the $(i, (d,t))$th entry records the outcome of unit $i$ under treatment $d$ at time $t$. 
Like recommender systems, panel data often exhibits MNAR patterns, as policymakers may selectively implement programs based on expected outcomes.
However, the structure of policy implementation introduces additional sparsity constraints: only one policy can be applied at a time, and policies, once introduced, generally cannot be removed. 
These constraints clearly violate positivity and independence, but they also induce {\em block sparsity patterns} that are more structured and more restrictive than the irregular sparsity patterns that naturally arise in matrix completion problems. 

Thus, the panel data literature offers frameworks that relax positivity and independence, but often under restrictive observation geometries. Standard matrix completion allows richer and more irregular observation patterns, but usually relies on stronger probabilistic assumptions. 
This intersection motivates the central question of this paper: 
{\em can we obtain entrywise guarantees for matrix completion under MNAR data, without relying on the limiting assumptions of positivity and independence, and for a broad class of observation patterns?} 

%

\subsection{How the Missingness Mechanism Affects Matrix Recovery: a Teaser} \label{sec:teaser}
We begin with a simple illustration designed to isolate the role of the missingness mechanism. To remove the confounding effect of noise, we consider a noiseless setting and compare recovery under three observation designs. 
\begin{enumerate} [label=(\alph*)]
	\item {\bf MCAR}: Each entry is observed independently with probability $p=0.35$.
	
	\item {\bf Conventional MNAR}: Each entry is observed with its own probability that is correlated with the underlying value. The standard assumptions of positivity and independence are maintained.
	
	\item {\bf Panel-based MNAR}: Observation probabilities again depend on the underlying values, but now the two classical assumptions fail: some entries may have zero probability of being observed, and observations may be dependent across entries.
\end{enumerate}
Figure~\ref{fig:sparsity_patterns} reveals the observation patterns induced by these three mechanisms. 
A recommender system interpretation provides useful intuition. Under MCAR, users rate movies entirely at random. Under conventional MNAR, users are more likely to rate movies they strongly liked or disliked and less likely to rate movies toward which they were indifferent. 
Under panel-based MNAR, users rate a common set of core movies and, among the remaining titles, only those in their preferred genres. This represents a more extreme form of selection. Detailed constructions of designs (b) and (c) are provided in Appendix~\ref{sec:teaser.details} of the Supplementary Material. 
%

As benchmarks, we consider two widely used matrix completion algorithms: universal singular value thresholding (\USVT) of \cite{Chatterjee15}, a spectral method, and \softimpute~of \cite{softimpute14}, an optimization-based method. Both were originally designed for MCAR data. To adapt them to MNAR settings, we apply standard debiasing techniques following \cite{bhattacharya2021matrix} and \cite{ma2019missing}. We also include a preview of our proposed method, {\em synthetic nearest neighbors} (\SNN), which is introduced formally in Section~\ref{sec:estimator}. 
The regularization hyperparameters for de-biased \softimpute~and \SNN~are selected using five-fold cross-validation.

Figure~\ref{fig:teaser.err} reports the entrywise root mean squared error (RMSE) for each method. Under MCAR, de-biased \softimpute~and \SNN~achieve substantially lower error than modified \USVT. Under the conventional MNAR mechanism, \SNN~remains accurate, whereas both benchmark methods incur considerably larger errors. Under panel-based MNAR, \SNN~continues to have the lowest average error, although its performance is more dispersed across trials. 

These results illustrate that the observation mechanism can materially affect matrix recovery. Even methods equipped with standard debiasing corrections may deteriorate when positivity or independence fails. This motivates the development of procedures that remain reliable across a broad range of MNAR mechanisms and observation patterns.

\begin{figure*}[t!]
	\centering 
	\begin{subfigure}{0.22\linewidth}
		\centering 
		\includegraphics[width=\textwidth]
		{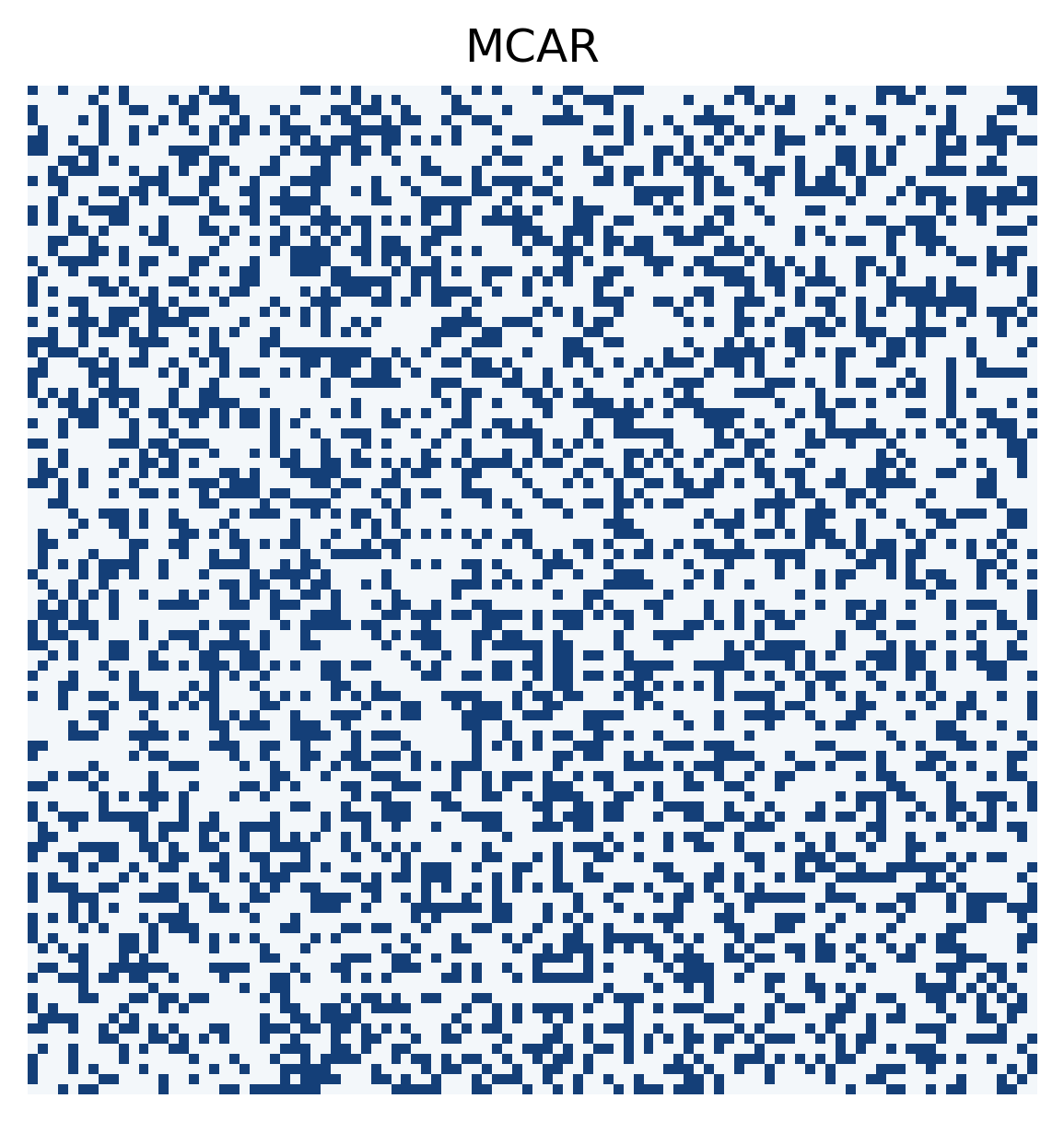}
		\caption{MCAR.} 
		\label{fig:sparsity_MCAR}
	\end{subfigure} 
	\quad	
        \begin{subfigure}{0.22\linewidth}
		\centering 
		\includegraphics[width=\textwidth]
		{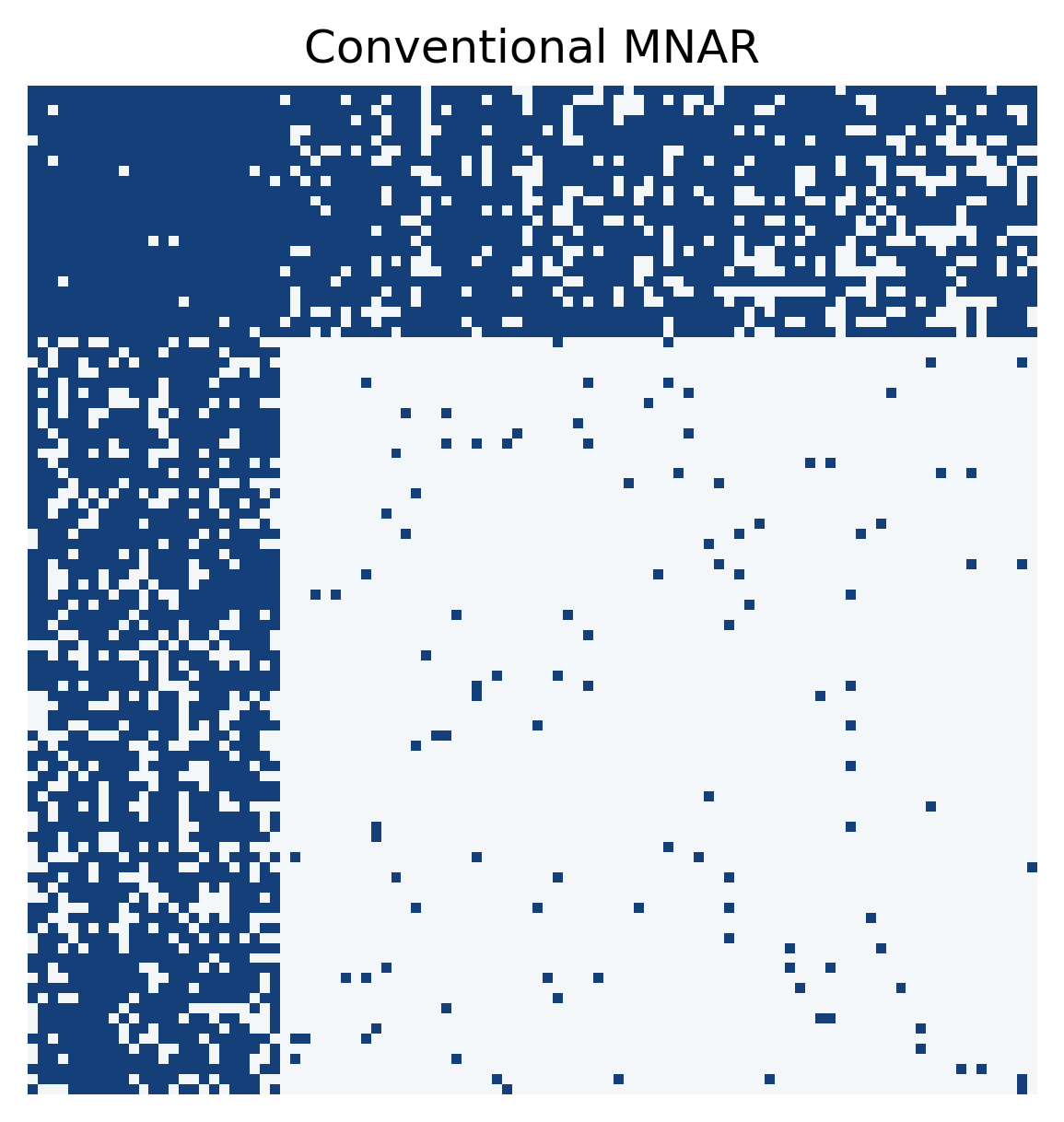}
		\caption{Conventional MNAR.} 
		\label{fig:sparsity_limited_MNAR} 
	\end{subfigure}
	\quad 
	\begin{subfigure}{0.22\linewidth}
		\centering 
		\includegraphics[width=\textwidth]
		{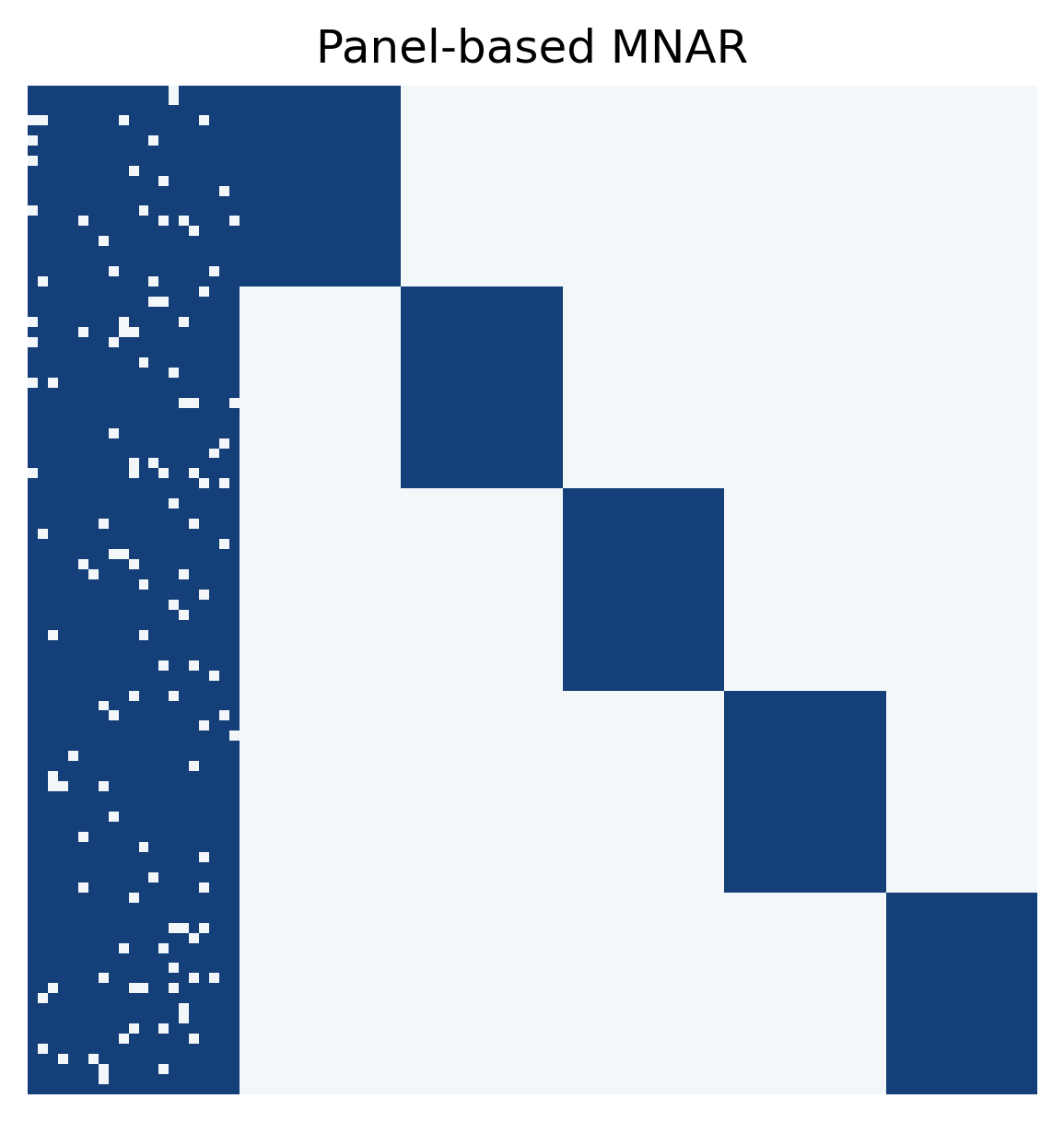}
		\caption{Panel-based MNAR.} 
		\label{fig:sparsity_general_MNAR} 
	\end{subfigure} 
	\caption{Sparsity patterns induced by different missingness mechanisms.}
	\label{fig:sparsity_patterns} 
\end{figure*}

\begin{figure}[t!]
	\centering 
	\begin{subfigure}[b]{0.32\textwidth}
		\centering 
		\includegraphics[width=\linewidth]
		{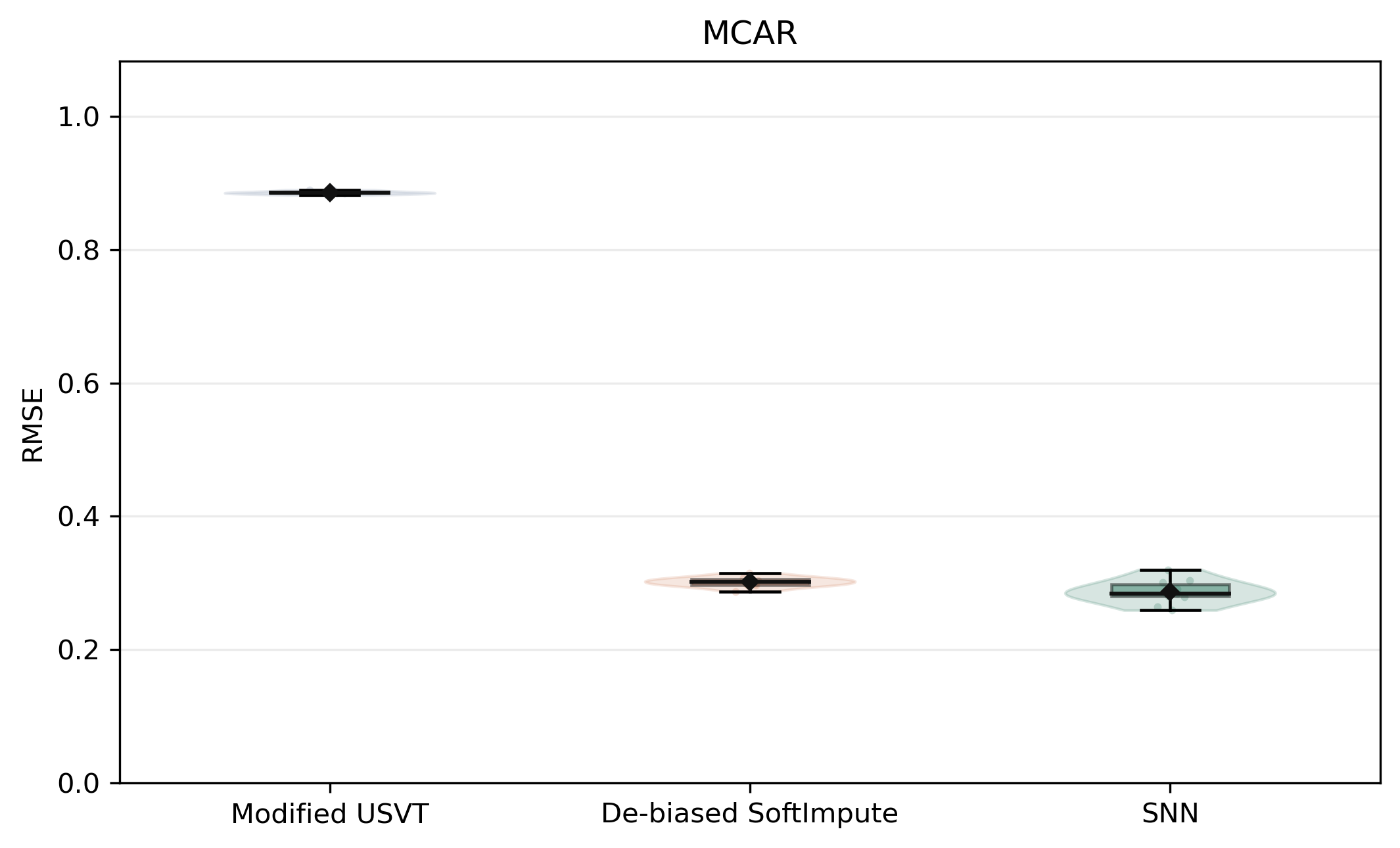}
		\caption{MCAR.} 
		\label{fig:teaser_MCAR_USVT} 
	\end{subfigure}
	\begin{subfigure}[b]{0.32\textwidth}
		\centering 
		\includegraphics[width=\linewidth]
		{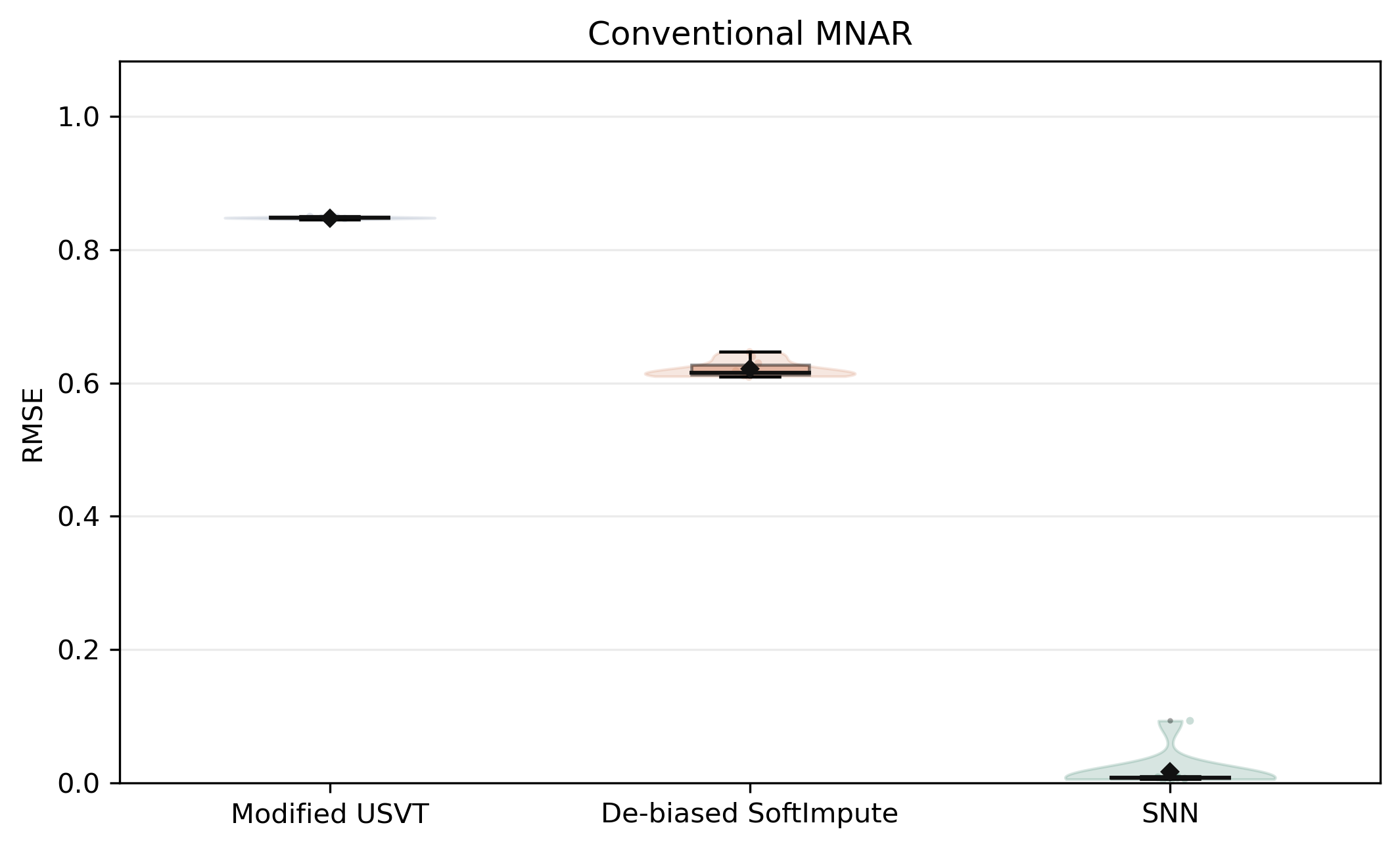}
		\caption{Conventional MNAR.} 
		\label{fig:teaser_MCAR_softimpute} 
	\end{subfigure} 
	\begin{subfigure}[b]{0.32\textwidth}
		\centering 
		\includegraphics[width=\linewidth]
		{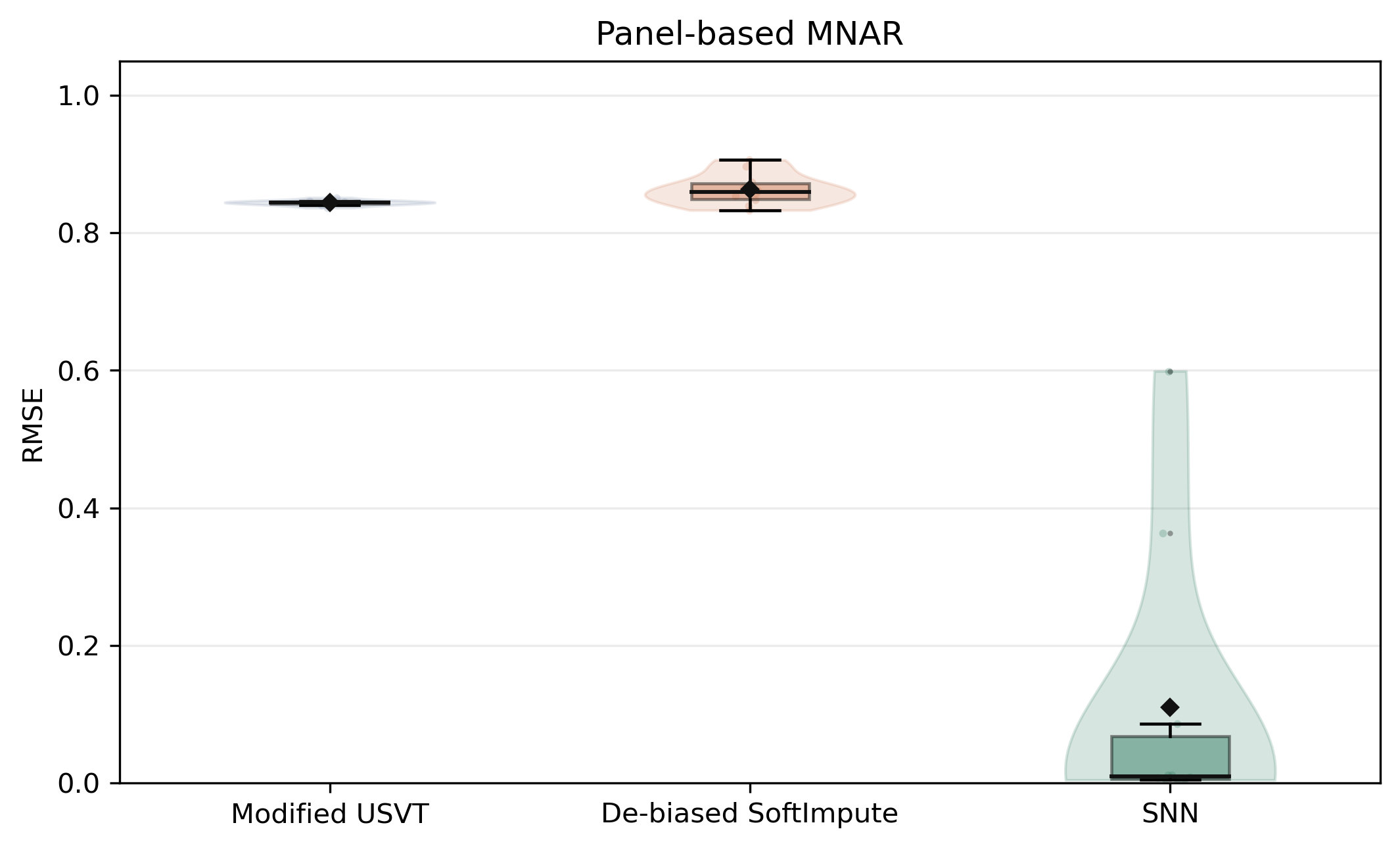}
		\caption{Panel-based MNAR.} 
		\label{fig:teaser_MCAR_SNN} 
	\end{subfigure}
	\caption{Entrywise RMSE for modified \USVT, de-biased \softimpute, and \SNN~under (a) MCAR, (b) conventional MNAR, and (c) panel-based MNAR.}
	\label{fig:teaser.err} 
\end{figure}

\subsection{Contributions} 
This paper develops a causal framework for matrix completion under a broad class of MNAR mechanisms. The framework is inspired by the panel data literature, but is designed for the more irregular sparsity patterns that arise in matrix completion. 

Within a low-rank factor model, we provide an identification principle for individual entries. We show that an entry can be identified from observed data even when its probability of observation may be zero, when observation indicators may be dependent across entries, and when the observation mechanism may be correlated with the underlying signal. 
This directly relaxes the positivity and independence assumptions that underlie much of the existing MNAR matrix completion literature. Relative to standard panel data models, our approach also relaxes the requirement of a prescribed block-sparse geometry. 
Instead, identification is target-specific: a given entry is recoverable whenever the observed mask contains a {\em sufficiently informative local pattern} around that entry (see the conditions of Theorem~\ref{thm:consistency} and the subsequent discussion).

To operationalize this framework, we introduce \SNN, an estimator that blends the local nature of nearest neighbor methods with the reweighting logic of synthetic controls. 
Under suitable assumptions, we establish finite-sample entrywise error bounds and consistency for mean 
recovery (Theorem \ref{thm:consistency}). 
We also show that the theory specializes to the MCAR setting as a benchmark, and we exhibit an observation design under which a shared local information pattern yields uniform entrywise recovery with transparent sample-complexity trade-offs. 
Beyond point estimation, we prove asymptotic normality under heteroskedastic noise (Theorem \ref{thm:normality}) 
and develop feasible studentized inference (Theorem \ref{thm:var.est.asymp}). 
The accompanying variance estimator applies the \SNN~principle to the squared observed matrix. 
We prove that this estimator is consistent under bounded noise and asymptotically unbiased under general subgaussian noise (Theorem \ref{thm:var.est}). 
Simulation studies support the theoretical results across a range of missingness mechanisms and observation patterns.

%
\subsection{Related Work} 

\noindent \textbf{Matrix completion.} 
Low-rank matrix completion has been extensively studied under the MCAR assumption. 
Within this regime, a wide array of methods have been proposed, including optimization-based approaches \citep{CandesTao10, Recht11, softimpute}, spectral techniques \citep{KeshavanMontanariOh10a, donoho14, Chatterjee15}, and collaborative filtering algorithms \citep{goldberg1992using, koren2015advances}.
For a comprehensive overview of standard assumptions and theoretical guarantees, see \cite{ieee_matrix_completion_overview}.  

A more nuanced paradigm is {\em missing at random} (MAR), where the missingness is allowed to depend on observed covariates. 
Conditional on the covariates, the outcomes and missingness mechanism are assumed to be independent, an assumption also known as unconfoundedness or selection on observables in the causal inference literature \citep{imbens_rubin_2015}.  
This naturally leads to methods that estimate propensity scores based on covariates. 
Key contributions in this regime include \cite{Liang2016, schnabelfwang16, WangNeurips2018} and \cite{wang2019}. 

The most general and challenging regime is MNAR, which encompasses all mechanisms beyond MCAR and MAR. 
As discussed in \cite{ma2019missing}, existing MNAR approaches hinge on positivity and independence. 
A common strategy under these assumptions is to model the propensity score matrix as low-rank and estimate it using matrix completion techniques on the observation mask \citep{ma2019missing, bhattacharya2021matrix}. 
This generalizes MAR approaches by removing the dependence on covariate information. 
Additional works addressing MNAR mechanisms include \cite{sportisse2020estimation_PCA, causal_recommender_systems, foucart} and \cite{zhu2019high}. 
Many prior works in the matrix completion literature do not explicitly distinguish MAR from MNAR; for a detailed discussion of this distinction, see \cite{little2019statistical}. 

Most existing analyses focus on global estimation error metrics. A few works provide entrywise error bounds  \citep{LeeLiShahSong16, chen_2020, chen_2021}, but their analyses rely heavily on independent sampling and therefore do not directly cover settings with structured missingness such as ours. Their bounds also scale inversely with a polynomial of the minimum propensity, making positivity indispensable. 
Recently, \cite{kanxu25} investigates matrix completion in the context of matching markets, providing entrywise error guarantees for a nuclear-norm-regularized estimator. 
Their framework accommodates dependent observations but also allows entries to be repeatedly observed. 
This stands in contrast to our setting, where each entry is observed at most once. 

\medskip 
\noindent \textbf{Panel data.}
A growing body of work connects matrix completion with econometric panel data models, which often assume that the potential outcomes matrix is low rank \citep{amjad2018robust, amjad2019mrsc, athey2021matrix, tianyi_1, fernandez2020low, lihua25}. 
This line of research is also closely related to work on factor models \citep{bai2019matrix, CAHAN2023113}. 
\cite{amjad2018robust, amjad2019mrsc, athey2021matrix} provide an important contribution  that connected matrix completion literature with panel data setting and establishing global error bounds under block-sparsity pattern.
\cite{Choi20092024} refine approach of \cite{athey2021matrix} by partitioning missing entries into smaller groups and applying convex estimation to each group, obtaining entrywise error bounds and asymptotic normality for certain statistics.  
\cite{yan2024entrywiseinferencemissingpanel} build on this idea by developing a computationally efficient algorithm with improved entrywise inference for staggered adoption designs. 

Our contribution is complementary to this line of work in scope rather than in dominance. Whereas panel-factor methods generally exploit global treatment-timing geometry, \SNN~instead relies on a target-specific recovery condition. 
Panel-specific methods may therefore yield sharper rates or stronger inference when their global geometry is credible. 
Our framework instead targets settings where the observation geometry is too irregular to be captured by classical panel designs, but still contains enough local information for entrywise recovery.

\medskip 
\noindent \textbf{Earlier versions.}
An earlier version of this work circulated as \cite{agarwal2021causalmatrixcompletion}, and a shorter extended abstract appeared at the {\em Conference on Learning Theory} in 2023. Those versions introduced the causal formulation, a preliminary \SNN~estimator, and an oracle asymptotic normality result. The present article gives a substantially expanded treatment. 
Methodologically, it introduces a simpler estimator for mean recovery, a heteroskedastic variance estimation procedure based on the same local principle, and a more tractable subroutine for finding informative local submatrices. Statistically, it develops a sharper asymptotic normality result, feasible studentized inference under heteroskedastic noise, and a theory for low-rank heteroskedastic variance estimation. It also situates \SNN~more explicitly relative to subsequent work on panel-factor models.

\subsection{Paper Organization} 
Section~\ref{sec:problem_setup} introduces our causal framework and formalizes the causal estimand. 
%
%
Section~\ref{sec:estimator} presents the \SNN~algorithm. 
Section~\ref{sec:theoretical_results} establishes statistical guarantees for \SNN~as an estimator of the mean matrix.
Section~\ref{sec:var.est} develops feasible entrywise inference and extends the analysis to heteroskedastic variance estimation.
Section~\ref{sec:sims} reports simulation studies. 
%
%
Section~\ref{sec:conclusion} concludes. 
%

\subsection{General Notation} 
For a positive integer $n$, let $[n] = \{1, \dots, n\}$. 
For index sets $\Ic \subseteq [m] $ and $\Jc \subseteq [n]$, let $\bM_{\Ic, \Jc}$ denote the $| \Ic | \times | \Jc |$ sub-matrix of $\bM \in \Rb^{m \times n}$ whose rows and columns are indexed by $\Ic$ and $\Jc$.
For a vector $\bv \in \Rb^m$, define $\bv_{\Ic}$ analogously. 
We reserve $\bI$ and $\bzero$ as the identity matrix and matrix/vector of zeros. 
Let $\dagger$ denote the Moore-Penrose pseudoinverse and $\circ$ the entrywise (Hadamard) product. 
For $p \in [1, \infty]$, let $\|\bv\|_p$ denote the $\ell_p$-norm. 
%
%
Denote the operator and Frobenius norms as $\| \bM \|_\txtop$ and $\| \bM \|_F$. 
The subgaussian and sub-exponential norms of a random vector $\bv$ are denoted by $\| \bv \|_{\psi_2}$ and $\|\bv\|_{\psi_1}$. 
Convergence in probability and distribution are denoted as $\xrightarrow{p}$ and $\rightsquigarrow$. 

\section{A Causal Framework for Matrix Completion} \label{sec:problem_setup}
We adopt the potential outcomes framework \citep{neyman, rubin} to formulate matrix completion as a causal estimation problem. 
Let $\btY = [\tY_{ij}] \in \Rb^{m \times n}$ denote the matrix of potential outcomes, where $\tY_{ij}$ is the outcome that would be observed if unit $i \in [m]$ was exposed to treatment $j \in [n]$. 
Let $\bD = [D_{ij}] \in \{0,1\}^{m \times n}$ denote the matrix of treatment assignments. 
The observed outcomes matrix $\bY = [Y_{ij}] \in \{\Rb \cup \{\star\}\}^{m \times n}$, where $\star$ denotes a missing value, is defined entrywise as 
\begin{align}\label{eq:sutva}
	Y_{ij} \coloneqq \tY_{ij} \cdot \boldsymbol{1} \{D_{ij} = 1\} + \left( \star \right) \cdot \boldsymbol{1} \{D_{ij} = 0 \}. 
\end{align} 
The observation law of \eqref{eq:sutva}, known as the {\em stable unit treatment value assumption}, is a standard assumption in the causal inference literature that precludes interference \citep{imbens_rubin_2015}. 

\subsection{Matrix Factor Model}
Without additional assumptions, the available observations offer no information about a missing entry.
To enable a faithful recovery of the underlying matrix, we impose a low-rank structure on the expected outcomes, as is standard in matrix completion.

\begin{assumption}[Low-rank factor model]\label{assump:lfm} 
For every pair $(i,j)$, let $\tY_{ij} = \left\langle \bu_i, \bv_j \right \rangle + \varepsilon_{ij}$, 
where $\bu_i, \bv_j \in \Rb^r$ are latent row and column factors, and $\varepsilon_{ij} \in \Rb$ is idiosyncratic noise. 
\end{assumption} 

With $\varepsilon_{ij}$ denoting (typically zero-mean) noise, $\langle \bu_i, \bv_j \rangle$ represents the underlying signal. 
When $r \ll \min\{m, n\}$, the expected potential outcomes can be explained by significantly fewer factors than the dimensions of the matrix. 
This renders the matrix recovery problem well-posed and identifiable under sensible conditions. 
For these reasons, the signal-plus-noise factor model in Assumption~\ref{assump:lfm} is foundational in both matrix completion \citep{mc_survey} and econometrics factor models \citep{athey2021matrix}. 
Empirical and theoretical work also supports the prevalence of approximate low-rank structure in data science applications \citep{udell2018big}, and smooth nonlinear latent variable models can often be well approximated by low-rank representations \citep{xu2017rates}. 

\subsection{Causal Estimand} \label{sec:framework_causal_parameters}
Anchoring on Assumption~\ref{assump:lfm}, we define the causal estimand for entry $(i,j)$ as $$A_{ij} \coloneqq \Ex\left[\tY_{ij} \mid \bu_i, \bv_j \right],$$ the expected potential outcome for the $(i,j)$th entry, conditional on its latent factors. 
The choice of $(i,j)$ is arbitrary: the analysis targets entrywise recovery, allowing the error associated with an individual missing entry to be studied directly. 
%
Appendix~\ref{sec:add.params} of the Supplementary Material extends this framework to linear functionals of the mean matrix, including weighted averages and contrasts.

\subsubsection{Additional Notation} 
%
Let $\bU \in \Rb^{m \times r}$ and $\bV \in \Rb^{n \times r}$ collect the latent row and column factors, and let $\bE = [\varepsilon_{ij}] \in \Rb^{m \times n}$ denote the noise matrix. 
Define $\Ec = \{\bU, \bV, \bD\}$ as the information set, which we treat as fixed. 
Further, define the neighborhood rows of column $j$ and neighborhood columns of row $i$ by $\NR(j) =  \{\ell \in [m]: D_{\ell j} = 1\}$ and $\NC(i) = \{q \in [n]: D_{iq} = 1\}$. 
We call $\AR \subseteq \NR(j)$ and $\AC \subseteq \NC(i)$ a set of {\em anchor rows} and {\em anchor columns} for the pair $(i,j)$ if $D_{\ell q} = 1$ for all $(\ell, q) \in \AR \times \AC$.  
Thus, anchor rows are observed at the target column, anchor columns are observed for the target row, and the anchor submatrix $\nAR \times \nAC$ is fully observed. 
Together, they form a local observed cross surrounding the missing entry $(i,j)$.
For a visualization, see Figure~\ref{fig:obs_pattern}. 

\begin{figure} [!t]
	\centering 
		\includegraphics[width=0.65\linewidth]
		{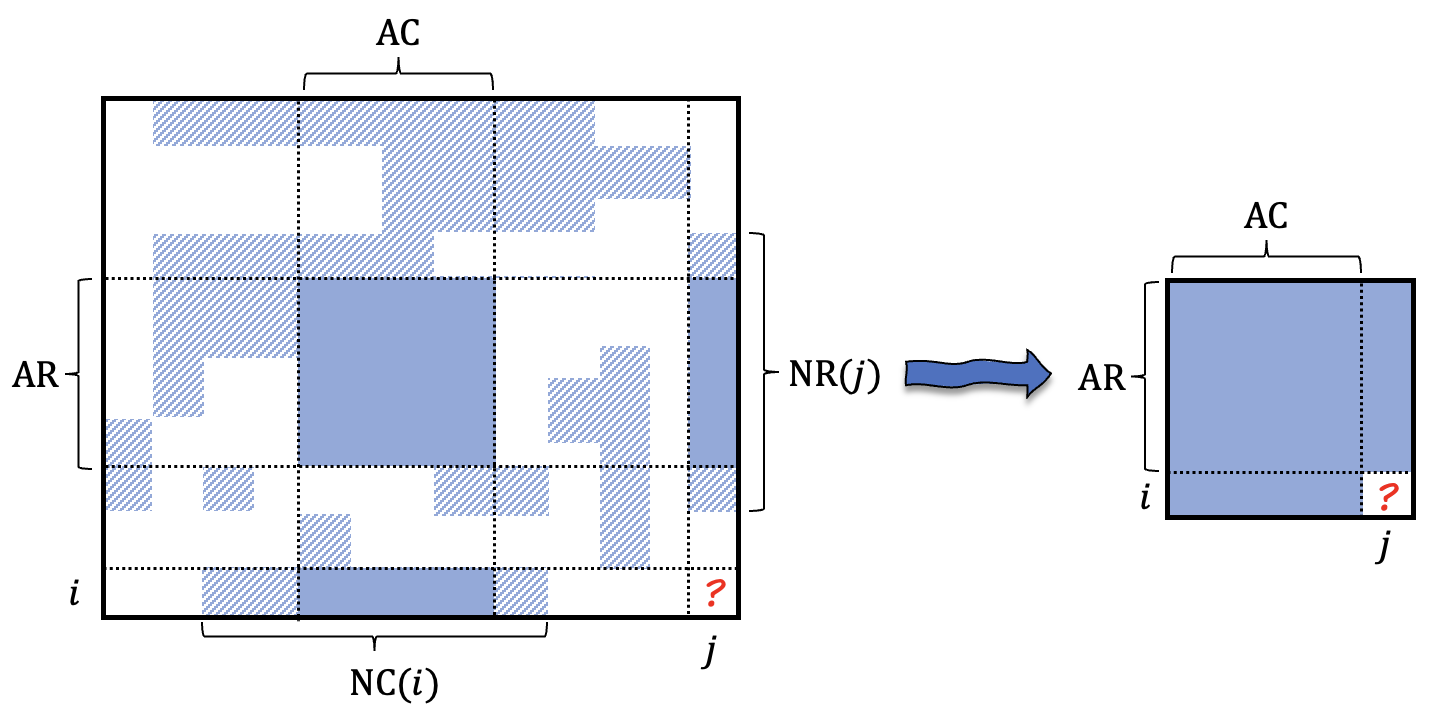}
	\caption{Visualization of the observed matrix $\bY$, highlighting the neighborhood sets $(\NR(j), \NC(i))$ and anchor sets $(\nAR, \nAC)$ associated with the target $(i,j)$ pair. Unobserved entries are depicted in white.}
	\label{fig:obs_pattern} 
\end{figure}

\subsubsection{Assumptions for Estimation} 
We state assumptions to estimate $A_{ij}$ from observed data. 

\begin{assumption}[Selection on latent factors] \label{assump:mean_ind} 
Let $\Ex[ \bE \mid \Ec] = \bzero$. 
\end{assumption} 
By the tower law, Assumption~\ref{assump:mean_ind} implies $\Ex[\bE \mid \bU, \bV] = \bzero$. 
Coupled with Assumption~\ref{assump:lfm}, this implies $\bA = \Ex[\btY \mid \bU, \bV]$ has rank $r$, with $r$ typically much smaller than $m$ and $n$. 
Assumptions~\ref{assump:lfm} and \ref{assump:mean_ind} also jointly imply $\Ex[\btY \mid \bU, \bV] = \Ex[\btY \mid \Ec]$, so the potential outcomes are mean independent of the treatment assignments, conditioned on the latent factors. 
Hence, the latent factors serve as latent confounders, playing an analogous role to observed covariates under the
classical selection on observables assumption.
This form of latent ignorability is conceptually aligned with assumptions explored in panel data models \citep{athey2021matrix, avi_eli_21, agarwal2021synthetic} and related contexts \citep{kallus2018causal}. 
As with any assumption concerning unobserved confounding, Assumption~\ref{assump:mean_ind} is fundamentally untestable. 
Its justification relies on domain expertise and context-specific reasoning about the data-generating process. 

\begin{assumption}[Linear row span inclusion] \label{assump:linear_span}  
Conditioned on $\Ec$ and given anchor rows $\nAR$, we have $\bu_i \in {\normalfont \text{span}}(\{ \bu_\ell: \ell \in \nAR \})$. 
\end{assumption}
Assumption~\ref{assump:linear_span} requires the latent factor of the target row to be expressible as a linear combination of the latent factors associated with the anchor rows. 
This condition arises naturally under the low-rank structure in Assumption~\ref{assump:lfm}. If $\text{span}(\{\bu_\ell: \ell \in \AR\}) = \Rb^r$, then Assumption~\ref{assump:linear_span} follows immediately because $\bu_i \in \Rb^r$. 
In practical terms, enough rows must be observed at column $j$ so that their latent factors span a rich enough subspace to include the target row factor. Although this condition cannot be verified directly because the factors are latent, its plausibility can be assessed empirically through training error diagnostics. 

\subsubsection{Estimation Result} 
%
With our assumptions in hand, we next establish that our causal estimand $A_{ij}$ can be expressed using quantities that are, in principle, estimable from observed data.
\begin{proposition}\label{thm:identification}
%
%
Let Assumptions~\ref{assump:lfm}--\ref{assump:linear_span} hold. 
Then, there exists a $\bbeta \in \Rb^{|\nAR |}$ such that (a) $A_{ij} = \sum_{\ell \in \nAR} \beta_\ell \cdot \Ex[Y_{\ell j} \mid \Ec]$, and (b) $\Ex[Y_{iq} \mid \cE] = \sum_{\ell \in \nAR} \beta_\ell \cdot \Ex[ Y_{\ell q} \mid \Ec]$ for every $q \in \nAC$. 
\end{proposition} 

The first display states that the target estimand $A_{ij}$ can be recovered from the anchor rows once the coefficient vector $\bbeta$ is known. The second display suggests a natural way to estimate $\bbeta$: regress the observed entries of the target row over the anchor columns on the corresponding observed entries of the anchor rows. Section~\ref{sec:estimator} develops this idea into a concrete algorithm.

This framework accommodates a broad class of MNAR mechanisms. In particular, the assignment matrix $\bD$ may depend arbitrarily on the latent factors $\bU$ and $\bV$. Since $\bU \bV^\top$ is the signal matrix of expected potential outcomes, each observation indicator $D_{ij}$
 may be a random or deterministic function of the underlying value $A_{ij}$. 
Moreover, the indicators in $\bD$ may be dependent across entries, and positivity is not required. 
The flexibility of the observation mechanism is supplied by the factor model assumptions. The framework is therefore outcome-model-based, rather than design-based: identification follows from the structure of the potential outcomes and the availability of an informative local anchor cross, not from probabilistic restrictions on the assignment matrix.

\section{An Algorithm for Matrix Completion with MNAR Data} \label{sec:estimator}

%
We introduce {\em synthetic nearest neighbors} (\SNN), an algorithm for matrix completion under MNAR data.
Motivated by Proposition~\ref{thm:identification}, \SNN~takes the following steps to estimate the $(i,j)$th entry. 
%
%
\begin{enumerate} 

	\item {\bf Anchor set discovery}: Obtain anchor sets $(\nAR, \nAC)$ such that the anchor submatrix $\bY_{\AR, \AC} \coloneqq [Y_{\ell q}: (\ell, q) \in \AR \times \AC] \in \Rb^{|\AR| \times |\AC|}$, target-row vector $\bY_{i, \AC} \coloneqq [Y_{iq}: q \in \AC] \in \Rb^{|\AC|}$, and target-column vector $\bY_{\nAR, j} \coloneqq [Y_{qj}: q \in \AC] \in \Rb^{|\AC|}$ are all fully observed. 
	
	\item {\bf Spectral denoising}: Let the singular value decomposition of the anchor submatrix be $\bY_{\AR, \AC} = \sum_{\ell \ge 1} \hs_\ell \hbu_\ell \hbv_\ell^\top$, where $\hs_\ell \in \Rb$ are the singular values in decreasing order, and $\hbu_\ell \in \Rb^{|\AR|}$, $\hbv_\ell \in \Rb^{|\AC|}$ are the corresponding left and right singular vectors. 
Given a spectral threshold $t \le \min\{|\nAR|, |\nAC|\}$, define the rank-$t$ approximation by $\bY^{(t)}_{\AR, \AC} \coloneqq \sum_{\ell = 1}^{t} \hs_\ell \hbu_\ell \hbv_\ell^\top$. 

	\item {\bf Parameter estimation}: Estimate the coefficients by principal component regression (\PCR)  
	\begin{align} 
		\hbbeta&\in \argmin_{\bbeta \in \Rb^{|\AR|}} ~\left\| \bY_{i, \AC} -  \left( \bY^{(t)}_{\AR, \AC} \right)^\top \bbeta \right\|_2^2 
		= \left( \bY^{(t)}_{\AR, \AC} \right)^{\top, \dagger} \cdot \bY_{i, \AC}.
		\label{eq:beta.hat}  
	\end{align} 
	
	\item {\bf Point estimation}: Estimate the target entry by $\hA_{ij} \coloneqq \sum_{\ell \in \AR} \hbeta_\ell \cdot Y_{\ell j} =  \left \langle \bY_{\AR, j}, \hbbeta \right \rangle$. 
	
	\item {\bf Clipping:} If $A_{ij} \in [a, b]$ is known, define 
	$\hA^{\clip}_{ij} \coloneqq \min\left\{ b, \max\left\{ a, \hA_{ij} \right\} \right\}$.

\end{enumerate}

\subsection{Interpretation} \label{sec:snn.intuition}
As illustrated by Figure~\ref{fig:obs_pattern}, \SNN~reduces the original $m \times n$ completion problem to the local array 
\begin{align}
	\begin{pmatrix}
		&\bY_{\nAR, \nAC}~ & \bY_{\AR, j}
		\\
		& \bY^\top_{i, \AC}~ & \star
	\end{pmatrix}. 
\end{align}
Within this local submatrix, all entries are observed except the target entry $(i,j)$. These observed entries provide the information used to estimate $A_{ij}$. 
\SNN, however, does not operate directly on the raw entries. 
Under Assumption~\ref{assump:lfm}, each outcome contains both signal and idiosyncratic noise, whereas Assumption~\ref{assump:mean_ind} implies that the conditional mean of the anchor block is low rank. The leading singular directions of $\bY_{\nAR, \nAC}$ are therefore interpreted as carrying signal, while the smaller singular values primarily reflect noise. Accordingly, \SNN~performs spectral denoising to extract the dominant components, $\bY^{(t)}_{\AR, \AC}$, which serves as an estimate of the low-rank signal $\Ex[\bY_{\nAR, \nAC} \mid \Ec]$. 


After denoising, \SNN~follows the intuition of nearest-neighbor collaborative filtering \citep{goldberg1992using, linden2003amazon}, but with a key modification. Standard nearest-neighbor methods search for one or more rows that are individually close to row $i$. \SNN~instead constructs a synthetic neighbor for row $i$ by reweighting the rows in $\AR$ according to $\hbbeta$. Thus, \SNN~does not require any single anchor row to resemble the target row. Rather, it only requires that a weighted combination of anchor rows approximate the target row.

From this perspective, \SNN~is conceptually aligned with synthetic controls \citep{abadie1, abadie2}. In its canonical form, synthetic controls constrains the elements of $\hbbeta$ to be nonnegative and sum to one. By contrast, \SNN~constrains the fitted weights through \PCR: $\hbbeta$ lies in the linear subspace spanned by the top $t$ left singular vectors of the anchor submatrix. 
This restriction is naturally adapted to the low-rank factor model. 
%
In short, \SNN~localizes estimation to an observed anchor neighborhood around a target entry, as in nearest-neighbor methods, but replaces uniform or distance-based averaging with PCR-based synthetic-control weighting.

Although \SNN~is written in row-side form, the same point estimate admits an equivalent column-side representation. Define the column-side \PCR~estimate by 
\begin{align}
	\hbalpha \coloneqq \left(\bY^{(t)}_{\nAR,\nAC}\right)^\dagger \cdot \bY_{\nAR, j}. 
	\label{eq:alpha.hat} 
\end{align}
By \cite[Theorem 1 and Corollary 1]{sameroot}, $\hA_{ij} \coloneqq \big\langle \bY_{\AR, j}, \hbbeta \big\rangle = \big\langle \bY_{i, \AC}, \hbalpha \big\rangle = \big\langle \hbalpha, \big(\bY^{(t)}_{\nAR,\nAC}\big)^\top \hbbeta \big\rangle$. 
These identities give three equivalent interpretations of \SNN: a row-side synthesis, a column-side synthesis, and a two-sided synthesis that uses both simultaneously through the denoised anchor block.
This symmetry recurs throughout the theory. Assumption~\ref{assump:subspace} in Section~\ref{sec:additional_assumptions} formalizes the column-side span condition, analogous to Assumption~\ref{assump:linear_span}; the inferential theory in Section~\ref{sec:normality} shows how both sides enter the first-order uncertainty of $\hA_{ij}$; and Section~\ref{sec:var.est} turns the resulting decomposition into feasible studentization. 
For concreteness, we retain the row-side notation.

\subsection{Practical Considerations} \label{sec:fine.print}

\noindent \textbf{Discovering anchor sets.} 
The first step of \SNN~is to identify a valid anchor pair $(\nAR, \nAC)$. 
Such pairs arise naturally in panel applications with block-sparse observation patterns, but similar structures also occur in less regular settings.
For instance, \cite{ma2019missing} documented pronounced block patterns in recommender systems datasets, such as the \texttt{Coat} \citep{schnabelfwang16} and \texttt{MovieLens-100k} \citep{movielens}. 
Spectral biclustering \citep{spectral_bicluster} offers one computationally tractable way to locate coherent row-column blocks. Alternatively, the observation mask can be viewed as the incidence matrix of a bipartite graph, in which every valid anchor pair corresponds to a biclique. 
Appendix~\ref{sec:max_biclique} in the Supplementary Material formalizes this reduction and reviews algorithms for finding maximal bicliques. We use these approaches in the examples of
Section~\ref{sec:teaser}. Corollary~\ref{cor:mcar} further shows that valid anchor pairs also arise with high probability under MCAR sampling when the observation probability is sufficiently large. 

\medskip \noindent \textbf{Choosing the spectral threshold.} 
The parameter $t$ controls the degree of spectral regularization. Choosing $t$ too small may discard relevant signal; choosing it too large may reintroduce noise.
A simple diagnostic is to inspect the singular-value spectrum and select $t$ near its elbow. A more data-driven alternative is cross-validation on held-out observed entries. One may also use a universal singular-value cutoff, retaining only those components above a prescribed threshold, as in \cite{donoho14} and \cite{Chatterjee15}. 
%

\medskip \noindent \textbf{Ensemble \SNN.} Appendix~\ref{sec:ensemble} of the Supplementary Material introduces an ensemble extension that constructs synthetic neighbors from multiple anchor crosses and aggregates the resulting estimates.

\section{Entrywise Guarantees for Mean Recovery} \label{sec:theoretical_results} 
This section establishes statistical guarantees for \SNN~as an estimator of the target mean $A_{ij}$. 
Throughout, for any random object $\bW$, write $\bbW \coloneqq \Ex[\bW \mid \Ec]$. 

\subsection{Additional Assumptions for Mean Recovery} \label{sec:additional_assumptions}
We state additional assumptions utilized in our analysis. 

\begin{assumption}[Subgaussian noise] \label{assump:subg}
Conditioned on $\Ec$, $\varepsilon_{ij}$ are independent subgaussian random variables satisfying $\Var(\varepsilon_{ij})= \sigma_{ij}^2 \le \sigma^2$ and $\| \varepsilon_{ij} \|_{\psi_2} \le C_\varepsilon \sigma_{ij}$ for constants $C_\varepsilon, \sigma > 0$. 
\end{assumption} 

While the latent row and column factors may be arbitrarily correlated, Assumption~\ref{assump:subg} imposes independence and subgaussian tail behavior on the noise. Although this is a standard assumption in the matrix completion literature, it may be restrictive in certain applications. A formal treatment under more general noise models remains an important direction for future research.

\begin{assumption}[Bounded signal] \label{assump:bounded}
Let $\langle \bu_i, \bv_j \rangle \in [-1,1]$.  
\end{assumption} 
The bound $[-1,1]$ can be extended to $[a,b]$ for any $a, b \in \Rb$ with $a \le b$.

\begin{assumption}[Well-balanced spectra]\label{assump:spectra}
Given anchor sets $(\nAR, \nAC)$, the condition number $\kappa$ of $\bbY_{\nAR, \nAC}$ satisfies $\kappa^{-1} \ge c_y$, and $\| \bbY_{\nAR, \nAC} \|_F^2 \ge c_y' \cdot |\nAR| \cdot |\nAC|$ for constants $c_y, c_y' > 0$.
\end{assumption} 

Assumption~\ref{assump:spectra} requires the nonzero singular values of the anchor signal matrix to be neither too small nor too uneven. While potentially restrictive, it ensures the low-rank structure is visible above the noise floor. This condition parallels familiar pervasiveness, incoherence, and beta-min requirements from factor models, high-dimensional inference, and matrix completion \citep{chamberlainfactor, beta_min, fan2018eigenvector, bai_matrix, cai2021nonconvex}. Empirically, it can be assessed through the singular-value profile of $\bY_{\nAR, \nAC}$, as discussed in Section~\ref{sec:fine.print}. 

\begin{assumption}[Linear column span inclusion]\label{assump:subspace}
Conditioned on $\Ec$ and given anchor columns $\nAC$, we have $\bv_j \in {\normalfont \text{span}}(\{ \bv_q: q \in \nAC \})$. 
\end{assumption} 

Assumption~\ref{assump:subspace} is the column-side counterpart of Assumption~\ref{assump:linear_span}. Assumption~\ref{assump:linear_span} ensures that the target row $i$ can be synthesized from the anchor rows $\nAR$; Assumption~\ref{assump:subspace} ensures that the target column $j$ can be extrapolated from the anchor columns $\nAC$. Together, they make the local anchor cross informative enough to identify the target entry. 

To see why Assumption~\ref{assump:subspace} is needed, recall that \SNN~estimates the row-side weights by regressing $\bY_{i, \nAC}$ on the denoised anchor block $\big(\bY^{(t)}_{\nAR, \nAC}\big)^\top$. If the population anchor block $\bbY_{\nAR, \nAC}$ is rank deficient, the coefficients $\bbeta$ satisfying the row representation in Proposition~\ref{thm:identification} are not unique. \PCR~instead identifies the projection of any valid representer onto the column space of $\bbY_{\nAR, \nAC}$, denoted as 
\begin{align}
	\bbeta^* \coloneqq \bbY_{\nAR, \nAC}^{\top, \dagger} \bbY_{\nAR, \nAC}^\top \cdot \bbeta = \bbY_{\nAR, \nAC}^{\top, \dagger}  \bbY_{i, \nAC}. 
	\label{eq:beta.star}
\end{align}
Assumption~\ref{assump:subspace} ensures that this projection preserves the prediction at the target column. Accordingly, Assumption~\ref{assump:linear_span} provides a representation of the target row, while Assumption~\ref{assump:subspace} validates its extrapolation.
Viewed more broadly, Assumption~\ref{assump:subspace} is a local generalization condition: the held-out factor $\bv_j$ must lie in the span of the training-column factors $\{\bv_q : q \in \AC\}$. 
This is a linear-algebraic analogue of covariate similarity assumptions in statistical learning. 
%

\subsection{Non-asymptotic Error}
Below, we use $\lesssim$ to reflect inequalities that ignore constants factors independent of our model parameters.
For convenience, we denote $|\ARC| \coloneqq |\AR| \cdot |\AC|$ and $L \coloneqq \log(|\nARC|)$. 

\begin{theorem} \label{thm:consistency}
%
%
Let Assumptions~\ref{assump:lfm}--\ref{assump:subspace} hold with $t = \br \coloneqq  \emph{rank}(\Ex[\bY_{\nAR, \nAC}\mid\Ec])$. Suppose that 
\begin{align}
	 |\nARC | \ge C \sigma^2 \br  \left( \sqrt{|\nAR|} + \sqrt{|\nAC|} + \sqrt{L} \right)^2 \label{eq:rank.sep}
\end{align}
and $c \sigma^2 L \le \min\{|\nAR|, |\nAC|\}$ for sufficiently large constants $C, c > 0$. 
Then, conditional on $\Ec$, with probability at least $1 - \mathcal{O}(|\nARC|^{-10})$, 
\begin{align}
	\left| \hA_{ij} - A_{ij} \right| &\lesssim  \Phi \coloneqq \frac{\sigma \br}{\min \left\{ \sqrt{|\nAR|}, \sqrt{|\nAC|} \right\} } 
	+ \frac{\sqrt{\sigma \br} L^{1/4}}{\min\left\{ |\nAC|, |\nAR|^2 / \left(\sigma^2 L \right) \right\}^{1/4}}.
	\label{eq:thm.hp.special}
\end{align}
\end{theorem} 

\noindent \textbf{Rate implications.}
Theorem~\ref{thm:consistency} gives a non-asymptotic entrywise error bound that suggests \SNN~is consistent whenever $\Phi = o(1)$. 
In the balanced regime $|\nAR| \asymp |\nAC|$, the bound reduces to $\Phi = \Oc\left(\br |\nAR|^{-1/2} + \sqrt{\br}L^{1/4}  |\nAR|^{-1/4} \right)$. Thus, a sufficient effective-rank condition is $\br = o(\sqrt{|\nAR| / L} )$. In the favorable asymmetric regime $|\nAC| \asymp |\nAR|^2$, the rate becomes $\Phi = \Oc\left( |\nAR|^{-1/2} \cdot (\br + \sqrt{\br L}) \right)$, yielding $\br = \tilde{o}(\sqrt{|\nAR|})$. Hence, for fixed $\br$, \SNN~approaches the near-parametric $|\nAR|^{-1/2}$ scale. 

The apparent advantage of the asymmetric regime should not be interpreted as an intrinsic benefit of having more anchor columns than anchor rows. It arises from the current perturbation analysis, which controls the row-side \PCR~error using conservative operator-norm bounds. Since \SNN~is symmetric in rows and columns, as highlighted in Section~\ref{sec:snn.intuition}, a sharper analysis may yield more symmetric sufficient conditions.

The auxiliary condition $c \sigma^2 L \le \min\{|\nAR|, |\nAC|\}$ is primarily used to simplify the displayed rate. The more substantive requirement is the signal-to-noise condition in \eqref{eq:rank.sep}, which ensures that the smallest nonzero singular value of the population anchor block dominates the operator-norm fluctuation of the noisy anchor block. 
Finally, Theorem~\ref{thm:consistency} is stated for the oracle choice $t = \br$. 
A formal analysis under misspecified rank is left for future work, although existing results suggest that overestimating the rank is generally less damaging than underestimating it \citep{pcr_aos}.

\medskip \noindent \textbf{Connection to MCAR.}
The next result specializes Theorem~\ref{thm:consistency} to the classical MCAR setting.
\begin{corollary} \label{cor:mcar}
Fix a target entry $(i,j)$, $\delta > 0$, and $0 < \eta < 1$. Let $m = n$, and let $d = o\big(\log(n)\big)$ with $d \rightarrow \infty$. Consider the MCAR regime with the observation probability $p$ satisfying
\begin{align}
	p \ge \left[\frac{C\log(2/\eta)}{\binom{n-1}{d}^2}\right]^{1/[d(d+2)]} 
\end{align}
for a sufficiently large constant $C > 0$. 
Then, for all sufficiently large $n$, with probability at least $1 - \eta/2$, a valid $d \times d$ anchor pair $(\nAR, \nAC)$ exists. 
Suppose further that a mask-measurable selection of such an anchor cross satisfies the conditions of Theorem~\ref{thm:consistency} and $d \ge c \sigma^2 \br^2 \max\{\delta^{-2}, \delta^{-4} \log(d)\}$ for a sufficiently large $c > 0$. Then, over both the MCAR mask and outcome noise, 
\begin{align}
	\Pb\left( \left|\hA_{ij}- A_{ij} \right| \le \delta \mid \bU, \bV \right) \ge 1 - \eta. 
\end{align}
%
%
%
\end{corollary}
Corollary~\ref{cor:mcar} shows that valid anchor crosses arise even under fully random observation patterns, provided the observation probability is large enough (though still scaling as $o(1)$). This is not intended as a sharp MCAR guarantee. Rather, it serves as a bridge: the same local-anchor logic that drives \SNN~under MNAR also recovers the familiar MCAR setting as a special case.

The comparison also clarifies the price of robustness. Suppressing logarithmic factors, classical MCAR matrix completion can exploit global random sampling to obtain a Frobenius error $(1/n) \| \bhA - \bA \|_F = \tilde{\Oc}(\delta)$ with $\tilde{\Oc}(nr \delta^{-2})$ samples. 
Under MNAR sampling, such global sample-complexity statements are generally harder to formulate. 

A more transparent comparison arises when every target can use a shared anchor pair $(\nAR, \nAC)$, as in Figure~\ref{fig:Lshape}. In a recommender system, the anchor rows may represent highly active viewers or critics, while the anchor columns may represent widely viewed canonical films, such as those used during onboarding. 
Up to logarithmic factors, uniform versions of Theorem~\ref{thm:consistency} give $\Oc(\delta)$ entrywise recovery when $|\nAR| = \tilde{\Omega}(\br^2 \delta^{-2})$ and $|\nAC| = \tilde{\Omega}(\br^2 \delta^{-4})$ using $\tilde{\Oc}(n \br^2 \delta^{-4})$ samples.
Thus, when $\br \asymp r$, \SNN~pays an additional factor of order $r \delta^{-2}$ relative to the MCAR benchmark, but provides entrywise guarantees under substantially more general missingness mechanisms. 
A sharper account of the trade-offs between MCAR-tailored and MNAR-robust procedures is left for future work.

\begin{figure} [!t]
	\centering 
		\includegraphics[width=0.3\linewidth]
		{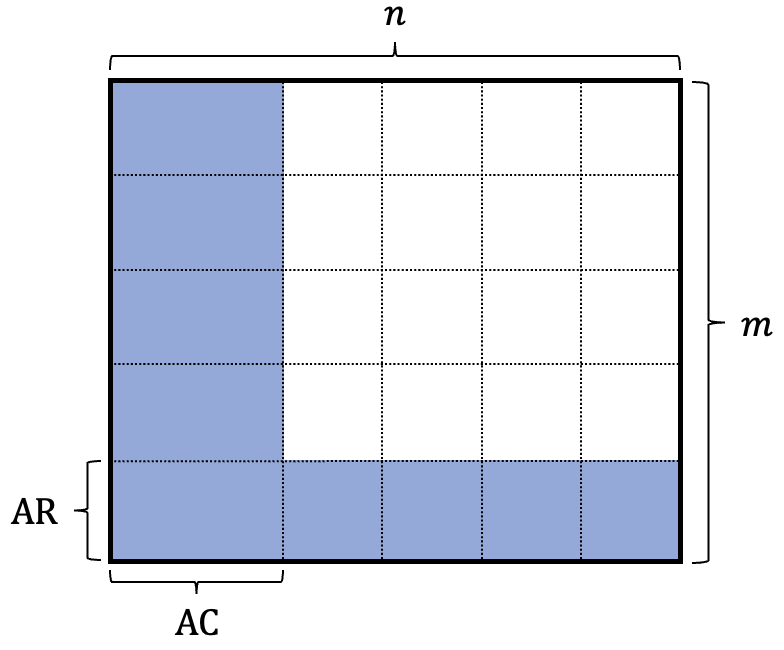}
	\caption{Proposed observation pattern to achieve entrywise matrix recovery for \SNN.}
	\label{fig:Lshape} 
\end{figure}

\subsection{Asymptotic Normality} \label{sec:normality}
We next refine the finite-sample guarantee of Theorem~\ref{thm:consistency} into an asymptotic distributional statement. 
Recall the row-side population representer $\bbeta^*$ defined in \eqref{eq:beta.star}, and define its column counterpart by $\balpha^* \coloneqq \bbY_{\nAR,\nAC}^\dagger \bbY_{\nAR, j}$. 
To facilitate our presentation below, define
\begin{align}
	\upsilon^2 \coloneqq \sum_{\ell \in \nAR} \left(\beta^*_\ell\right)^2 \sigma^2_{\ell j} 
	+ \sum_{q \in \nAC} \left(\alpha^*_q\right)^2 \sigma^2_{iq}
	+ \sum_{\ell \in \nAR} \sum_{q \in \nAC} \left(\beta^*_\ell \alpha^*_q\right)^2 \sigma^2_{\ell q}. 
	\label{eq:var.asymp}
\end{align}

\begin{theorem} \label{thm:normality}
Let the setup of Theorem~\ref{thm:consistency} hold. 
Conditional on $\Ec$, $\hA_{ij} - A_{ij} = G + R$. 
If Assumption~\ref{assump:subg} is specialized to Gaussian noise, then $G / \upsilon \sim \Nc(0,1)$; for general subgaussian noise, if 
\begin{align}
	\frac{\max\left\{ \max_{\ell \in \nAR} \left| \beta^*_\ell \right| \sigma_{\ell j}, ~\max_{q \in \nAC} \left| \alpha^*_q \right| \sigma_{iq}, ~\max_{(\ell, q) \in \nAR \times \nAC} \left| \beta^*_\ell \alpha^*_q \right| \sigma_{\ell q}   \right\}}{\upsilon} = o(1),
	\label{eq:clt}
\end{align}
then $G / \upsilon \rightsquigarrow \Nc(0,1)$. 
Moreover, with probability at least $1 - \Oc(|\nARC|^{-10})$, 
\begin{align}
	\left| R \right| \lesssim 
	\Psi \coloneqq 
	\sigma \Phi \cdot \left( \frac{ \max\left\{\sqrt{\br}, \sqrt{L}\right\}}{\min \left\{ \sqrt{|\nAR|}, \sqrt{|\nAC|} \right\}} + \frac{ \sqrt{\br L}}{\sqrt{|\nARC|}} 
	\right),
	\label{eq:psi}
\end{align}
where $\Phi$ is defined in \eqref{eq:thm.hp.special}. 
%
Consequently, if $\Psi / \upsilon = o(1)$, then as $|\nAR|, |\nAC| \rightarrow \infty$, 
\begin{align}
	\frac{\hA_{ij} - A_{ij}}{\upsilon} \rightsquigarrow \mathcal{N}(0,1). 
\end{align}
\end{theorem}

\noindent \textbf{Decomposition interpretation.} The variance decomposition in \eqref{eq:var.asymp} makes the row-column symmetry of \SNN~statistically consequential. Its first term reflects noise in the target column, weighted by the row representer $\bbeta^*$; its second reflects noise in the target row, weighted by the column representer $\balpha^*$; and its third reflects noise in the anchor block, weighted by their products. Thus, even though the point estimate $\hA_{ij}$ may be written from either side, its first-order uncertainty depends on both. 
The remainder envelope $\Psi$ in \eqref{eq:psi} collects the higher-order effects of estimating the two representers and the anchor subspace. 



\medskip \noindent \textbf{Rate implications.} 
%
Two ingredients are needed for asymptotic normality.
First, the leading fluctuation must satisfy a central limit theorem; condition \eqref{eq:clt} is a Lyapunov-type notion that prevents any single weighted noise entry from dominating the leading term. Second, the remainder must vanish relative to the standard-error scale so that $\Psi / \upsilon = o(1)$. 

For illustration, suppose the relevant variances are bounded below and the population representers are nondegenerate: $\| \bbeta^* \|_2 \gtrsim |\nAR|^{-1/2}$ and $\| \balpha^* \|_2 \gtrsim |\nAC|^{-1/2}$. In the balanced regime $|\nAR| \asymp |\nAC|$, $\Psi / \upsilon = \Oc\left( (\sqrt{\br} + \sqrt{L}) \Phi \right)$, yielding the condition $\br = o(|\nAR|^{1/4} L^{-1/4} )$. 
Thus, inference imposes a stricter effective-rank requirement than consistency: the point estimator must not only converge, but admit a first-order linear approximation on the standard-error scale. The product structure of $\Psi$---first-stage error $\Phi$ multiplied by a perturbation factor---makes this distinction explicit. 

\subsection{Interpretation and Scope}
The guarantees in Theorems~\ref{thm:consistency} and \ref{thm:normality} are inherently local. Recovery of $A_{ij}$ depends on the existence of an informative anchor cross around that entry---not on positivity, independent sampling, staggered adoption, or any prescribed geometry of the full observation mask. Outside the anchor structure, the missingness pattern may be deterministic, dependent, and irregular, subject to our factor model. 

This locality defines both the strength and the limitation of \SNN. The method replaces global design assumptions with local outcome-model conditions: suitable row and column spans and a sufficiently strong anchor signal. When the observation pattern has exploitable global structure, procedures tailored to that structure may be sharper. When the global pattern is irregular but contains informative local crosses, \SNN~provides an entrywise route to recovery and inference.

\section{Feasible Entrywise Inference} \label{sec:var.est}
Theorem~\ref{thm:normality} gives an oracle normal approximation for \SNN~in terms of the unknown leading standard error $\upsilon$. We now convert that result into a feasible entrywise inference procedure. This requires two pieces: an estimator of $\upsilon^2$, and estimators of the heteroskedastic noise variances that enter $\upsilon^2$. 

\subsection{Plug-in Studentization} 
Recall the \PCR~estimates $(\hbbeta, \hbalpha)$ from \eqref{eq:beta.hat} and \eqref{eq:alpha.hat}. 
We estimate $\upsilon^2$ using the natural plug-in estimator 
\begin{align}
	\hupsilon^2 \coloneqq \sum_{\ell \in \nAR} \hbeta_\ell^2 \hsigma^2_{\ell j} 
	+ \sum_{q \in \nAC} \halpha_q^2 \hsigma^2_{iq}
	+ \sum_{\ell \in \nAR} \sum_{q \in \nAC} \hbeta^2_\ell \halpha_q^2 \hsigma^2_{\ell q},
	\label{eq:hvar.asymp}
\end{align}
where each $\hsigma^2$ estimates the corresponding noise variance and is constrained to lie in $[0, \sigma^2_+]$. 

\begin{theorem} \label{thm:var.est.asymp}
Let the setup of Theorem~\ref{thm:normality} hold. Assume further that there exists a deterministic sequence $\Lambda_\sigma$ such that, with probability at least $1-p_\sigma$, 
\begin{align}
	\max\left\{ 
	\max_{\ell \in \nAR} \left| \hsigma^2_{\ell j} - \sigma^2_{\ell j} \right|, 
	\max_{q \in \nAC} \left| \hsigma^2_{iq} - \sigma^2_{iq} \right|, 
	\max_{\ell \in \nAR, q \in \nAC} \left| \hsigma^2_{\ell q} - \sigma^2_{\ell q} \right|
	\right\} \le \Lambda_\sigma.
	\label{eq:var.bound}
\end{align}
%
Then, conditional on $\Ec$, with probability at least $1 - p_\sigma - \Oc(|\nARC|^{-10})$, 
\begin{align}
	\left|\hupsilon^2 - \upsilon^2 \right| 
	\lesssim 
	\Gamma 
	&\coloneqq 
	\left(\sigma^2_+ + \Lambda_\sigma \right) 
	\left\{ \frac{\br\Phi + \Phi^2}{\min\left\{|\nAR|, |\nAC| \right\}} 
	\left( 1 + \frac{\br + \Phi^2}{\min\left\{|\nAR|, |\nAC|\right\}} \right) 
		\right\}
	+  \frac{\br \Lambda_\sigma}{\min\{|\nAR|, |\nAC|\}}. 
	\label{eq:gamma} 
\end{align}
Consequently, if $\Gamma / \upsilon^2 = o(1)$, then as $|\nAR|, |\nAC| \rightarrow \infty$, $\hupsilon / \upsilon \xrightarrow{p} 1$ and thus, 
\begin{align}
	\frac{\hA_{ij} - A_{ij}}{\hupsilon} \rightsquigarrow \Nc(0,1). 
\end{align}
\end{theorem}

\noindent \textbf{Rate implications.} 
Theorem~\ref{thm:var.est.asymp} states that $\hupsilon^2$ consistently estimates $\upsilon^2$, and therefore permits feasible studentization, once two requirements are met. First, the local variance estimates must satisfy the uniform bound in \eqref{eq:var.bound} over the variance entries that appear in $\upsilon^2$. Second, the variance-perturbation envelope $\Gamma$ must be negligible relative to $\upsilon^2$. 
The envelope $\Gamma$ has two sources. The term multiplied by $(\sigma^2_+ + \Lambda_\sigma)$ captures the effect of estimating the \PCR~row and column weights. The final term measures the direct contribution of estimating the local heteroskedastic variances. 

To interpret $\Gamma / \upsilon^2 = o(1)$, consider the same nondegenerate setting used after Theorem~\ref{thm:normality}. In the balanced regime $|\nAC| \asymp |\nAR|$, feasible studentization requires $\br = \tilde{o}(|\nAR|^{1/6})$, which is stronger than the sufficient condition for oracle normality. 
This reflects the additional need to estimate both representers and the local variance profile accurately on the variance scale. 

\medskip 
\noindent \textbf{An entrywise confidence interval.} For any significance level $\gamma \in (0,1)$, Theorem~\ref{thm:var.est.asymp} justifies the entrywise $(1-\gamma)100\%$ confidence interval for $A_{ij}$:
\begin{align}
	\texttt{CI}(\gamma) \coloneqq \left[ \hA_{ij} \pm z_{1-\gamma/2} \cdot \hupsilon \right], \label{eq:conf.iv} 
\end{align}
where $z_{1-\gamma/2}$ is the $(1-\gamma/2)$ quantile of the standard normal distribution.

\subsection{Estimating Heteroskedastic Noise} 
Theorem~\ref{thm:var.est.asymp} reduces feasible entrywise inference to estimating the variance entries that appear in the leading variance formula. We now provide such an estimation procedure.

\subsubsection{Moment Construction} \label{sec:b.strategy}
Theorem~\ref{thm:consistency} establishes recovery of the first moment $A_{ij}$. For variance estimation, we leverage the elementary identity $\sigma^2_{ij} = B_{ij} - A_{ij}^2$, where $B_{ij} \coloneqq \Ex[\tY^2_{ij}\mid\bu_i, \bv_j]$ is the conditional second moment. Let $\bSigma = [\sigma^2_{ij}]$ and $\bB = [B_{ij}]$. If both $\bSigma$ and $\bA$ are low-rank, then $\bB$ inherits a low-rank structure as well, since 
\begin{align}
	\rank(\bB) \le \rank(\bSigma) + \rank(\bA)^2. \label{eq:rankB}
\end{align} 
This observation suggests the following \SNN-based procedure. 

\begin{enumerate}

	\item {\bf First moment estimation}: Apply \SNN~to $\bY$ and clip the resulting estimate to $\hA^\clip_{ij} \in [-1,1]$. 
		
	\item {\bf Second moment estimation}: Apply \SNN~to the elementwise squared matrix $\bY \circ \bY$, obtaining the raw estimate $\hB_{ij}$ and its clipped counterpart $\hB^\clip_{ij} \in [0, 1 +\sigma^2]$, where $\sigma$ is the noise bound from Assumption~\ref{assump:subg}. Let $t_b$ denote the spectral threshold used for this second-moment regression. 

	\item {\bf Variance estimation}: Define   
	$\hsigma^2_{ij} \coloneqq \hB_{ij} -  \big( \hA^{\clip}_{ij} \big)^2$ and $\hsigma^{2, \clip}_{ij} \coloneqq \hB^\clip_{ij} - \big( \hA^{\clip}_{ij} \big)^2$.
	
\end{enumerate}
This moment-subtraction strategy is not specific to \SNN. In principle, it may be paired with any matrix completion method capable of estimating first and second moments under the relevant observation structure.

\subsubsection{Second-Moment Conditions} \label{sec:var.est.assump} 
We impose second-moment analogues of the assumptions used for the mean matrix. 
Write $B_{ij} \coloneqq \langle \btheta_i, \bvarphi_j \rangle$, where $\btheta_i, \bvarphi_j \in \Rb^\zeta$ are latent factors associated with row $i$ and column $j$, and $\zeta$ is compatible with the rank bound in \eqref{eq:rankB}. 

\begin{assumption} \label{assump:linear.span.B}  
Conditioned on $\Ec$ and given anchor sets $(\nAR, \nAC)$, let $\btheta_i \in {\normalfont \text{span}}(\{ \btheta_\ell: \ell \in \nAR \})$ and $\bvarphi_j \in {\normalfont \text{span}}(\{ \bvarphi_q: q \in \nAC \})$. 
\end{assumption}

\begin{assumption} \label{assump:spectra.B}
Given anchor sets $(\nAR, \nAC)$,  the condition number $\kappa$ of $\Ex[\bY_{\nAR, \nAC} \circ  \bY_{\nAR, \nAC} \mid \Ec]$ satisfies $\kappa^{-1} \ge c$, and $\| \Ex[\bY_{\nAR, \nAC} \circ  \bY_{\nAR, \nAC} \mid \Ec] \|_F^2 \ge c' \cdot | \nAR| \cdot |\nAC|$ for constants $c, c' > 0$.
\end{assumption} 

Assumption~\ref{assump:linear.span.B} is the second-moment counterpart of the row- and column-span conditions in Assumptions~\ref{assump:linear_span} and \ref{assump:subspace}. Assumption~\ref{assump:spectra.B}, analogous to Assumption~\ref{assump:spectra}, is the corresponding spectral condition: it requires the second-moment anchor block to have a sufficiently strong low-rank signal.

\subsubsection{Variance-Recovery Guarantees.} \label{sec:var.est.results}
We now state guarantees for heteroskedastic variance estimation in two regimes: bounded noise and general subgaussian noise.

\begin{theorem} \label{thm:var.est}
Let the setup of Theorem~\ref{thm:consistency} hold, and suppose Assumptions~\ref{assump:linear.span.B} and \ref{assump:spectra.B} also hold. 
Let $t_b = \bar{\lambda} \coloneqq \rank(\Ex[\bY_{\nAR, \nAC} \circ \bY_{\nAR, \nAC}\mid\Ec])$, where $\bar{\lambda} \le \btau + \br^2$ and $\btau \coloneqq \rank(\bSigma_{\nAR, \nAC})$.
\begin{enumerate} [label=(\roman*)]
\item \underline{Bounded noise.} Suppose the noise entries $\varepsilon_{ij}$ are bounded. Then, conditional on $\Ec$, with probability at least $1 - \Oc(|\nARC|^{-10})$ 
\begin{align} 
	\left| \hsigma^{2,\eclip}_{ij} - \sigma^2_{ij} \right|
	 \lesssim 
	 \frac{\left(1+ \sigma^2\right) \left(\br^2 + \btau \right)}{ \min \left\{ \sqrt{|\nAR|}, \sqrt{|\nAC|} \right\} } 
     	 + \frac{(1 + \sqrt{\sigma})(\br + \sqrt{\btau}) L^{1/4}}{\min\left\{ |\nAC|, |\nAR|^2 / \left((1+\sigma^2) L \right) \right\}^{1/4}}.
	  \label{eq:var.est.bounded}
\end{align} 

\item \underline{General subgaussian noise.} Suppose the noise entries $\varepsilon_{ij}$ satisfy Assumption~\ref{assump:subg} and $|\hB_{ij} - B_{ij}\mid\le M$ for some deterministic $M > 0$. Then 
\begin{align}
	\left| \Ex\left[ \hsigma^2_{ij} - \sigma^2_{ij} \mid \Ec \right] \right|
	  &\lesssim \frac{\left(1 + \sigma^4 \right) \left(\br^2 + \btau \right) }{\min \left\{ \sqrt{|\nAR|} L^{-3/2}, \sqrt{|\nAC|} \right\}}
	+ \frac{ \left(1 + \sigma^6 \right) \left(\br^2 + \btau \right)^{3/2}}{\min \left\{ |\nAR|   L^{-3}, |\nAC| \right\}} 
    + \frac{M}{|\nARC|^{10}}.
	\label{eq:var.est.general} 
\end{align}
\end{enumerate}
\end{theorem} 

\noindent \textbf{Bounded noise.} 
When $\br^4 + \btau^2 = \tilde{o}( \min \{ |\nAR|, |\nAC| \})$, the clipped variance estimator is consistent. 
This result parallels Theorem~\ref{thm:consistency}: under bounded noise, the squared deviation $Y_{ij}^2 - B_{ij}$ remains subgaussian, allowing the techniques from Theorem~\ref{thm:consistency} to extend. 
In this regime, Theorem~\ref{thm:var.est}(i) can be used inside Theorem~\ref{thm:var.est.asymp} by taking $\Lambda_\sigma$ to be the right-hand side of \eqref{eq:var.est.bounded}. A union bound over the variance entries appearing in \eqref{eq:var.bound} then gives $p_\sigma \coloneqq \Oc(|\nARC|^{-9})$.

\medskip \noindent \textbf{General subgaussian noise.} In this broader regime, squared observations are sub-exponential rather than subgaussian. This changes the concentration behavior of the anchor block and prevents a direct high-probability analogue of the bounded-noise result. 
The theorem therefore establishes asymptotic unbiasedness, not consistency. In particular, it does not by itself provide the uniform variance estimation bound $\Lambda_\sigma$ required for Theorem~\ref{thm:var.est.asymp}.

A sufficient condition for asymptotic unbiasedness is $\br^2 + \btau = \tilde{o}( \min \{ \sqrt{|\nAR|}, \sqrt{|\nAC|} \} )$ and $M = o\left(|\nARC|^{10} \right)$. In contrast to the bounded noise setting, the balanced regime $|\nAC| \asymp |\nAR|$ is most favorable for the stated unbiasedness guarantee. 

\subsection{Scope of the Heteroskedastic Extension} 
Recent panel-factor inference results have primarily emphasized homoskedastic, or effectively homoskedastic, noise. In this respect, Theorems~\ref{thm:var.est.asymp} and \ref{thm:var.est} provide a useful heteroskedastic extension. They allow uncertainty to be calibrated at the same local, target-specific level as the point estimator itself, which is meaningful when some rows, columns, or anchor blocks are intrinsically noisier than others. This refinement should not be read as a universal replacement for existing inference methods: when homoskedasticity is credible, or when the observation design has exploitable global structure, specialized panel methods may remain simpler or sharper. Moreover, feasible heteroskedastic inference carries a real cost: it requires estimating the local variance profile in addition to the two representers, and the strongest guarantees currently available rely on bounded noise.

Theorem~\ref{thm:var.est} may also be useful beyond its role in Theorem~\ref{thm:var.est.asymp}. 
The variance matrix describes irreducible outcome uncertainty, identifies entries whose outcomes are intrinsically unstable, supports risk-sensitive decisions, and can serve as a nuisance input for other inferential procedures. From this perspective, Theorem~\ref{thm:var.est} addresses a distinct statistical problem: recovering a heteroskedastic variance matrix from a single partially observed matrix under MNAR sampling. Thus, beyond enabling studentized \SNN~intervals, the heteroskedastic extension may be serviceable whenever entry-specific uncertainty is itself a substantive object of interest.

\section{Simulation Studies} \label{sec:sims}
We conduct two simulation studies to evaluate the finite-sample performance of \SNN. The first examines the coverage and length of the feasible entrywise confidence intervals developed in Section~\ref{sec:var.est} under heteroskedastic noise. The second compares \SNN~with the panel-factor estimator of \cite{yan2024entrywiseinferencemissingpanel} in a staggered adoption setting specifically tailored to their method.

\subsection{Coverage Study}
We first evaluate the finite-sample performance of the feasible \SNN~confidence interval in \eqref{eq:conf.iv} under heteroskedastic noise. 

\subsubsection{Setup} For each matrix dimension $n \in \{25, 50, 100, 250, 500\}$, we generate an $n \times n$ outcome matrix. We designate the lower-right entry $(n,n)$ as the target, mask that entry, and observe all remaining entries. The signal matrix is deterministic for each $n$ and is constructed from smooth Fourier row and column features. Define $\rho_i \coloneqq (i-1/2)/n$, $\theta_j \coloneqq (j-1/2)/n+0.173$, and
$f(t) \coloneqq \big[1,\sqrt{2}\sin(2\pi t),\sqrt{2}\cos(2\pi t),\sqrt{2}\sin(4\pi t),\sqrt{2}\cos(4\pi t)\big]^\top$. 
Let the diagonal matrix of Fourier coefficients be $\bGamma \coloneqq \operatorname{diag}(0.50,0.38,0.30,0.23,0.17)$ and define the unscaled signal by
$A^{\mathrm{raw}}_{ij} \coloneqq f(\rho_i)^\top \bGamma f(\theta_j)$.
The coefficients place the greatest weight on the intercept and first harmonic, with smaller weights on the second harmonic. The phase shift in $\theta_j$ distinguishes the row and column factors while retaining the same smooth Fourier basis. We then multiply the entire matrix by a common scaling constant so that $|A_{ij}| \le 0.8$ for every $(i,j)$. By construction, $\rank(\bA) = 5$, and the target row and target column belong to the spans of the corresponding observed rows and columns.


For each observed entry, the outcome is generated as $Y_{ij} = A_{ij} + \varepsilon_{ij}$. We consider two noise distributions. In the Gaussian setting, $\varepsilon_{ij} \sim \Nc(0, \sigma^2_{ij})$, whereas in the bounded noise setting, $\varepsilon_{ij} \sim \texttt{Uniform}\big[-\sqrt{3} \sigma_{ij}, \sqrt{3} \sigma_{ij} \big]$. Both distributions have variance $\sigma^2_{ij}$. The heteroskedastic variance matrix is generated as a bounded, positive, low-rank Fourier surface and normalized to satisfy $(1/n^2) \sum_{i=1}^n \sum_{j=1}^n \sigma^2_{ij} = \sigma^2$, where $\sigma = 0.3$. Thus, the two experiments share the same entrywise variance profile and differ only in the shape of the noise distribution. 

For each matrix dimension and noise distribution, we conduct 100 independent replications. \SNN~is supplied with the oracle signal rank of five so that the experiment isolates the performance of the inferential procedure rather than the accuracy of rank selection. In each replication, we construct the nominal 90\% confidence interval $\big[ \hA_{nn} \pm z_{0.95} \hupsilon \big]$, where $\hupsilon$ is obtained from the plug-in variance estimator in \eqref{eq:hvar.asymp}. Its entrywise noise-variance inputs are estimated using the moment construction in Section~\ref{sec:b.strategy}, which applies \SNN~to both the observed outcome matrix and its entrywise square. 

\subsubsection{Results} Figure~\ref{fig:inference} reports three measures of finite-sample performance. Panel (a) shows the RMSE of the estimated leading variance, Panel (b) shows empirical coverage, and Panel (c) shows average interval length. The variance estimation error declines sharply as $n$ grows. At the same time, empirical coverage remains close to the nominal 90\% level under both noise distributions, while the the average interval length steadily decreases. Taken together, these results indicate that increasing the local sample size improves precision without producing visible undercoverage. 

The uniform noise experiment lies within the bounded-noise setting of Theorem~\ref{thm:var.est}(i), and its results are consistent with the feasible inference guarantee in Theorem~\ref{thm:var.est.asymp}. The Gaussian experiment exhibits similar behavior but lies outside the scope of the current feasible inference theory. The Gaussian results therefore provide empirical evidence that the procedure may perform well beyond the bounded noise regime, rather than a formal verification of coverage in that broader setting.

\begin{figure}[t!]
	\centering 
	\begin{subfigure}[b]{0.32\textwidth}
		\centering 
		\includegraphics[width=\linewidth]
		{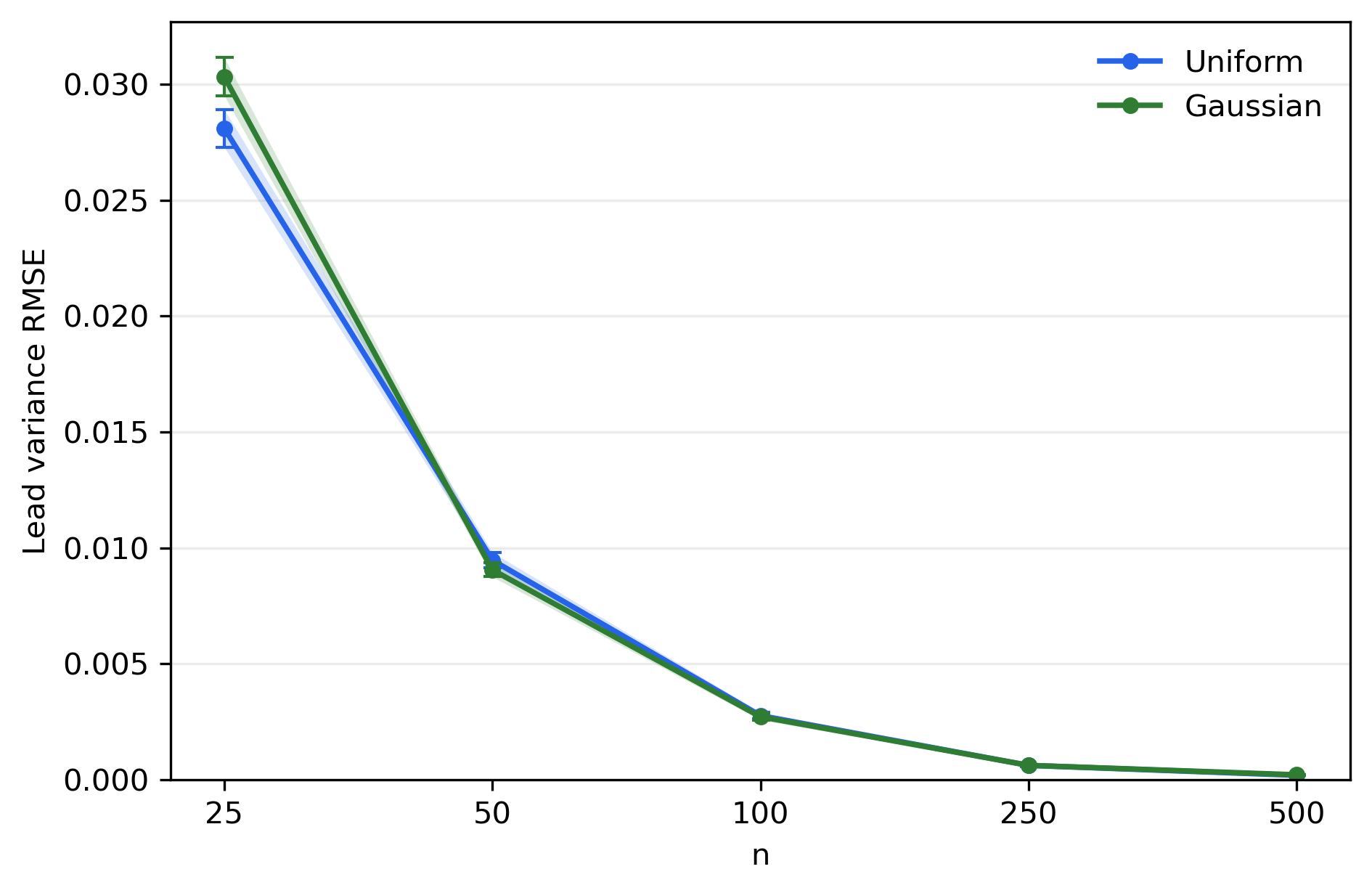}
		\caption{Lead variance estimation.} 
		\label{fig:lead.var} 
	\end{subfigure} 
	\begin{subfigure}[b]{0.32\textwidth}
		\centering 
		\includegraphics[width=\linewidth]
		{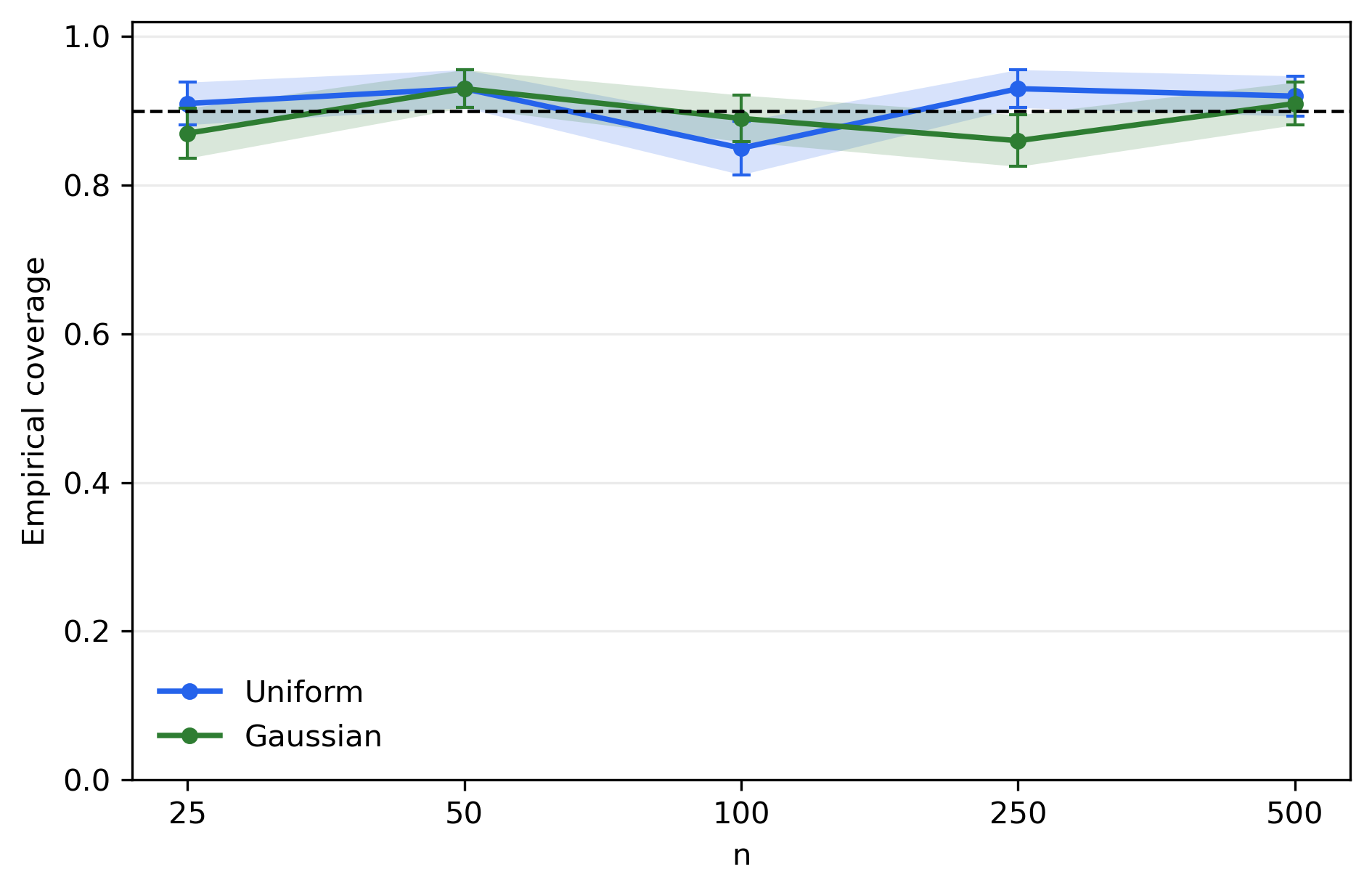}
		\caption{Empirical coverage.} 
		\label{fig:coverage} 
	\end{subfigure}
	\begin{subfigure}[b]{0.32\textwidth}
		\centering 
		\includegraphics[width=\linewidth]
		{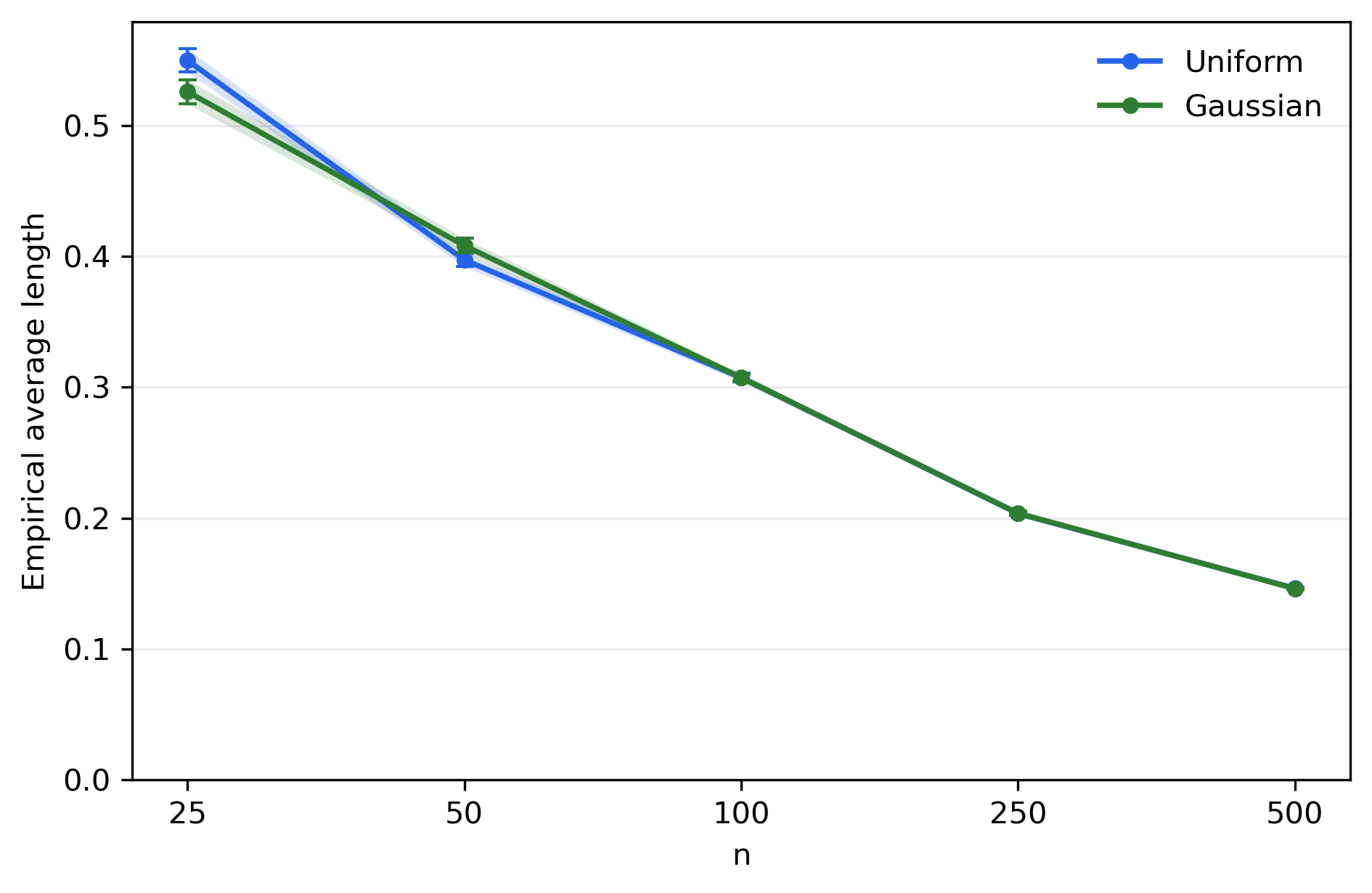}
		\caption{Average interval length.} 
		\label{fig:length} 
	\end{subfigure}
	\caption{Finite-sample performance of feasible \SNN~inference under heteroskedastic noise: (a) RMSE of the estimated leading variance, (b) empirical coverage of nominal 90\% confidence intervals, and (c) average confidence-interval length.}
	\label{fig:inference} 
\end{figure}

\subsection{Comparison with Related Panel-Factor Work} 
We next compare \SNN~with the estimator of \cite{yan2024entrywiseinferencemissingpanel}, which we refer to as the \YW~estimator. We reproduce their ``Scaling with Rank'' experiment from Section 5.2 of their paper, which uses a staggered adoption observation pattern. 

\subsubsection{Setup} The simulation uses a $500 \times 500$ matrix, targets the final entry, and varies the signal rank $r$ from 1 to 30. Following \cite{yan2024entrywiseinferencemissingpanel}, we consider their two data-generating cases. Both estimators are supplied with the oracle rank $r$, so the comparison isolates estimation performance rather than rank-selection accuracy. We refer the reader to \cite{yan2024entrywiseinferencemissingpanel} for the remaining details of the simulation design. 

\subsubsection{Results} Figure~\ref{fig:yw_snn} reports the MSEs of each estimator as a function of the signal rank. 
In Case 1, the errors of \SNN~and the \YW~estimator are closely aligned across the full range of ranks. In Case 2, \SNN~has modestly higher MSE. This difference is unsurprising: the simulated observation pattern has exactly the staggered adoption structure for which the \YW~estimator is designed. \SNN, by contrast, is intended for a broader class of observation masks and bases each estimate only on the information contained in a local anchor cross.

The comparison illustrates the trade-off between specialization and robustness. When a credible global structure, such as staggered adoption, is available, an estimator tailored to that structure may achieve better finite-sample performance. \SNN~instead remains applicable when the global observation pattern is irregular, provided that sufficiently informative local anchor crosses are available, as demonstrated in Section~\ref{sec:teaser}. \SNN~should therefore be viewed as complementary to, rather than a replacement for, panel-specific estimators.


\begin{figure}[t!]
	\centering 
	\begin{subfigure}[b]{0.35\textwidth}
		\centering 
		\includegraphics[width=\linewidth]
		{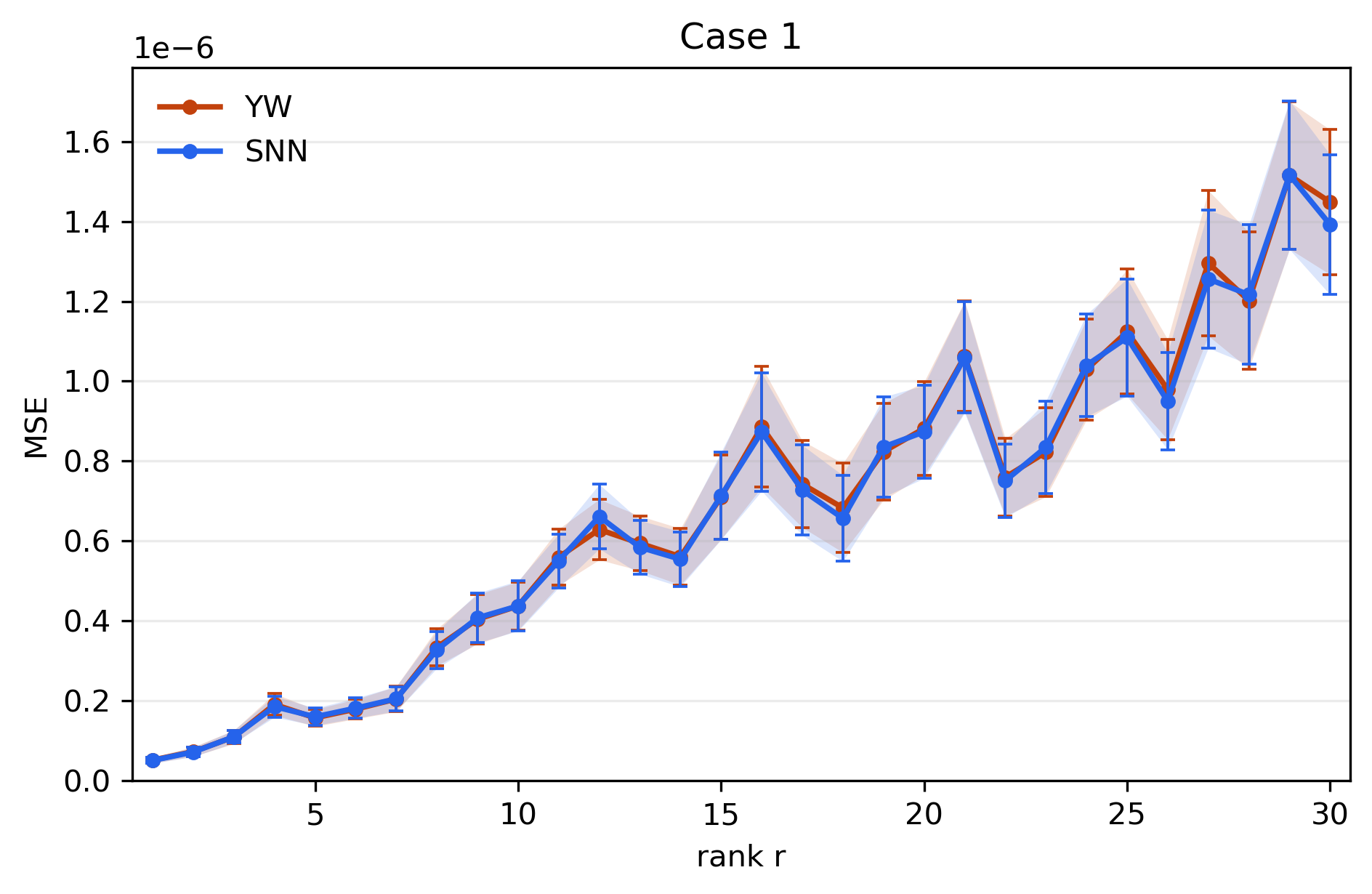}
		\caption{Case 1.} 
		\label{fig:stagger.1} 
	\end{subfigure}
	\quad 
	\begin{subfigure}[b]{0.35\textwidth}
		\centering 
		\includegraphics[width=\linewidth]
		{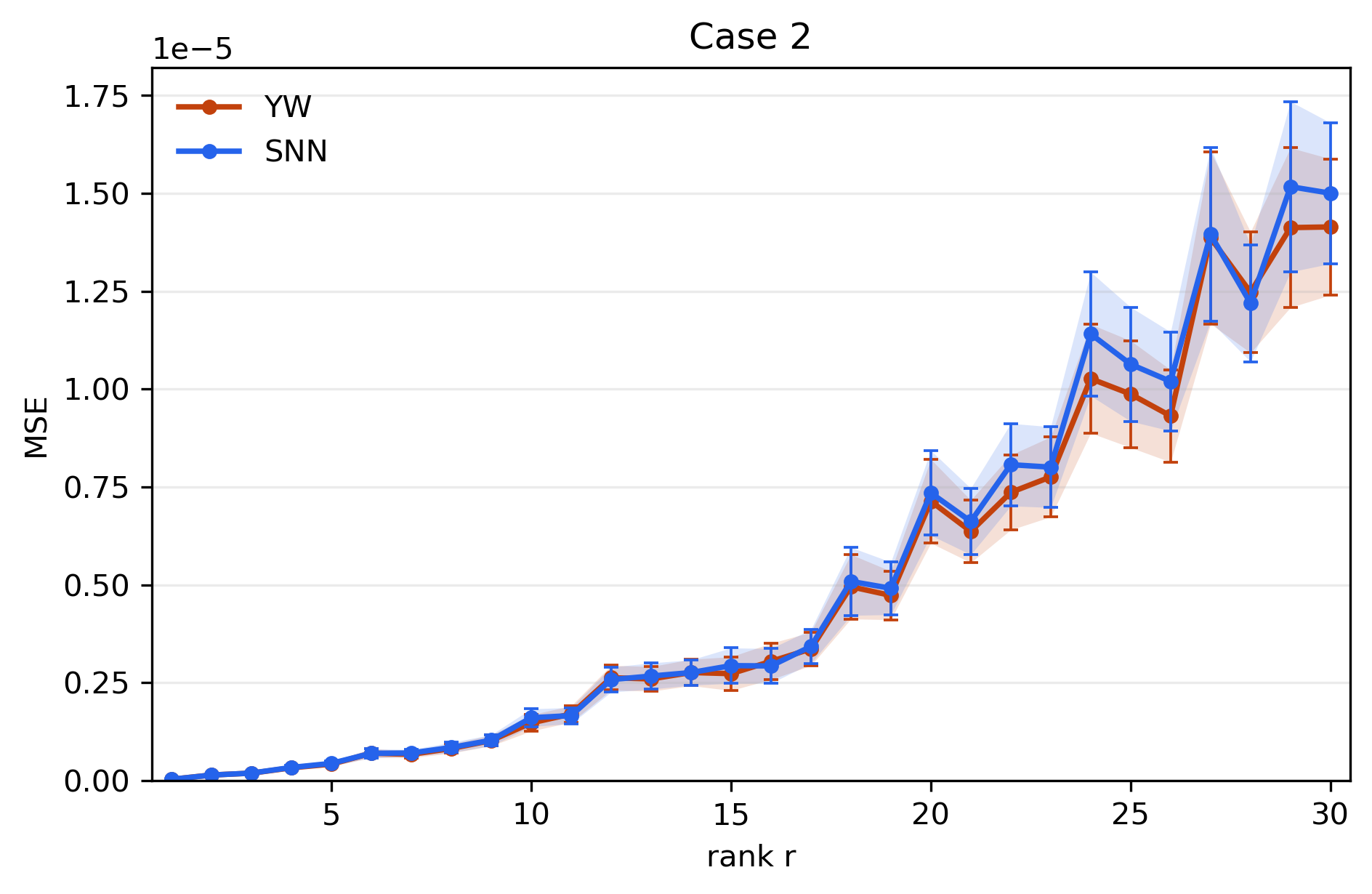}
		\caption{Case 2.} 
		\label{fig:stagger.2} 
	\end{subfigure} 
	\caption{MSE as a function of signal rank for \SNN~and the estimator of \cite{yan2024entrywiseinferencemissingpanel}: (a) Case 1 and (b) Case 2.}
	\label{fig:yw_snn} 
\end{figure}

\section{Conclusion} \label{sec:conclusion}
This article develops a causal framework for matrix completion under MNAR sampling. 
The framework draws on ideas from synthetic controls and panel data, but targets a different regime from much of the existing literature. Unlike standard MNAR matrix completion approaches, it does not require positivity or independent observation indicators; an entry may have zero probability of being observed, and the observation pattern may be dependent across entries. Unlike classical panel-factor methods, it does not require a prescribed global geometry such as staggered adoption, tall-wide panels, or rigid block sparsity. Instead, identification is local: a target entry can be recovered whenever the observed mask contains a sufficiently informative anchor structure around that entry.

To implement this principle, we introduce \SNN, an estimator that combines the locality of nearest-neighbor methods with the reweighting logic of synthetic controls. We establish finite-sample entrywise guarantees for mean recovery, derive asymptotic normality under heteroskedastic noise, and develop feasible studentized inference. We also extend the same local-anchor logic to heteroskedastic variance estimation, proving consistency under bounded noise and asymptotic unbiasedness under general subgaussian noise. The simulation evidence is consistent with these theoretical results and illustrates the robustness of \SNN~across missingness mechanisms and observation patterns. 

Looking ahead, \SNN’s position between classical matrix completion and panel-factor methods suggests two related directions. First, a sharper comparison across MCAR sampling, local MNAR designs, and structured panel geometries would clarify when local anchor methods are preferable to procedures tailored to global random sampling or panel-specific observation patterns. Second, it would be valuable to develop hybrid procedures that exploit global structure when it is reliable while retaining the local anchor logic that makes \SNN~robust to irregular MNAR patterns. For example, global factor estimates could stabilize anchor discovery and spectral denoising, whereas local anchor crosses could deliver target-specific recovery and inference in regions where a global observation model is weak or misspecified. Such hybrids may provide a more adaptive bridge between classical matrix completion, panel-factor methods, and the local MNAR framework developed here.

\bibliographystyle{alpha}
\bibliography{bib}

\clearpage
\appendix
\renewcommand{\thepage}{S\arabic{page}}
\renewcommand{\theequation}{S\arabic{equation}}
\renewcommand{\theassumption}{S\arabic{assumption}}
\renewcommand{\thelemma}{S\arabic{lemma}}
\setcounter{equation}{0}
\setcounter{section}{0}
\setcounter{assumption}{0}
\setcounter{lemma}{0}
\setcounter{page}{1}

\makeatletter
\def\@seccntformat#1{\@ifundefined{#1@cntformat}%
   {\csname the#1\endcsname\space}%
   {\csname #1@cntformat\endcsname}}
\makeatother
\renewcommand{\thesection}{S\arabic{section}}


\begin{center}
    \LARGE \scshape Appendix
\end{center}
\vspace{12pt}

%
Appendix~\ref{sec:teaser.details} details the simulation setup of Section~\ref{sec:teaser},
Appendix~\ref{sec:ensemble} presents an ensemble variant of \SNN,
Appendix~\ref{sec:max_biclique} describes how anchor sets can be identified by reducing the problem to maximal-biclique discovery in a bipartite graph,
and Appendix~\ref{sec:add.params} discusses related causal estimands.
The remainder of the supplementary material establishes the paper’s formal results.

\section{Details of Section~\ref{sec:teaser}} \label{sec:teaser.details} 
This section specifies the simulation designs used in Section~\ref{sec:teaser}. 
We focus here on the two MNAR mechanisms: one that preserves the conventional assumptions of positivity and independent sampling, and one that deliberately violates both. 
To isolate the role of the missingness mechanism, we use a noiseless design and set $\bE = \bzero$. 
For clarity, we frame our simulations in the context of a recommender system, where rows represent users and columns represent movies. 

\subsection{Conventional MNAR: Preserving Positivity and Independence} \label{sec:conventional}
This design captures self-selection in ratings: users are more likely to rate movies that they strongly liked or disliked and less likely to rate movies toward which they were indifferent. 
To create a dense source of information while retaining value-dependent selection, we introduce two distinguished groups: 
(i) {\em Core users} represent movie enthusiasts or professional critics who rate many movies; 
(ii) {\em Core movies} represent widely viewed, culturally prominent films (e.g., {\em Star Wars}) or the small set of movies that a streaming platform may present during onboarding to elicit a new user’s preferences. The resulting observation pattern preserves both positivity and independent sampling.

We fix $m = 100$, $n=100$, and $r = 5$. 
The true ratings matrix is $\bA = \bU \bV^\top$, where $\bU \in \Rb^{m \times r}$ and $\bV \in \Rb^{n \times r}$ contain latent user and movie factors. 
To make the recommender-system interpretation transparent, users and movies are first assigned balanced latent genre labels in $[r]$; the corresponding coordinate of each factor vector receives a common positive shift, with smaller independent Gaussian perturbations in all coordinates. 
The entries of $\bU \bV^\top$ are then affinely scaled to lie in the interval $[1,5]$, reflecting standard rating systems. 

Let $m_{\core}$ and $n_{\core}$ denote the number of core users and core movies. 
We partition the entries into four cohorts: 
$\mathcal{C}_{\core} \coloneqq \{(i,j): i \le m_{\core}, j \le n_{\core}\}$ as the subset of core users and core movies; 
$\mathcal{C}_{\tuser} \coloneqq \{(i,j): i \le m_{\core}, j > n_{\core}\}$ as the subset of core users and standard movies;
$\mathcal{C}_{\tmovie} \coloneqq \{(i,j): i > m_{\core}, j \le n_{\core}\}$ as the subset of standard users and core movies; and
$\mathcal{C}_{\standard} \coloneqq \{(i,j): i > m_{\core}, j > n_{\core}\}$ as the subset of standard users and standard movies. 
In the simulations, we set $m_{\core}=n_{\core}=25$. 

For each entry $(i,j)$, the observation probability $p_{ij}$ is larger when the underlying rating is more extreme. 
Specifically, define the rating-extremeness score
$
s_{ij} \coloneqq 0.15 + |A_{ij}-3|/2.
$
Ratings farther from the neutral midpoint of three therefore receive larger scores. 
Within each cohort $\mathcal{C}$, we set
$
p_{ij} \coloneqq \min\{\max\{\kappa_{\mathcal{C}} s_{ij}, 0.01\}, 0.98\},
$
where the cohort-specific constant $\kappa_{\mathcal{C}}$ is chosen so that the average observation probability within that cohort matches a target rate. 
The target rates are $98\%$ for $\mathcal{C}_{\core}$, $75\%$ for both $\mathcal{C}_{\tuser}$ and $\mathcal{C}_{\tmovie}$, and for $2\%$ for $\mathcal{C}_{\standard}$. These choices produce an overall observation rate of approximately $35\%$.

Conditional on the resulting propensity matrix $\bP = [p_{ij}]$, the observation indicators are drawn independently: $D_{ij} \sim \texttt{Bernoulli}(p_{ij})$. Because every $p_{ij}$ is bounded below by $0.01$, the design satisfies positivity. Moreover, because the indicators are sampled independently conditional on $\bP$, it also satisfies the conventional independence assumption. We generate $\bA$ and $\bP$ once and then conduct 10 independent trials. Across trials, the only source of randomness is the realization of the observation mask $\bD$.

\subsection{Panel-Based MNAR: Violating Positivity and Independence}
\label{sec:panel.mnar}
We next consider a more extreme selection mechanism that explicitly violates the positivity and independence assumptions commonly imposed in MNAR matrix completion. As in the preceding design, a set of core movies is rated by most users. Outside this core, however, a user rates a movie only when it belongs to that user’s favorite genre.
This rule creates structured regions in which the observation probability is exactly zero. It therefore violates positivity. Moreover, many observation indicators are jointly determined by the same user and movie genre labels, producing strong dependence across entries. 

We generate $\bA \in \Rb^{100 \times 100}$ using the same latent-factor construction as in Section~\ref{sec:conventional}, with $r = \rank(\bA) = 5$. 
The five latent dimensions are interpreted as movie genres. In this view, $U_{ik}$ quantifies user $i$'s affinity for genre $k$, and $V_{jk}$ indicates the extent to which movie $j$ represents genre $k$.
We use $n_{\core}=21$ core movies. This value, combined with the genre-matching rule described below, produces an overall observation rate close to 35\%. In each replication, the identities of the core movies are selected and, for notational convenience, indexed first. Thus, core movies satisfy $j \le n_{\core}$. 

Every core movie is observed independently by each user: $D_{ij} \sim \texttt{Bernoulli}(0.93)$ for $i \in [m]$ and $j \le n_{\core}$. This dense core provides anchor information for \SNN~while retaining enough missingness to make the recovery problem nontrivial.
For the remaining movies, define user $i$'s favorite genre as $k^{\flat}(i) = \argmax_{k \in [r]} U_{ik}$, 
and assign movie $j$ to the genre $k^\sharp(j) = \arg\max_{k \in [r]} V_{jk}$. 
For each non-core movie $j > n_{\core}$, the observation indicator is $D_{ij} = \boldsymbol{1} \big\{k^\flat(i) = k^\sharp(j) \big\}$. 
Thus, a user always rates a non-core movie from the user’s favorite genre and never rates a non-core movie from another genre. Entries associated with mismatched genres consequently have zero probability of observation, while entries associated with matched genres are deterministically observed. 

We generate $\bA$ once and conduct 10 trials. In each trial, we resample the set of core movies and the Bernoulli observations within the core block. Conditional on the fixed user and movie genre labels and the selected core set, the observations outside the core follow the deterministic genre-matching rule.

\subsection{Estimator Implementation and Evaluation}
Across the three missingness mechanisms, we compare \SNN~with modified \USVT~and de-biased \softimpute.
The modified \USVT~estimator first estimates the observation-propensity matrix by applying \USVT~to the mask and then rescales the de-noised observed-rating matrix by these estimated propensities.
The de-biased \softimpute~estimator follows the inverse-propensity weighting approach of \cite{ma2019missing}; its regularization parameter is selected by 5-fold cross-validation. 
For \SNN, anchor blocks are found by spectral biclustering, with a maximal-biclique search (see Appendix~\ref{sec:max_biclique}) used whenever the spectral block is insufficient.
The SNN spectral rank is selected by five-fold cross-validation. 
All estimators are clipped to the rating scale $[1,5]$.
We report RMSE over the missing entries and summarize variability over the 10 independently sampled masks.


\section{Ensemble \SNN} \label{sec:ensemble}
The main theory considers a single valid anchor cross for each target entry. In practice, the observation mask may admit several such crosses. Combining them can reduce sensitivity to the particular anchor set selected, although adding weak or highly redundant blocks need not improve performance. This section describes an equal-weight ensemble, gives a practical rule for choosing the number of included blocks, and briefly records the resulting statistical implications. 

\subsection{Algorithmic Extension}
Let $K_{\max} \ge 1$ denote the number of candidate anchor blocks retained after the feasibility screening described below. 
See Figure~\ref{fig:snn} for a visualization of this extension. 

\begin{enumerate} 

	\item {\em Anchor set discovery}: Obtain candidate pairs $\{(\nAR_k, \nAC_k): k \in [K_{\max}]\}$ such that each pair induces a local array in which every entry other than the target $(i,j)$ is observed.
	
	\item {\em Local estimation}: For each $k$, apply \SNN~to the corresponding local array, allowing the spectral threshold $t_k$ to vary across blocks, and obtain the local estimate $\hA^{(k)}_{ij}$. 
		
	\item {\em Global estimation}: For $1 \le K \le K_{\max}$, define the $K$-block ensemble by $\hA^{[K]}_{ij} \coloneqq \frac{1}{K} \sum_{k=1}^K \hA^{(k)}_{ij}$. 

\end{enumerate} 
The candidate blocks need not be disjoint and may have different dimensions or effective ranks. The theory continues to require that every included block satisfy the relevant row- and column-span conditions and contain a sufficiently strong low-rank signal. Because these population conditions cannot be verified directly, the empirical screening and cross-validation procedure below should be viewed as diagnostics rather than formal tests of the assumptions.

\begin{figure} [!t]
	\centering 
		\includegraphics[width=0.9\linewidth]
		{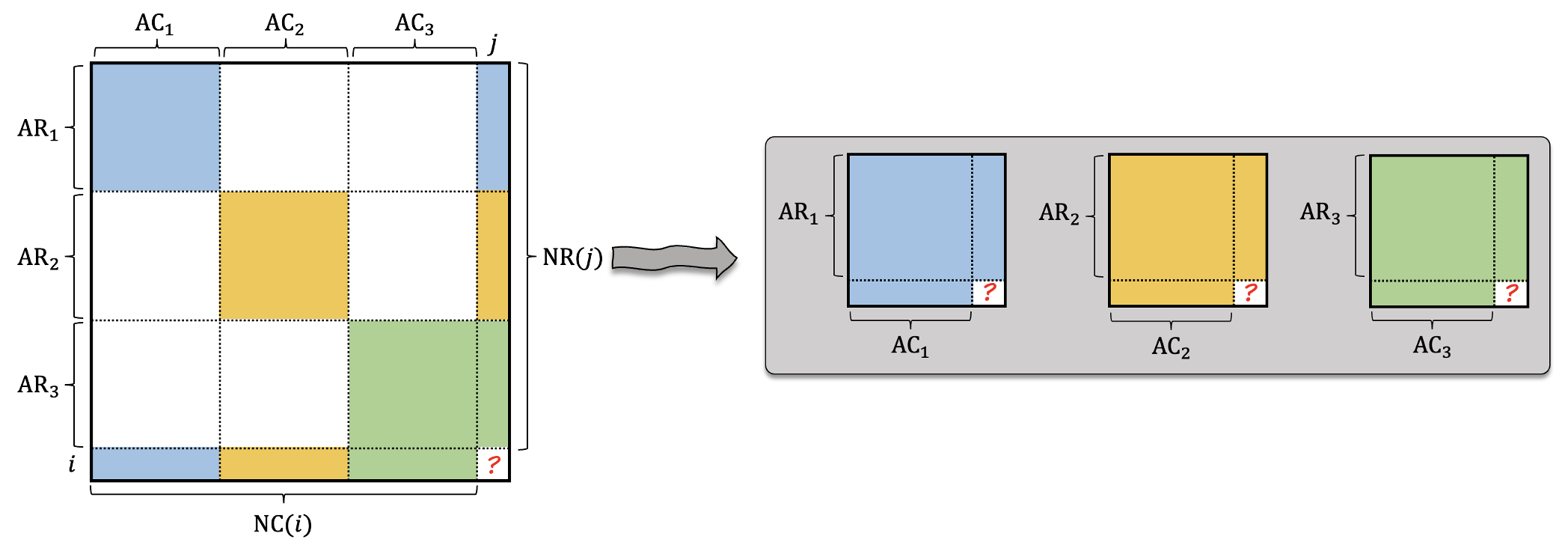}
	\caption{Visualization of the observed matrix $\bY$, highlighting the neighborhood sets $(\NR(j), \NC(i))$, and anchor sets $\{(\nAR_k, \nAC_k): k\in [3]\}$~associated with a specific $(i,j)$ pair. Each matching anchor pair is shown in a distinct color. Unobserved entries are depicted in white.}
	\label{fig:snn} 
\end{figure}

\subsection{Choosing \texorpdfstring{$K$}{K} in Practice}
To ensure reliable local estimation, each anchor set must be sufficiently rich in structure. 
Thus, more blocks are not automatically more preferable. A small or ill-conditioned block may introduce substantial estimation error, while two nearly identical blocks contribute highly correlated estimates and therefore little additional information. We recommend the following target-specific procedure.


\medskip \noindent \textbf{Feasibility screening.} 
We recommend first screening candidate blocks for feasibility. Each retained block should have dimensions comfortably larger than its selected rank, singular values separated from the estimated noise floor, and numerically stable \PCR~weights. For inference, one may additionally exclude blocks with highly concentrated influence weights or unstable variance estimates. Near-duplicate blocks can be down-ranked because they contribute little independent information.
The remaining blocks may be ordered using a fixed score favoring larger blocks, stronger spectral separation, and lower overlap. We then recommend pseudo-target cross-validation outlined next. 

\medskip \noindent \textbf{Pseudo-target cross-validation.} 
Construct a set $\Hc$ of observed entries that can be temporarily masked while leaving enough observations to form valid anchor crosses. To make the validation problem resemble the actual target, priority can be given to entries in the same row or column, or to entries with similar neighborhood sizes and anchor geometry. For each validation fold, treat its entries as missing, repeat the anchor discovery and estimation procedure without using their outcomes, and compute the corresponding ensemble predictions. For a candidate value of $K$, define 
\begin{align}
	\widehat{\texttt{CV}}(K) \coloneqq \frac{1}{|\Hc_K|} \sum_{h=(a,b) \in \Hc_K} \left(Y_{ab} - \hA_{ab}^{[K, -h]} \right)^2,
\end{align}
where $\hA_{ab}^{[K, -h]}$ is obtained with $Y_{ab}$ held out, and $\Hc_K$ contains the pseudo-targets for which at least $K$ admissible blocks remain. When more than $K$ blocks are available, an inner validation fold can order them by their individual predictive performance, with ties favoring stronger spectral separation and lower overlap. The outer fold then evaluates the ensemble size.

A parsimonious rule is to select the smallest $K$ whose cross-validation loss is within one standard error of the minimum. This avoids adding blocks whose apparent incremental improvement is indistinguishable from validation noise. The resulting $K$ may vary across target entries.

Because the outcomes are MNAR, validation on observed entries is not an unbiased estimate of prediction risk at a missing target. The procedure is therefore a heuristic. Matching pseudo-targets to the target’s row, column, and local observation geometry helps make the comparison more informative.

\medskip \noindent \textbf{Point estimation versus inference.} 
For prediction, let $K_{\text{est}}$ be selected by the preceding rule from all blocks that pass the point-estimation screen. For inference, it may be preferable to choose a possibly smaller $K_{\text{inf}}$ from the stricter inferentially admissible pool. To avoid unaccounted post-selection effects, $K_{\text{inf}}$ should ideally be pre-specified or chosen using tuning outcomes that are not reused in the final inferential estimator. A mask-only rule based on anchor dimensions and observation geometry is another conservative option.
When $K_{\text{est}} \neq K_{\text{inf}}$, the confidence interval must be centered at $\hA^{[K_{\text{inf}}]}_{ij}$ rather than at the point-prediction estimator $\hA^{[K_{\text{est}}]}_{ij}$. The two estimates may be reported separately, with the latter used for prediction and the former used for inference.

\subsection{Statistical Implications} 
For each included block, write the expansion from Theorem~\ref{thm:normality} as $\hA^{(k)}_{ij} - A_{ij} = G_k + R_k$. The ensemble then satisfies 
\begin{align}
	\hA^{[K]}_{ij} - A_{ij} = \bar{G}_K + \bar{R}_K,
	\quad
	\bar{G}_K \coloneqq \frac{1}{K} \sum_{k=1}^K G_k,
	\quad
	\bar{R}_K \coloneqq \frac{1}{K} \sum_{k=1}^K R_k. 
\end{align}
Let $\varphi_{k, ab}$ denote the influence coefficient of observation $(a,b)$ in $G_k$, defined as in \eqref{eq:influence} of Appendix~\ref{sec:add.params}. Conditional on $\Ec$, the leading variance is 
\begin{align}
	\upsilon^2_K = \frac{1}{K^2} \sum_{k=1}^K \sum_{\ell=1}^K \Cov(G_k, G_\ell \mid \Ec)
	= \sum_{a=1}^m \sum_{b=1}^n \left(\frac{1}{K} \sum_{k=1}^K \varphi_{k, ab} \right)^2 \sigma^2_{ab}. 
\end{align}
If the local noise supports are disjoint and the blockwise variances are comparable, averaging yields the familiar variance reduction of approximately $1/K$. With overlapping anchor crosses, however, shared observations generate covariance terms, and the effective number of independent blocks may be much smaller than $K$. Consequently, the variance should not be estimated by simply summing the entrywise variances as though the local estimates were independent. The corresponding plug-in estimator instead aggregates the estimated influence coefficients before squaring: 
\begin{align}
	\hupsilon^2_K \coloneqq \sum_{a=1}^m \sum_{b=1}^n \left(\frac{1}{K} \sum_{k=1}^K \hvarphi_{k, ab} \right)^2 \hsigma^2_{ab}. 
\end{align}
For fixed $K$, asymptotic normality is expected under the blockwise conditions of Theorem~\ref{thm:normality}, together with $\bar{R}_K = o_p(\upsilon_K)$ and an aggregate no-dominance condition ensuring that no single observation controls $\bar{G}_K$. If $K$ grows with the local dimensions, uniform control over the blockwise remainders and influence weights is additionally required. Equal weighting is a transparent default; covariance-aware weights could in principle improve efficiency, but would require reliably estimating the full cross-block covariance structure.

\section{Finding Anchor Sets via Maximal Biclique Search}  \label{sec:max_biclique} 
We recast the problem of identifying anchor sets as an instance of the classical graph-theoretic problem of finding maximal bicliques. 
To describe this reduction, we first describe the graph-theoretic framework. 

\medskip \noindent
\textbf{Graph theory background.}
Let $\Gc = (\Vc_1, \Vc_2, \Ec)$ denote a bipartite graph, where $(\Vc_1, \Vc_2)$ are the disjoint vertex sets, and $\Ec \in \Vc_1 \times \Vc_2$ is the edge set, i.e., $(v_1, v_2) \in \Ec$ if there is an edge between $v_1 \in \Vc_1$ and $v_2 \in \Vc_2$. 
Equivalently, $\Gc$ can be represented via a bipartite incidence (or adjacency) matrix $\bW \in \{0, 1\}^{|\Vc_1| \times |\Vc_2|}$, where $W_{ij} = 1$ if and only if $(v_i, v_j) \in \Ec$. 
A {\em biclique} is a complete bipartite subgraph $\Gc$, meaning there is an edge between every pair of nodes in the two vertex subsets. 
We denote such a subgraph as $\BCc \subseteq \Gc$. 

\medskip \noindent
\textbf{Reduction.} 
In our setting, the observation mask $\bD \in \{0,1\}^{m \times n}$ induces a bipartite graph, where rows and columns of $\bD$ correspond to vertex sets $\Vc_1$ and $\Vc_2$. 
We define the edge set $\Ec$ such that $(v_i, v_j) \in \Ec$ if and only if $D_{ij} = 1$. 
Hence, the incidence matrix $\bW$ of this graph is exactly equal to $\bD$. 

\medskip \noindent
\textbf{Algorithmic implementation.}
Given this reduction, we describe the practical implementation of the \anchor~algorithm stated in Algorithm~\ref{alg:anchors}. 
For the purposes of generality, we consider the ensemble \SNN~discussed in Appendix~\ref{sec:ensemble}. 
We assume access to two subroutines: 
\begin{enumerate} [label=(\alph*)]

\item $\texttt{createGraph}$: Converts an incidence matrix $\bW$ into a bipartite graph $\Gc$. This can be implemented using, for example, the Python package \texttt{NetworkX}. 

\item $\texttt{maxBiclique}$: Given a bipartite graph $\Gc$, returns a collection of $K$ {\em maximal} bicliques $\{ \BCc^{(k)}: k \in [K] \}$.

\end{enumerate} 
Each maximal biclique $\BCc^{(k)} = (\Vc^{(k)}_1, \Vc^{(k)}_2, \Ec^{(k)})$ defines a pair of anchor rows and columns: $(\AR_k, \AC_k) \leftarrow (\Vc^{(k)}_1, \Vc^{(k)}_2)$. 
At this juncture, it is crucial to distinguish between {\em maximal} and {\em maximum} bicliques, as they differ significantly in both definition and computational complexity for discovery. 

\medskip \noindent
\textbf{Maximal versus maximum bicliques.}
A maximal biclique is a complete bipartite subgraph that cannot be extended by including any additional vertex from either partition without losing its completeness. 
By contrast, a maximum biclique is the largest such subgraph according to a chosen metric---typically, the number of edges, total nodes, or size of one partition. 
While a graph may contain many maximal bicliques, there is usually only one (or a few) maximum biclique(s). 

Finding a single maximal biclique is computationally straightforward and can be achieved in polynomial time via greedy strategies. 
However, enumerating {\em all} maximal bicliques can be computationally expensive, as their number can grow exponentially in the worst case. 
Despite this, efficient enumeration algorithms perform well in practice on moderately sized or sparse graphs. 
On the other hand, the problem of identifying a maximum biclique is NP-hard, and even approximate solutions are difficult to compute. 
This computational gap makes maximal bicliques particularly attractive in applications that require structured submatrices but can tolerate suboptimal size. 

Consequently, many biclique-based methods, including Algorithm~\ref{alg:anchors}, rely on maximal rather than maximum bicliques for scalability and feasibility. 
For further discussion on biclique enumeration algorithms, we refer the interested reader to \cite{biclique2} and references therein. 


\begin{algorithm} [t!]
\caption{AnchorSubMatrix$(i,j)$}
\label{alg:anchors}
\begin{algorithmic}
\Require \texttt{createGraph}, \texttt{maxBiclique}
\Statex Find $\NR(j)$ and $\NR(i)$
\Statex Assign $\bW \gets [D_{\ell q} : (\ell,q)\in \NR(j)\times \NR(i)]$
\Statex Generate $\Gc \gets \texttt{createGraph}(\bW)$
\Statex Compute $\{\BCc^{(k)} = (\Vc_1^{(k)},\Vc_2^{(k)},\Ec^{(k)})\}_{k\in[K]} \gets \texttt{maxBiclique}(\Gc)$
\Ensure $\{(\AR_k, \AC_k) : k \in [K]\}$, where $(\AR_k,\AC_k)\gets (\Vc_1^{(k)}, \Vc_2^{(k)})$
\end{algorithmic}
\end{algorithm}

\section{Related Causal Estimands} \label{sec:add.params}
The main text focuses on entrywise recovery of $A_{ij}$, but the same asymptotic linear representation also accommodates pre-specified weighted averages and, more generally, linear functionals of the mean matrix. 
Let $\Tc \subseteq [m] \times [n]$ be a collection of target entries and let $\bomega = [\omega_{ij}]$ be deterministic weights. For an average, the weights are typically nonnegative and sum to one, although signed weights may also be used to form contrasts. Define
\begin{align}
	\theta_{\bomega} \coloneqq \sum_{(i,j) \in \Tc} \omega_{ij} A_{ij},
	\quad
	\htheta_{\bomega} \coloneqq \sum_{(i,j) \in \Tc} \omega_{ij} \hA_{ij},
\end{align} 
where each $\hA_{ij}$ is computed from a valid local anchor cross. Entrywise inference is recovered as the special case in which $\Tc = \{(i,j)\}$ and $\omega_{ij} = 1$. 

To describe inference for $\theta_{\bomega}$, index a target entry by $\tau = (i,j)$, and write $(\AR_\tau, \AC_\tau)$, $\bbeta^*_\tau$, and $\balpha^*_\tau$ for its local anchor sets and population parameters. Theorem~\ref{thm:normality} gives $\hA_\tau - A_\tau = G_\tau + R_\tau$, 
where the leading term can be written as 
\begin{align}
	G_\tau = \sum_{a=1}^m \sum_{b=1}^n \varphi_{\tau, ab} \cdot \varepsilon_{ab},
\end{align}
with influence coefficients
\begin{align}
	\varphi_{\tau, ab} \coloneqq \begin{cases}
		\beta^*_{\tau, a}, & (a,b) \in \AR_\tau \times \{j\}, 
		\\
		\alpha^*_{\tau, b}, & (a,b) \in \{i\} \times \AC_\tau,
		\\
		- \beta^*_{\tau, a} \alpha^*_{\tau, b}, & (a,b) \in \AR_\tau \times \AC_\tau,
		\\
		0, & \text{otherwise}. 
	\end{cases}
	\label{eq:influence}
\end{align}
Consequently, 
\begin{align}
	\htheta_{\bomega} - \theta_{\bomega} = \sum_{a=1}^m \sum_{b=1}^n \eta_{ab} \cdot \varepsilon_{ab} + R_{\bomega},
\end{align}
where
\begin{align}
	\eta_{ab} \coloneqq \sum_{\tau \in \Tc} \omega_\tau \varphi_{\tau, ab}, 
	\quad
	R_{\bomega} \coloneqq \sum_{\tau \in \Tc} \omega_\tau R_\tau. 
\end{align}
Under Assumption~\ref{assump:subg}, the conditional variance of the leading term is therefore
\begin{align}
	\upsilon^2_{\bomega} \coloneqq \sum_{a=1}^m \sum_{b=1}^n \eta^2_{ab} \cdot \sigma^2_{ab}. 
\end{align}
This expression automatically accounts for dependence among the \SNN~estimates caused by shared anchor rows, anchor columns, or anchor-block observations: the influence coefficients must be aggregated before they are squared. If the noise supports of the local anchor crosses are disjoint, the covariance terms vanish and 
\begin{align}
	\upsilon^2_{\bomega} = \sum_{\tau \in \Tc} \omega^2_\tau \upsilon^2_\tau,
\end{align}
where $\upsilon^2_\tau$ is the entrywise variance in Theorem~\ref{thm:normality}. With overlapping crosses, the variance generally also contains nonzero cross-covariance terms.

Under Gaussian noise, the leading term is conditionally Gaussian. Under general subgaussian noise, an aggregate analogue of condition \eqref{eq:clt}, together with a negligible weighted remainder, is sufficient:
\begin{align}
	\max_{a,b} \frac{|\eta_{ab}| \sigma_{ab}}{\upsilon_{\bomega}} = o(1),
	\quad
	\frac{R_{\bomega}}{\upsilon_{\bomega}} \xrightarrow{p} 0. 
\end{align}
Under these conditions,
\begin{align}
	\frac{\htheta_{\bomega} - \theta_{\bomega}}{\upsilon_{\bomega}} \rightsquigarrow \Nc(0,1). 
\end{align}
A natural feasible variance estimator aggregates the estimated influence coefficients in the same manner: 
\begin{align}
	\hupsilon^2_{\bomega} \coloneqq \sum_{a=1}^m \sum_{b=1}^n \heta_{ab} \hsigma^2_{ab},
	\quad
	\heta_{ab} \coloneqq \sum_{\tau \in \Tc} \omega_\tau \hvarphi_{\tau, ab},
\end{align}
where $\hvarphi_{\tau, ab}$ replaces the population representers by their \PCR~estimates and $\hsigma^2_{ab}$ is obtained as in Section~\ref{sec:var.est}. Whenever $\hupsilon_{\bomega} / \upsilon_{\bomega} \xrightarrow{p} 1$, a corresponding $(1-\gamma)\times100\%$ confidence interval is 
\begin{align}
	\left[ \htheta_{\bomega} \pm z_{1-\gamma/2} \hupsilon_{\bomega} \right]. 
\end{align}
Thus, entrywise and weighted-average inference can be expressed within the same linear-functional framework. A complete theorem for a growing collection $\Tc$ would additionally require uniform control of the entrywise remainders, representer estimates, and local variance estimates; we leave those technical extensions for future work.

\section{Notation for Proofs} \label{sec:proofs.notation}
In this section, we establish the notation to be used throughout our proofs. 
Let
\begin{align}
\by \coloneqq \bY_{i, \nAC},
\quad
\bX \coloneqq \bY_{\nAR, \nAC},
\quad
\bz \coloneqq \bY_{\nAR, j}. 
\end{align}
Denote $\bbW \coloneqq \Ex[\bW \mid \Ec]$ for any random object $\bW$. 
With this, we can define the population minimum $\ell_2$-norm coefficients as 
\begin{align}
	\bbeta^* \coloneqq \bbX^{\top, \dagger} \bby,
	\quad
	\balpha^* \coloneqq \bbX^\dagger \bbz.  \label{eq:pop.params} 
\end{align}
Moreover, we define the noise blocks as 
\begin{align}
	\bxi_y \coloneqq \by - \bby,
	\quad
	\bXi_X \coloneqq \bX - \bbX,
	\quad
	\bxi_z \coloneqq \bz - \bbz. 
\end{align}
Further, denote the de-noised design matrix as 
\begin{align}
	\hbX \coloneqq \bX^{(\br)} \coloneqq \bY^{(\br)}_{\nAR, \nAC},
\end{align}
where $\br \coloneqq \rank(\bbX)$. 
In turn, we denote the parameter estimates as 
\begin{align}
	\hbbeta \coloneqq \hbX^{\top, \dagger} \by,
	\quad
	\hbalpha \coloneqq \hbX^\dagger \bz,
\end{align}
and the corresponding recovery errors as 
\begin{align}
	\bDelta_\beta \coloneqq \hbbeta - \bbeta^*,
	\quad
	\bDelta_\alpha \coloneqq \hbalpha - \balpha^*. 
\end{align}
Finally, recall $|\nARC| \coloneqq |\nAR| \cdot |\nAC|$ and $L \coloneqq \log(|\nARC|)$.

\section{Known Auxiliary Results} \label{sec:known.results}
This section collects several standard results that are used repeatedly throughout the proofs.

The first lemma is Weyl’s inequality, which bounds the perturbation of each singular value by the operator norm of the corresponding matrix perturbation. 

\begin{lemma} \label{lemma:weyl} 
Given $\bA, \bB \in \Rb^{m \times n}$, let $\sigma_i$ and $\widehat{\sigma}_i$ be the $i$-th singular values of $\bA$ and $\bB$, respectively, in decreasing order and repeated by multiplicities. 
Then for all $i \le \min\{m, n\}$, $\abs{ \sigma_i - \widehat{\sigma}_i} \le \|\bA - \bB \|_{\eop}.$
\end{lemma} 

The next lemma controls the perturbation of a leading right-singular subspace, with the bound determined by the associated spectral gap.

\begin{lemma} {\cite[]{Wedin1972PerturbationBI}} \label{thm:wedin}
Given $\bA, \bB \in \Rb^{m \times n}$, let $\bV, \bhV \in \Rb^{n \times n}$ denote their respective right singular vectors.
Further, let $\bV_{k} \in \Rb^{n \times k}$ (respectively, $\bhV_{k} \in \Rb^{n \times k}$) correspond to the truncation of $\bV$ (respectively, $\bhV$), respectively, that retains the columns corresponding to the top $k$ singular values of $\bA$ (respectively, $\bB$).  
Let $s_{i}$ represent the $i$-th singular value of $\bA$.
Then, 
\begin{align}
	 \left\| \bhV_k \bhV_k^\top - \bV_k \bV_k^\top  \right\|_{\eop}  \leq \frac{2}{s_{k} - s_{k+1}} \cdot \left\|\bA - \bB \right\|_{\eop}.
\end{align}
\end{lemma}

The following result provides a high-probability operator-norm bound for a matrix with independent, mean-zero subgaussian entries.

\begin{lemma} {\cite[Theorem 4.4.5]{vershynin2018high}} \label{lemma:subg_matrix}
Let $\bA = [A_{ij}]$ be an $m \times n$ random matrix where the entries $A_{ij}$ are independent, mean zero, subgaussian random variables. 
Then for any $t > 0$, $$\| \bA \|_{\eop} \le CK \left(\sqrt{m} + \sqrt{n} + t \right)$$ with probability at least $1-2\exp(-t^2)$. 
Here, $K = \max_{i,j} \|A_{ij}\|_{\psi_2}$ and $C>0$ is an absolute constant. 
\end{lemma}  

The next lemma is a Hoeffding-type inequality for weighted sums of independent, mean-zero subgaussian random variables.
\begin{lemma} {\cite[Theorem 2.6.3]{vershynin2018high}} \label{lemma:hoeffding}
Let $X_1, \cdots, X_N$ be independent, mean-zero, subgaussian random variables, and $\ba = [a_i] \in \Rb^N$. Then, for every $t \ge 0$, we have
\begin{align}
	\Pb \left( \left|\sum_{i=1}^N a_i X_i \right| \ge t \right) \le 2 \cdot \exp(-\frac{ct^2}{K^2 \| \ba \|_2^2}),
\end{align}
where $K = \max_i \| X_i \|_{\psi_2}$. 
\end{lemma}

The Hanson–Wright inequality below controls the deviation of a subgaussian quadratic form from its mean.
\begin{lemma}{\cite[Theorem 6.6.2]{vershynin2018high}} \label{lemma:hanson.wright}
Let $\bA$ be an $n \times n$ matrix. Let $X = (X_1, \dots, X_n) \in \Rb^n$ be a random vector with independent, mean-zero, subgaussian coordinates. Then, for every $t \ge 0$, we have
\begin{align}
	\Pb\left( \left| X^\top \bA X - \Ex X^\top \bA X \right| \ge t \right) \le 2 \exp(-c\min\left\{ \frac{t^2}{K^4 \| \bA \|_F^2}, \frac{t}{K^2 \| \bA \|_\txtop} \right\}),
\end{align}
where $K = \max_i \| X_i \|_{\psi_2}$. 
\end{lemma}

The next lemma extends operator-norm concentration to matrices with independent, mean-zero sub-exponential entries.

\begin{lemma}{\cite[Theorem H.1]{agarwal2019robustness}} \label{lemma:dg}
Let $\bA = [A_{ij}]$ be an $m \times n$ random matrix where the entries $A_{ij}$ are independent, mean zero, sub-exponential random variables with $K = \max_{i,j} \|A_{ij}\|_{\psi_1}$. 
Then for any $\delta > 0$, with probability at least $1-\frac{2}{n^{1+\delta} m^\delta}$, 
\begin{align}
	\| \bA  \|_{\eop} \lesssim \sqrt{m} + K \sqrt{n (1+\delta) \log(mn)} \cdot \left(1 + (2+\delta)\log(mn) \right). 
\end{align}
\end{lemma}

The final lemma is Janson’s inequality, which bounds the probability that none of a finite family of potentially overlapping subsets is contained in an independently sampled random set. 

\begin{lemma} {\cite[Theorem 2.18(i)]{janson2000random}} \label{lemma:janson}
Let $\Omega$ be a finite set, and let $\Rc \subseteq \Omega$ be formed by including each $x \in \Omega$ independently with probability $p \in [0,1]$. Let $\{\Sc_\alpha: \alpha \in \Ic\}$ be a finite family of nonempty subsets of $\Omega$. For each $\alpha \in \Ic$, define $I_\alpha \coloneqq \boldsymbol{1}\{\Sc_\alpha \subseteq \Rc\}$, $X \coloneqq \sum_{\alpha \in \Ic} I_\alpha$, $\mu \coloneqq \Ex[X]$. Write $\alpha \sim \beta$ if $\alpha \neq \beta$ and $\Sc_\alpha \cap \Sc_\beta \neq \emptyset$, and define $\Delta \coloneqq \sum_{\alpha \sim \beta} \Ex[I_\alpha I_\beta]$, where the sum is over ordered pairs. Then 
\begin{align}
	\Pb\left(X = 0 \right) \le \exp(-\mu + \frac{\Delta}{2}). 
\end{align}
\end{lemma}

\section{Proof of Proposition~\ref{thm:identification}} \label{sec:identification} 
\begin{proof}[Proof of Proposition~\ref{thm:identification}] 
By Assumption~\ref{assump:linear_span}, there exists a coefficients vector $\bbeta \in \Rb^{|\AR|}$ such that
\begin{align}
	\bu_i &= \sum_{\ell \in \AR} \beta_\ell \cdot \bu_\ell, \label{eq:a3.1}
\end{align}
conditional on $\cE$. 
Hence, we have that  
\begin{align}
A_{ij}  & = \Ex\left [\left \langle  \bu_i,  \bv_j \right \rangle + \varepsilon_{ij} \mid \bu_i, \bv_j \right] &&\because \text{Assumption~\ref{assump:lfm}}
\\&=  \left \langle \bu_i,  \bv_j \right \rangle \mid \{\bu_i, \bv_j\} &&\because \text{Assumption~\ref{assump:mean_ind}} 
\\&=  \left \langle \bu_i,  \bv_j \right \rangle \mid \cE
\\&= \sum_{\ell \in \AR} \beta_\ell \cdot \left \langle  \bu_\ell,  \bv_j \right\rangle \mid \cE &&\because \text{\eqref{eq:a3.1}} 
\\&= \sum_{\ell \in \AR} \beta_\ell \cdot \Ex \left[ \left \langle \bu_\ell, \bv_j \right \rangle + \varepsilon_{\ell j}  \mid \cE \right] &&\because \text{Assumption~\ref{assump:mean_ind}} 
\\ & = \sum_{\ell \in \AR} \beta_\ell \cdot \Ex [  \tY_{\ell j} \mid \cE ] &&\because \text{Assumption~\ref{assump:lfm}} 
\\& = \sum_{\ell \in \AR} \beta_\ell \cdot \Ex \left[  Y_{\ell j} \mid \cE \right]. &&\because \text{\eqref{eq:sutva}} 
\end{align}
We note that the third equality holds since $\langle \bu_i, \bv_j \rangle$ is deterministic conditional on $\{\bu_i, \bv_j\}$. 
This completes the proof for Proposition~\ref{thm:identification}(a). 
To establish Proposition~\ref{thm:identification}(b), we follow the same arguments laid out above.
Specifically, observe that for any $q \in \AC$, we have
\begin{align}
	\Ex\left[ Y_{iq} \mid \cE \right] 
	&= \Ex \left[ Y_{iq} \mid \cE \right] &&\because \text{\eqref{eq:sutva}}
	\\
	&= \Ex \left[ \left\langle \bu_i, \bv_q \right \rangle + \varepsilon_{iq} \mid \Ec \right]  && \because \text{Assumption~\ref{assump:lfm}} 
	\\
	&= \left \langle \bu_i, \bv_q \right \rangle \mid \cE && \because \text{Assumption~\ref{assump:mean_ind}}
	\\
	&= \sum_{\ell \in \AR} \beta_\ell \cdot \left \langle \bu_\ell, \bv_q \right \rangle \mid \cE && \because \text{\eqref{eq:a3.1}} 
	\\
	&= \sum_{\ell \in \AR} \beta_\ell \cdot \Ex \left[ \left \langle \bu_\ell, \bv_q \right \rangle + \varepsilon_{\ell q} \mid \cE \right] && \because \text{Assumption~\ref{assump:mean_ind}}
	\\
	&= \sum_{\ell \in \AR} \beta_\ell \cdot  \Ex \left[ \tY_{\ell q} \mid \cE \right] && \because \text{Assumption~\ref{assump:lfm}}
	\\
	&= \sum_{\ell \in \AR}  \beta_\ell \cdot \Ex \left[ Y_{\ell q} \mid \cE \right]. && \because \text{\eqref{eq:sutva}}
\end{align}
This completes the proof. 
\end{proof} 
\section{\PCR~Result} \label{sec:proofs.helper.lemmas} 
In this section, we bound the \PCR~estimation error in Lemma~\ref{lemma:param.est.1}. This result underpins the proofs for our primary results. We then state the primary lemmas that aid in establishing Lemma~\ref{lemma:param.est.1}, complete its proof, then provide proofs to the primary lemmas in Appendix~\ref{sec:pcr.lemmas.proofs}. 

\begin{lemma}\label{lemma:param.est.1}
%
%
Let the setup of Theorem~\ref{thm:consistency} hold. 
Then, conditional on $\cE$, with probability at least $1 - \mathcal{O}\left( |\nARC|^{-10} \right)$, 
\begin{align}
	\| \bDelta_\beta \|_2 \lesssim \Lambda_\beta \coloneqq \frac{\sigma \br} { \sqrt{|\nAR|} \cdot \min\left\{\sqrt{|\nAR|}, \sqrt{|\nAC|}\right\}} 
		+ \frac{\sqrt{\sigma \br} \cdot L^{1/4}}{\sqrt{|\nAR|} \cdot |\nAC|^{1/4}}.
		\label{eq:pcr.error}  
\end{align}
\end{lemma}

\subsection{Key Lemmas for Lemma~\ref{lemma:param.est.1}} \label{sec:pcr.key.lemmas}
We introduce notation that will be used throughout the proofs in this section. 
Denote the SVDs of  $\bbX$ and $\hbX$ as $\bbX = \bU \bS \bV^\top$ and $\hbX = \bhU \bhS \bhV^\top$. 
Denote $\bP = \bU \bU^\top$ and $\bhP = \bhU \bhU^\top$ as the projection matrices onto the column spaces of $\bbX$ and $\hbX$. 
Similarly, denote $\bQ = \bV \bV^\top$ and $\bhQ = \bhV \bhV^\top$ as the projection matrices onto the row spaces of $\bbX$ and $\hbX$. 
Let $s_\ell$ and $\hs_\ell$ denote the $\ell$-th singular values of $\bbX$ and $\hbX$.

\begin{lemma} \label{lemma:theta} 
Let Assumptions~\ref{assump:lfm}, \ref{assump:mean_ind}, \ref{assump:linear_span}, and \ref{assump:subspace} hold. 
Then, conditional on $\Ec$, $A_{ij} = \bbz^\top \bbeta^*$ and $\bby = \bbX^\top \bbeta^*$, 
where $\bbeta^*$ is defined in \eqref{eq:beta.star}. 
\end{lemma}

\begin{lemma} \label{lemma:singular_values}
Let Assumptions  \ref{assump:lfm}, \ref{assump:mean_ind}, and \ref{assump:subg} hold. 
Then, conditional on $\Ec$, for any $t > 0$ with probability at least $1-2\exp(-t^2)$, 
\begin{align}
	\left| s_{\br} - \hs_{\br} \right| \lesssim \sigma \left( \sqrt{\abs{\nAR}} + \sqrt{\abs{\nAC}} + t \right).
\end{align} 
\end{lemma}

\begin{lemma} \label{lemma:subspace} 
Let Assumptions  \ref{assump:lfm}, \ref{assump:mean_ind}, and \ref{assump:subg} hold. 
Then, conditional on $\Ec$, for any $t > 0$ with probability at least $1-2\exp(-t^2)$, 
\begin{align}
        \left \| \bhP - \bP \right \|_{\emph{op}}
        \lesssim \sigma s_{\br}^{-1} \left(\sqrt{\abs{\nAC}} + \sqrt{\abs{\nAR}} + t \right).
\end{align}
\end{lemma} 

\begin{lemma}  \label{lemma:annoying} 
Let Assumptions \ref{assump:lfm}--\ref{assump:spectra} hold. 
Define the sigma-field $\Hc \coloneqq \sigma(\Ec, \bX)$. 
Then, conditional on $\Hc$, with probability at least $1 - \mathcal{O}\left( |\nARC|^{-10} \right)$, 
\begin{align}
	\left \langle \hbX^\top \bDelta_\beta, \bxi_y \right \rangle  \lesssim \sigma^2 \br + \sigma^2 \left(\sqrt{\br L} + L \right) + \sigma \sqrt{L} \cdot \left\{ \sqrt{\abs{\nAC}} +  \left\| \bXi_X \right\|_{\eop} \cdot \left \| \bbeta^* \right\|_2 \right\},
\end{align} 
where $L \coloneqq \log(|\nARC|)$. 
\end{lemma} 

\subsection{Proof of Lemma~\ref{lemma:param.est.1}} 

\begin{proof}[Proof of Lemma~\ref{lemma:param.est.1}]
Condition on $\Ec$. 
Observe that 
\begin{align}
	\left\| \bDelta_\beta\right\|_2^2 &= \left\| \bhP \bDelta_\beta\right\|_2^2 + \left\| \bhP^\perp \bDelta_\beta\right\|_2^2,
\end{align}
where $\bhP^\perp = \bI - \bhP$. 
By definition, we have $\bhP^\perp \hbbeta = \bzero$ and $\bP^\perp \bbeta^* = \bzero$. 
As a result, we have 
\begin{align}
	\left\| \bDelta_\beta \right\|_2^2 &= \left\| \bhP \bDelta_\beta \right\|_2^2 + \left\| \bhP^\perp \bbeta^* \right\|_2^2 
	\\ &= \left\| \bhP \bDelta_\beta \right\|_2^2 + \left\| \left(\bhP^\perp - \bP^\perp \right) \bbeta^* \right\|_2^2
	\\ &= \left\| \bhP \bDelta_\beta \right\|_2^2 + \left\| \left(\bhP - \bP \right) \bbeta^* \right\|_2^2. \label{eq:delta.k.1}
\end{align} 
We proceed to bound each term on the RHS of \eqref{eq:delta.k.1} separately. 

\vspace{10pt}  
\noindent {\em Bounding $\| \bhP \bDelta_\beta\|_2^2$.} 
Since $\bhU$ is an isometry, it follows that 
\begin{align}\label{eq:thm1.2.1}
	\left\| \bhP \bDelta_\beta \right\|_2^2 &= \left\| \bhU^\top \bDelta_\beta \right\|_2^2.
\end{align}
We upper bound \eqref{eq:thm1.2.1} as follows: consider
\begin{align}
	\left\| \hbX^\top \bDelta_\beta\right\|_2^2 
	= \left(\bhU^\top \bDelta\right)^\top \bhS^2 \left(\bhU^\top \bDelta\right)   
	\geq \hs_{\br}^2 \cdot \left\| \bhU^\top \bDelta_\beta\right\|_2^2. \label{eq:key_singular_value_trick}
\end{align}
Collectively, \eqref{eq:thm1.2.1} and \eqref{eq:key_singular_value_trick} imply  
\begin{align}
	\Big\| \bhP \bDelta_\beta \Big\|_2^2 
	\le \hs_{\br}^{-2} \cdot \left\| \hbX^\top \bDelta_\beta \right\|_2^2 \label{eq:intermediate_bound_2_Lem10.1}
\end{align}
To bound the numerator in \eqref{eq:intermediate_bound_2_Lem10.1}, note 
\begin{align}
	\left\| \hbX^\top \bDelta_\beta\right\|_2^2 
	&\le 2 \cdot \left\| \hbX^\top \hbbeta -\bby \right\|_2^2 
	+ 2 \cdot \left\| \bby - \hbX^\top \bbeta^* \right\|_2^2
	\\ 
	& = 2 \cdot \left\| \hbX^\top \hbbeta  -\bby \right\|_2^2
	+ 2 \cdot \left\| \left( \bbX - \hbX \right)^\top \bbeta^* \right \|_2^2,  \label{eq:intermediate_bound_1_Lem10.1} 
\end{align}
where the final equality leverages Lemma~\ref{lemma:theta}. 
We then bound the second term on the RHS of \eqref{eq:intermediate_bound_1_Lem10.1} as 
\begin{align}\label{eq:2_infty_matrix}
	\left\| \left( \bbX - \hbX \right)^\top \bbeta^* \right\|_2^2
	&\le \left\| \bbX - \hbX \right\|_{\txtop}^2 \cdot \left\|\bbeta^* \right\|_2^2. 
\end{align}
Substituting \eqref{eq:intermediate_bound_1_Lem10.1} into \eqref{eq:intermediate_bound_2_Lem10.1} and subsequently using \eqref{eq:2_infty_matrix} implies
\begin{align}
	\left\| \bhP \bDelta_\beta  \right\|_2^2  
	& \le \frac{2}{\hs_{\br}^2} \left\{ \left\| \hbX^\top \hbbeta  -\bby \right\|_2^2
	+ \left\| \bbX - \hbX \right\|_{\txtop}^2 \cdot \left\|\bbeta^* \right\|_2^2 \right\}.\label{eq:thm1.2.2}
\end{align}
Next, observe that
\begin{align}
	\left\| \hbX^\top \hbbeta - \by \right\|_2^2 
	&= \left\| \hbX^\top \hbbeta - \bby - \bxi_y \right\|_2^2 
	\\ &= \left\| \hbX^\top \hbbeta - \bby \right\|_2^2 + \left\|\bxi_y \right\|_2^2 
	 - 2 \left \langle \hbX^\top \hbbeta - \bby, \bxi_y \right \rangle. \label{eq:thm1.2.3}
\end{align}
By the optimality of $\hbbeta$ and Lemma~\ref{lemma:theta}, we have 
\begin{align}
	\left\| \hbX^\top \hbbeta - \by \right\|_2^2 
	&\le \left\| \hbX^\top \bbeta^* - \by \right\|_2^2
	\\ &= \left\| \hbX^\top \bbeta^* -  \bby - \bxi_y \right\|_2^2  
	\\ &= \left\| \hbX^\top \bbeta^* -  \bbX^\top \bbeta^* - \bxi_y \right\|_2^2  
	\\ &= \left\| \left(\hbX - \bbX \right)^\top \bbeta^* \right\|_2^2  + \left\| \bxi_y \right\|_2^2  - 2 \left\langle \hbX^\top \bbeta^* - \bby, \bxi_y \right\rangle. \label{eq:thm1.2.4}
\end{align}
From \eqref{eq:thm1.2.3} and \eqref{eq:thm1.2.4}, we have 
\begin{align}
	\left\| \hbX^\top \hbbeta - \bby \right\|_2^2 
	& \le \left\| \left(\hbX - \bbX \right)^\top \bbeta^* \right\|_2^2 
 	+ 2 \left\langle \hbX^\top \bDelta_\beta, \bxi_y \right\rangle 
	\\ &\quad \le \left\| \hbX - \bbX \right\|_{\txtop}^2 \cdot \left\|\bbeta^* \right\|_2^2 
	 + 2 \left\langle \hbX^\top \bDelta_\beta, \bxi_y \right\rangle. \label{eq:thm1.2.5}
\end{align}
Using \eqref{eq:thm1.2.2} and \eqref{eq:thm1.2.5}, we obtain 
\begin{align} 
 	\left\| \bhP \bDelta_\beta \right\|_2^2 
	& \le \frac{4}{\hs_{\br}^2} \left\{ \left\|  \hbX - \bbX \right\|_{\txtop}^2 \cdot \left\| \bbeta^* \right\|_2^2
	+ \left\langle \hbX^\top \bDelta_\beta, \bxi_y \right\rangle \right\}. \label{eq:thm1.2}
\end{align}
By Lemma~\ref{lemma:weyl} and the fact that $\rank(\bbX) = \br$, it follows that
\begin{align}
	\left\| \hbX - \bX \right\|_\txtop &= \hs_{\br+1} 
	= \left| \hs_{\br+1} - s_{\br+1} \right| \le \left\| \bX - \bbX \right\|_\txtop,
\end{align}
and hence, 
\begin{align}
	\left \| \hbX - \bbX \right\|_\txtop \le \left \| \hbX - \bX \right \|_\txtop + \left\| \bX - \bbX \right\|_\txtop 
	\le 2 \cdot \left \| \bXi_X \right \|_\txtop. \label{eq:op.useful}
\end{align}
We next record a useful fact. By Assumption~\ref{assump:bounded}, $\| \bby \|_2 \le \sqrt{|\nAC|}$ and hence, 
\begin{align}
	\left \| \bbeta^* \right \|_2 &= \| \bbX^{\top, \dagger} \bby \|_2
	\le \| \bbX^\dagger \|_\txtop \cdot \| \bby \|_2
	\le \frac{\sqrt{|\nAC|}}{s_{\br}}. \label{eq:beta.l2.bound}
\end{align}
By Assumption~\ref{assump:spectra}, we can simplify \eqref{eq:beta.l2.bound} as 
\begin{align}
	\left \| \bbeta^* \right \|_2 \le \frac{\sqrt{\br}}{\sqrt{|\nAR|}}. \label{eq:beta.l2.bound.2} 
\end{align}
Inserting \eqref{eq:op.useful} and \eqref{eq:beta.l2.bound.2} into \eqref{eq:thm1.2},
\begin{align}
	\left\| \bhP \bDelta_\beta \right\|_2^2 
	& \lesssim \frac{1}{\hs_{\br}^2} \left\{ \frac{\br \cdot \left \| \bXi_X \right \|^2_\txtop}{|\nAR|}
	+ \left\langle \hbX^\top \bDelta_\beta, \bxi_y \right\rangle \right\}.
\end{align}
Invoking Assumption~\ref{assump:spectra}, the rank condition \eqref{eq:rank.sep}, and Lemmas~\ref{lemma:subg_matrix} and \ref{lemma:annoying}, we conclude that with probability at least $1 - \mathcal{O}\left( |\ARC|^{-10} \right)$
\begin{align}
	\left\| \bhP \bDelta_\beta \right\|_2
	&\lesssim  \frac{\sigma \br } { \sqrt{|\nAR|} \cdot \min\left\{\sqrt{|\nAR|}, \sqrt{|\nAC|}\right\}} + \frac{\sqrt{\sigma \br} \cdot \log^{1/4}(|\nARC|)}{\sqrt{|\AR|} \cdot |\AC|^{1/4}} . 
	\label{eq:hat.pk.1}
\end{align}

\vspace{10pt}     
\noindent {\em Bounding $\| (\bhP - \bP) \bbeta^* \|_2$.} 
From \eqref{eq:beta.l2.bound.2}, 
\begin{align}
	\left\| \left(\bhP - \bP \right) \bbeta^* \right\|_2 
	\le \left\| \bhP - \bP \right\|_{\txtop} \cdot \left\| \bbeta^* \right\|_2
	\le \frac{\sqrt{\br}}{\sqrt{|\nAR|}} \cdot \left\| \bhP - \bP \right\|_{\txtop}. 
\end{align} 
Leveraging Assumption~\ref{assump:spectra}, the rank condition \eqref{eq:rank.sep}, and Lemma~\ref{lemma:subspace}, we have with probability at least $1 - \mathcal{O}\left( |\ARC|^{-10} \right)$: 
\begin{align}
	\left\| \left(\bhP - \bP \right) \bbeta^* \right\|_2 
	&\lesssim  \frac{\sigma \br}{\sqrt{|\AR|} \cdot \min \left\{ \sqrt{|\AR|}, \sqrt{|\AC|} \right\}} .  
	\label{eq:hat.pk.2} 
\end{align}

\vspace{10pt}  
\noindent {\em Collecting terms.} 
Combining \eqref{eq:hat.pk.1} with \eqref{eq:hat.pk.2} and simplifying, we obtain our desired result. 

\end{proof}

\subsection{Proofs of Key Lemmas for Lemma~\ref{lemma:param.est.1}} \label{sec:pcr.lemmas.proofs}

\subsubsection{Proof of Lemma~\ref{lemma:theta}}

\begin{proof}[Proof of Lemma~\ref{lemma:theta}] 
Condition on $\Ec$. 
By Assumption~\ref{assump:subspace}, there exists an $\balpha \in \Rb^{|\AC|}$ such that
\begin{align} 
	\bv_j &= \sum_{q \in \AC} \alpha_q \cdot \bv_q. \label{eq:a7.1}
\end{align}
Thus, following the proof of Proposition~\ref{thm:identification}, 
\begin{align}
	\Ex\left[ Y_{\ell j} \mid \cE \right] 
	&= \Ex \left[ \tY_{\ell j} \mid \cE \right] &&\because \text{\eqref{eq:sutva}}
	\\
	&= \Ex \left[ \left\langle \bu_\ell, \bv_j \right \rangle + \varepsilon_{\ell j} \mid \Ec \right]  && \because \text{Assumption~\ref{assump:lfm}} 
	\\
	&= \left \langle \bu_\ell, \bv_j \right \rangle \mid \cE && \because \text{Assumption~\ref{assump:mean_ind}}
	\\
	&= \sum_{q \in \AC} \alpha_q \cdot \left \langle \bu_\ell, \bv_q \right \rangle \mid \cE && \because \text{\eqref{eq:a7.1}} 
	\\
	&= \sum_{q \in \AC} \alpha_q \cdot \Ex \left[ \left \langle \bu_\ell, \bv_q \right \rangle + \varepsilon_{\ell q} \mid \cE \right] && \because \text{Assumption~\ref{assump:mean_ind}}
	\\
	&= \sum_{q \in \AC}  \alpha_q \cdot \Ex \left[ \tY_{\ell q} \mid \cE \right] && \because \text{Assumption~\ref{assump:lfm}}
	\\
	&= \sum_{q \in \AC}  \alpha_q \cdot \Ex \left[ Y_{\ell q} \mid \cE \right]. && \because \text{\eqref{eq:sutva}}
	\label{eq:alpha.exist}
\end{align}
This establishes that $\bbz$ lies within the column space of $\bbX$, which yields $\bbz = \bP \cdot\bbz$. 
In turn, 
\begin{align} 
   	\bbz^\top \bbeta 
    	&=  \bbz^\top \bP \bbeta 
	= \bbz^\top \bbeta^*. 
\end{align}
Hence, we can reformulate Proposition~\ref{thm:identification} as $A_{ij} = \bbz^\top \bbeta^*$. 
The same arguments can be applied to obtain $\bby = \bbX^\top \bbeta^*$. 
\end{proof}

\subsubsection{Proof of Lemma~\ref{lemma:singular_values}}
\begin{proof}[Proof of Lemma~\ref{lemma:singular_values}] 
Condition on $\Ec$. 
By Lemma~\ref{lemma:weyl}, we obtain 
\begin{align}
    \abs{ \hs_{\br} - s_{\br} }
    &\le \left\| \bX - \bbX \right\|_{\txtop} = \left \| \bXi_X \right \|_\txtop.  \label{eq:s.compressed.1} 
\end{align}
Applying Lemma~\ref{lemma:subg_matrix} to the above completes the proof. 
\end{proof}

\subsubsection{Proof of Lemma~\ref{lemma:subspace}}

\begin{proof}[Proof of Lemma~\ref{lemma:subspace}] 
Condition on $\Ec$. 
Recall $\br = \rank(\bbX)$ and thus $s_{\br+1} = 0$. 
By Lemma \ref{thm:wedin}, 
\begin{align} 
    \left\| \bhP - \bP \right\|_{\txtop} 
    &\le \frac{2}{s_{\br} - s_{\br+1}} \left\| \bX - \bbX \right\|_{\txtop}
    = \frac{2}{s_{\br}} \left\| \bX - \bbX \right\|_{\txtop}. 
\end{align}
Applying Lemma \ref{lemma:subg_matrix} to the above gives our desired result. 
\end{proof} 

\subsubsection{Proof of Lemma~\ref{lemma:annoying}}

\begin{proof}[Proof of Lemma~\ref{lemma:annoying}] 
Condition on $\Hc$. 
Note that $(\hbX, \bbeta^*)$ are $\Hc$-measurable, and the conditional law of $\bxi_y$ given $\Hc$ has independent mean-zero subgaussian coordinates. 
Next, recalling $\hbbeta = \bhU \bhS^{-1} \bhV^\top \by$, we obtain 
\begin{align} \label{eq:rewrite} 
	\hbX^\top \hbbeta = \bhV \bhS \bhU^\top \bhU \bhS^{-1} \bhV^\top \by 
	=  \bhQ \bby +  \bhQ \bxi_y.
\end{align}
We then obtain 
\begin{align}
	&\left\langle \hbX^\top \bDelta_\beta,  \bxi_y  \right \rangle 
	= \left\langle  \bhQ \bby, \bxi_y \right\rangle 
	+ \left\langle  \bhQ \bxi_y, \bxi_y \right\rangle 
	- \left\langle \hbX^\top \bbeta^*, \bxi_y \right\rangle.  \label{eq:lem4.0}
\end{align}
Note that $\bxi_y$ is independent of $\bX$ and is thus also independent of $\bhV$ and $\hbX$. 
Therefore, 
\begin{align}
    \Ex\left[\left\langle  \bhQ \bby, \bxi_y \right\rangle \mid \Hc \right] &= 0, \label{eq:lemma10.5_exp_2}
    \\ \Ex\left[\left\langle \hbX^\top \bbeta^*, \bxi_y \right\rangle \mid \Hc \right]. &= 0 \label{eq:lemma10.5_exp_3}
\end{align}
The cyclic property of the trace operator then yields 
\begin{align}
	\Ex\left[ \left\langle  \bhQ  \bxi_y, \bxi_y  \right\rangle \mid \Hc \right] 
	&= \Ex\left[ \bxi_y^\top   \bhQ \bxi_y \mid \Hc \right] 
	= \Ex \left[\tr(\bxi_y^\top   \bhQ \bxi_y) \mid \Hc\right] 
	\\ &= \tr(\Ex\left[ \bxi_y  \bxi_y^\top \mid \Hc\right] \cdot  \bhQ) 
	\le \tr(\sigma^2  \bhQ)  
	\\ &=  \sigma^2 \br. \label{eq:thm1.lem3.0}
\end{align}
Accordingly, we have 
\begin{align}\label{eq:thm1.lem3.1}
	\Ex \left[ \left\langle \hbX^\top \bDelta_\beta,  \bxi_y  \right\rangle  \mid \Hc \right] &\le \sigma^2  \br. 
\end{align}
Conditional on $\Hc$, the vector $\bhQ \bby$ is deterministic, and $\| \bhQ \bby \|_2 \le \| \bby \|_2 \le \sqrt{|\nAC|}$, where the last inequality follows from Assumption~\ref{assump:bounded}. 
By Lemma~\ref{lemma:hoeffding} and \eqref{eq:lemma10.5_exp_2}, it follows that for any $t > 0$
%
\begin{align}\label{eq:lem4.1}
	\Pb\left(  \left| \left\langle \bhQ \bby, \bxi_y \right\rangle \right|  \ge C \sigma \sqrt{|\nAC| t} \mid \Hc \right) &\le  2 \exp(-t). 
\end{align}
Similarly, using \eqref{eq:lemma10.5_exp_3}, we have for any $t > 0$
\begin{align}
	\Pb\left(  \left| \left\langle \hbX^\top \bbeta^*, \bxi_y \right \rangle \right| \geq 
	C \sigma \sqrt{t} \cdot \left\{ \sqrt{|\nAC|} + \left\| \bXi_X \right\|_\txtop \cdot \left \|\bbeta^* \right\|_2 \right\} \mid \Hc \right)
	 \le 2 \exp(-t), 
	 \label{eq:lem4.2}
\end{align}
where we use the fact that
\begin{align}
	\left\| \hbX^\top \bbeta^* \right\|_2 
	 &= \left\| \hbX^\top \bbeta^* \pm  \bby \right\|_2 
	\\ & = \left\| \left(\hbX - \bbX \right)^\top \bbeta^* + \bby \right\|_2 
	&& \because \text{Lemma~\ref{lemma:theta}}
	\\ &\leq \left\| \left(\hbX - \bbX \right)^\top \bbeta^* \right\|_2 + \left\|\bby \right\|_2 
    \\& \leq \left\| \hbX - \bbX \right\|_{\txtop} \cdot \left\| \bbeta^* \right\|_2 + \sqrt{\abs{\AC}} && \because \text{Assumption~\ref{assump:bounded}}
	\\ &\lesssim \left\|\bXi_X \right\|_{\txtop} \cdot \left\| \bbeta^* \right\|_2 + \sqrt{|\AC|}. &&\because \text{\eqref{eq:op.useful}}
\end{align} 
Furthermore, conditional on $\Hc$, $\bhQ$ is deterministic and is an orthogonal projection of rank $\br$, yielding $\| \bhQ \|_{\txtop} \le 1$ and $\|\bhQ \|^2_F = \br$. Therefore, by Lemma~\ref{lemma:hanson.wright} and \eqref{eq:thm1.lem3.1}, for every $t > 0$,
\begin{align}\label{eq:lem4.3}
	\Pb\left(  \left\langle  \bhQ \bxi_y, \bxi_y \right\rangle - \Ex\left[ \left\langle  \bhQ \bxi_y, \bxi_y \right\rangle \mid \Hc \right]   \geq C \sigma^2 \left\{ \sqrt{\br t} + t \right\} \mid \Hc \right) & \leq  2\exp(-t). 
\end{align}
Combining \eqref{eq:lem4.1}, \eqref{eq:lem4.2}, and \eqref{eq:lem4.3}, and using a union bound, conditional on $\Hc$, with probability at least $1 - 6\exp(-t)$,
\begin{align}
	\left\langle \hbX^\top \bDelta_\beta, \bxi_y \right \rangle 
	\lesssim \sigma^2 \br + \sigma^2 \left( \sqrt{\br t} + t \right) + \sigma \sqrt{t} \left( \sqrt{|\nAC|} + \| \bXi_X \|_\txtop \cdot \| \bbeta^* \|_2 \right). 
\end{align}
Because the conditional failure probability is bounded by $6 \exp(-t)$ uniformly over the realized value of $\bX$, integrating over $\bX$ gives the same bound conditional only on $\Ec$. Taking $t = C_0 L$ for $C_0 > 0$ sufficiently large gives our desired result. 
\end{proof}

\section{Proof of Theorem~\ref{thm:consistency}} \label{sec:proof.consistency} 
\begin{proof}[Proof of Theorem~\ref{thm:consistency}] 
Condition on $\Ec$. 
By Lemma~\ref{lemma:theta}, 
\begin{align}
	\hA_{ij} - A_{ij}
	&= \left \langle \bbz, \bDelta_\beta \right \rangle  + \left \langle \bxi_z, \hbbeta \right \rangle. 
	\label{eq:consistency.1}
\end{align} 
%
%
Under Assumption~\ref{assump:bounded}, 
\begin{align}
	\left| \left \langle \bbz,  \bDelta_\beta \right \rangle \right| 
	&\le \| \bbz \|_2 \cdot \|  \bDelta_\beta \|_2
	\le \sqrt{|\AR|} \cdot \|  \bDelta_\beta \|_2.
	\label{eq:consistency.4} 
\end{align}
Next, define the event
\begin{align}
	\Gc_\xi \coloneqq \left\{ \left| \left \langle \bxi_z, \hbbeta \right \rangle \right| \le C \sigma \sqrt{L} \cdot \left \| \hbbeta \right \|_2 \right\}
\end{align}
for some absolute constant $C> 0$. 
Define the sigma-field $\Hc \coloneqq \sigma(\Ec, \bX, \by)$. Then, $\hbbeta$ is $\Hc$-measurable while $\bxi_z$ has independent, mean-zero, subgaussian coordinates. Therefore, by Lemma~\ref{lemma:hoeffding}, 
\begin{align}
	\Pb\left( \Gc_{\xi}^c \mid \Ec \right) = \Ex\left[ \Pb\left( \Gc_{\xi}^c \mid \Hc \right) \mid \Ec \right] \lesssim |\nARC|^{-10}.  
\end{align}
Using \eqref{eq:beta.l2.bound.2}, we expand  
\begin{align}
	\left \| \hbbeta \right \|_2 \le \left \| \bbeta^* \right\|_2 + \left\| \bDelta_\beta \right \|_2 
	\lesssim \frac{\sqrt{\br}}{\sqrt{|\nAR|}} + \left\| \bDelta_\beta \right \|_2. 
\end{align}
Hence, on $\Gc_\xi$, we have
\begin{align}
	\left| \hA_{ij} - A_{ij} \right| 
	\lesssim \left( \sqrt{|\nAR|} + \sigma \sqrt{L} \right) \cdot \left \| \bDelta_\beta \right \|_2 
	+ \frac{\sigma \sqrt{\br L}}{\sqrt{|\nAR|}}. 
	%
\end{align}
Define the event
\begin{align}
	\Gc_{\PCR, \beta} \coloneqq \left\{ \left\| \bDelta_\beta \right \|_2 \le C_0 \Lambda_\beta \right\},
\end{align}
where $C_0>0$ is an absolute constant and $\Lambda_\beta$ is defined in \eqref{eq:pcr.error}. 
By Lemma~\ref{lemma:param.est.1}, $\Pb(\Gc_{\PCR, \beta}^c \mid \Ec) \lesssim |\nARC|^{-10}$. 
Ergo, on the joint event $\Gc_\xi \cap \Gc_{\PCR, \beta}$,  
\begin{align}
	\left| \hA_{ij} - A_{ij} \right| 
	\lesssim \left( 1 + \frac{\sigma \sqrt{L}}{\sqrt{|\nAR|}} \right) \cdot \left\{ \frac{\sigma \br}{\min\{\sqrt{|\nAC|}, \sqrt{|\nAR|} \}} + \frac{\sqrt{\sigma \br} L^{1/4}}{|\nAC|^{1/4}} \right\}
	+ \frac{\sigma \sqrt{\br L}}{\sqrt{|\nAR|}}.
\end{align}
Taking a union bound, we conclude the above holds with probability at least $1 - \Oc(|\nARC|^{-10})$. 
Simplifying the bound then concludes the proof. 
\end{proof}

\section{Proof of Corollary~\ref{cor:mcar}} \label{sec:proof.cor.mcar}

\begin{proof}[Proof of Corollary~\ref{cor:mcar}] 
Without loss of generality, consider the recovery of the entry $(i,j) = (1,1)$. 
The proof has two parts: first, we show that a fully observed anchor cross of the required size exists with high probability; second, we apply Theorem~\ref{thm:consistency} to that cross.

For every pair of sets $\Rc \subseteq [m] \setminus \{1\}$, $\Cc \subseteq [n] \setminus \{1\}$, $|\Rc| = |\Cc| = d$, let $\Gc_{\Rc \Cc}$ denote the event that all entries in $\Sc_{\Rc \Cc} \coloneqq (\Rc \times \Cc) \cup (\Rc \times \{1\}) \cup (\{1\} \times \Cc)$ are observed. On $\Gc_{\Rc \Cc}$, $\Rc$ and $\Cc$ form valid anchor sets for the target $(1,1)$. 
The three sets in $\Sc_{\Rc \Cc}$ are pairwise disjoint. Consequently, $| \Sc_{\Rc \Cc}| = d^2 + 2d = (d+1)^2 -1 \eqqcolon Q$ observations. Under MCAR, $\Pb(\Gc_{\Rc \Cc}) = p^Q$. Additionally, there are $N = \binom{n-1}{d}^2$ candidate pairs $(\Rc, \Cc)$. Define
\begin{align}
	X \coloneqq \sum_{\substack{\Rc \subseteq [m] \setminus \{1\}, |\Rc|=d \\ \Cc \subseteq [n] \setminus \{1\}, |\Cc|=d}} \boldsymbol{1}\left\{\Gc_{\Rc \Cc} \right\}. 
\end{align}
Then, $X \ge 1$ precisely when at least one valid $d \times d$ anchor cross exists, and 
\begin{align}
	\mu \coloneqq \Ex[X] = N p^Q. 
\end{align}
For a sufficiently large constant $C > 0$, set 
\begin{align}
	p_0 \coloneqq \left(\frac{C \log(2/\eta)}{N} \right)^{1/Q}. 
\end{align}
Since $\{X \ge 1\}$ is increasing in the observation indicators, it is enough to establish the desired probability bound at $p = p_0$. 
At $p = p_0$, 
\begin{align}
	\mu = N p_0^Q = C \log(2/\eta).  \label{eq:mcar.1}
\end{align}
Notably, the different events $\Gc_{\Rc \Cc}$ are not independent because different candidate crosses may share observations. 
%
Consider two distinct candidates $(\Rc, \Cc)$ and $(\Rc', \Cc')$, and write $a \coloneqq \left| \Rc \cap \Rc' \right|$ and $b \coloneqq \left| \Cc \cap \Cc' \right|$. 
Their required observation sets share exactly
\begin{align}
	s(a,b) \coloneqq a + b + ab \label{eq:mcar.2}
\end{align}
entries: $a$ entries in the target column, $b$ entries in the target row, and $ab$ entries in the anchor block. Hence, it follows that 
\begin{align}
	\Pb\left(\Gc_{\Rc \Cc} \cap \Gc_{\Rc' \Cc'} \right) = p_0^{2Q - s(a,b)}. \label{eq:mcar.3} 
\end{align}
Set $k \coloneqq n-1$, and define $\rho_k(a) \coloneqq \binom{d}{a} \binom{k-d}{d-a} \binom{k}{d}^{-1}$. 
%
This is the probability that two uniformly chosen $d$-subsets of a $k$-element set have intersection of size $a$.
Let $\Delta$ be the ordered dependency sum in Lemma~\ref{lemma:janson}. Using \eqref{eq:mcar.2} and \eqref{eq:mcar.3},
\begin{align}
	\frac{\Delta}{\mu^2} = \sum_{\substack{0 \le a,b \le d \\ (a,b) \neq \{(0,0), (d,d)\}}} 
	\rho_{k}(a) \rho_{k}(b) p_0^{-s(a,b)}. \label{eq:mcar.4}
\end{align}
The pair $(0,0)$ is excluded because the corresponding events depend on disjoint mask entries, while $(d,d)$ represents the same candidate. 

For $a \ge 1$, a union bound over the possible $a$-element intersection gives
\begin{align}
	\rho_k(a) \le \binom{d}{a} \left(\frac{d}{k-d+1} \right)^a. 
\end{align}
Since $a \le d$ and $d = o(\log(n))$, we have $k-d+1 \asymp n$. Therefore,
\begin{align}
	\rho_k(a) \le \binom{d}{a} \left(\frac{d}{k-d+1} \right)^a \le  \frac{1}{a!} \left(\frac{C_1 d^2}{n} \right)^a 
	\label{eq:mcar.5}
\end{align}
for an absolute constant $C_1 > 0$. The same inequality remains valid for $a=0$ when the right-hand side is interpreted as $1$. 
The standard binomial-coefficient bound gives $N \le (C_2 n / d)^{2d}$ for a sufficiently large $C_2 > 0$. 
%
%
Since $C \log(2/ \eta) \ge 1$ for large enough $C > 0$, 
\begin{align}
	p_0^{-s(a,b)} \le N^{s(a,b)/Q} \le \left(\frac{C_2 n}{d} \right)^{2s(a,b)/(d+2)}, \label{eq:mcar.7}
\end{align}
where we used $Q = d (d+2)$. 
Combining \eqref{eq:mcar.5} and \eqref{eq:mcar.7}, the summand in \eqref{eq:mcar.4} is bounded by
\begin{align}
	\frac{1}{a! b!} \left( \frac{C_3 d^2}{n} \right)^{a+b} \left(\frac{C_3 n}{d} \right)^{2s(a,b)/(d+2)} 
	\label{eq:mcar.8}
\end{align}
for an absolute constant $C_3 > 0$. 
Define 
\begin{align}
	h(a,b) \coloneqq a + b - \frac{2 s(a,b)}{d+2} = \frac{d(a+b) - 2ab}{d+2}. \label{eq:mcar.9}
\end{align}
Collecting the powers of $n$ and $d$ in \eqref{eq:mcar.8}, and using $d^a / a! \le (ed/a)^a$, shows that, for some $C_4 > 0$,
\begin{align}
	\frac{1}{a! b!} \left( \frac{C_3 d^2}{n} \right)^{a+b} \left(\frac{C_3 n}{d} \right)^{2s(a,b)/(d+2)}  \le \exp(C_4 d) \left(\frac{d}{n}\right)^{h(a,b)}. \label{eq:mcar.10}
\end{align}
It remains to lower bound $h(a,b)$. Set $g(a,b) \coloneqq d(a+b) - 2ab$. For $1 \le a \le d-1$, $g(a,b)$ is linear in $b$, so its minimum over $b \in [0,d]$ occurs at $b = 0$ or $b=d$. Consequently, $g(a,b) \ge \min\{da, d(d-a)\} \ge d$. If $a=0$, then $b \ge 1$ and $g(0,b) = db \ge d$. If $a=d$, distinctness requires $b \le d-1$, and $g(d,b) = d(d-b) \ge d$. Thus, for every pair appearing in \eqref{eq:mcar.4}, 
\begin{align}
	h(a,b) \ge \frac{d}{d+2}. \label{eq:mcar.11}
\end{align}
There are at most $(d+1)^2$ terms in \eqref{eq:mcar.4}. Hence, it follows from \eqref{eq:mcar.10} and \eqref{eq:mcar.11} that 
\begin{align}
	\frac{\Delta}{\mu^2} \le (d+1)^2 \exp(C_4 d) \left(\frac{d}{n} \right)^{d/(d+2)}. \label{eq:mcar.12}
\end{align}
The logarithm of the right-hand side is bounded by $2 \log(d+1) + C_4d - d/(d+2) \log(n/d)$. Because $d = o(\log(n))$, we have 
\begin{align}
	\frac{\Delta}{\mu^2} = o(1). \label{eq:mcar.13}
\end{align}
Since $\eta$ is fixed, $\mu = C \log(2/\eta)$ is fixed, so \eqref{eq:mcar.13} also gives $\Delta = o(1)$. Lemma~\ref{lemma:janson} now yields $\Pb(X = 0) \le \exp(-\mu + \Delta/2)$. For all sufficiently large $n$, $\Delta \le \mu$, and hence $\Pb(X = 0) \le \exp(-\mu/2) = (\eta/2)^{C/2}$. Choosing $C \ge 2$ sufficiently large gives
\begin{align}
	\Pb(X = 0) \le \frac{\eta}{2}. \label{eq:mcar.14}
\end{align}
Thus, with probability at least $1 - \eta/2$, a valid anchor pair with $d$ rows and $d$ columns exists. 

We now apply Theorem~\ref{thm:consistency}. On the event $\{X \ge 1\}$, apply the mask-measurable selection rule stipulated in the corollary. 
The auxiliary condition in Theorem~\ref{thm:consistency} gives $\sigma^2 \log(d^2) = o(d)$, and therefore, $\min\{d, d^2 /(\sigma^2 \log(d^2))\} = d$. Thus, 
\begin{align}
	\Phi \lesssim \frac{\sigma \br}{\sqrt{d}} + \frac{\sqrt{\sigma \br} \log^{1/4}(d)}{d^{1/4}}. \label{eq:mcar.15}
\end{align}
Choosing $c$ sufficiently large relative to the constant implicit in Theorem~\ref{thm:consistency} ensures that the theorem’s upper bound is at most $\delta$. 
Thus, on $\{X \ge 1\}$, Theorem~\ref{thm:consistency} gives 
\begin{align}
	\Pb\left( \left|\hA_{11} - A_{11} \right| > \delta \mid \Ec \right) \le C_5 d^{-20} \label{eq:mcar.16}
\end{align}
for a constant $C_5 > 0$. 
As a result, combining \eqref{eq:mcar.14} and \eqref{eq:mcar.16},
\begin{align}
	\Pb\left( \left|\hA_{11} - A_{11} \right| > \delta \mid \bU, \bV \right) \le \Pb\left(X = 0 \right) + C_5d^{-20} 
	\le \frac{\eta}{2} + C_5d^{-20}. 
\end{align}
Thus, for sufficiently large $n$, $C_5 d^{-20} \le \eta/2$, which gives our desired result. 
The argument applies identically to any target entry $(i,j)$. Moreover, none of the $Q$ required observations is the target entry itself, so the conclusion is unchanged if one conditions on $D_{ij} = 0$. 
\end{proof}

\section{Proof of Theorem~\ref{thm:normality}} \label{sec:proofs.normality}
Theorem~\ref{thm:consistency} establishes that \SNN~can achieve vanishing error under suitable conditions. 
To prove asymptotic normality, however, we need a sharper representation: we must identify the part of the error that remains at first order and show that all remaining terms are negligible relative to the standard deviation of this leading fluctuation. 

The direct expansion of the \SNN~estimator in \eqref{eq:consistency.1} reveals the challenge. In particular, it contains $\left \langle \bbz, \bDelta_\beta \right \rangle$, which is not automatically of smaller order than the stochastic fluctuation of $\hA_{ij}$. Moreover, $\bDelta_\beta$ depends on the noisy anchor block $\bX$ through its estimated singular subspaces, so analyzing its distribution directly would be unnecessarily delicate. 

Borrowing upon the analysis in \cite{twsf}, the key idea is to rewrite the target through a residual-corrected population identity. For $\boldb \in \Rb^{|\AR|}$ and $\bolda \in \Rb^{|\AC|}$, define 
\begin{align}
	\Sc(\boldb, \bolda; \bbX, \bby, \bbz) \coloneqq \left \langle \bbz, \boldb \right \rangle + \left \langle \bolda, \bby - \bbX^\top \boldb \right \rangle.  
	\label{eq:pop.score}
\end{align}
By Lemma~\ref{lemma:theta}, $\Sc(\bbeta^*, \bolda; \bbX, \bby, \bbz) = A_{ij}$ for every choice of $\bolda$ since $\bby - \bbX^\top \bbeta^* = \bzero$. 
The vector $\bolda$, therefore, does not change the population target. Instead, it lets us choose a representation whose first-order sensitivity to perturbations in $\bbeta^*$ vanishes. Differentiating $\Sc$ with respect to $\boldb$ shows that this sensitivity is $\bbz - \bbX \bolda$. Thus, the nuisance-linear term $\langle \bbz, \bDelta_\beta \rangle$ is canceled when 
\begin{align}
	\bbX \bolda = \bbz. 
	\label{eq:riesz.formula}
\end{align}
Assumption~\ref{assump:subspace}, through Lemma~\ref{lemma:theta}, ensures that $\bbz$ lies in the column space of $\bbX$, so \eqref{eq:riesz.formula} is feasible. 
We take the canonical minimum $\ell_2$-norm solution $\balpha^*$, as defined in \eqref{eq:pop.params}. 
It is also stable in norm: by Assumptions~\ref{assump:bounded} and \ref{assump:spectra}, 
\begin{align}
	\left\| \balpha^* \right\|_2 
	\le \frac{\sqrt{\br}}{\sqrt{|\nAC|}}. \label{eq:riesz.bound}
\end{align}
To use the same cancellation at the empirical level, we introduce a proof-only projected version of $\balpha^*$. 
Let $\bhQ \coloneqq \hbX^\dagger \hbX$ be the orthogonal projector onto the row space of $\hbX$, and define 
\begin{align}
	\tbalpha \coloneqq \bhQ \balpha^*, 
	\quad
	\bDelta_{\talpha} \coloneqq \tbalpha - \balpha^*.
	\label{eq:proof.proj}
\end{align}
This projection does not modify the estimator $\hA_{ij}$; instead, it allows us to add an empirical residual correction that is exactly zero. Indeed, since $\hbbeta = \hbX^{\top, \dagger} \by$, we have $\hbX^\top \hbbeta = \bhQ \by$. Moreover, because $\hbbeta$ lies in the left singular subspace retained by the rank-$\br$ truncation, $\bX^\top \hbbeta = \hbX^\top \hbbeta$. This yields 
\begin{align}
	\by - \bX^\top \hbbeta = \by - \hbX^\top \hbbeta = \left(\bI - \bhQ \right) \by,
\end{align}
and hence, 
\begin{align}
	\left\langle \tbalpha, \by - \bX^\top \hbbeta \right \rangle = \left\langle \bhQ \balpha^*, \left(\bI - \bhQ \right) \by \right \rangle = 0. 
\end{align}
Consequently, the estimator admits the exact identity
\begin{align}
	\hA_{ij} = \left \langle \bz, \hbbeta \right \rangle + \left\langle \tbalpha, \by - \bX^\top \hbbeta \right \rangle. 
	\label{eq:new.rep}
\end{align}
This representation is the starting point of the expansion. In what follows, our strategy is to show that the estimation error separates into $\hA_{ij} - A_{ij} = G + R$. 
The leading term $G$ is a sum of three linear forms in mutually independent noise blocks; the remainder $R$ collects the \PCR~error, representer-projection error, and design perturbation. The remainder of the proof formalizes this separation. 

\subsection{Key Lemmas for Theorem~\ref{thm:normality}} \label{sec:proofs.normality.lemmas}

\begin{lemma} \label{lemma:norm.decomp}
Let the setup of Theorem~\ref{thm:normality} hold. Then, conditional on $\Ec$, $\hA_{ij} - A_{ij} = G + R$, where 
\begin{align}
	G &\coloneqq \left \langle \bxi_z, \bbeta^* \right \rangle
	+ \left \langle \bxi_y, \balpha^* \right \rangle
	- \left \langle \balpha^*, \bXi_X^\top \bbeta^* \right \rangle
	\\
	R &\coloneqq 
	\left \langle \bxi_z, \bDelta_\beta \right \rangle
	+ \left \langle \bxi_y, \bDelta_{\talpha} \right \rangle
	- \left \langle \balpha^*, \bXi_X^\top \bDelta_\beta \right \rangle
	- \left \langle \bDelta_{\talpha}, \bXi_X^\top \bbeta^* \right \rangle
	- \left \langle \bDelta_{\talpha}, \bbX^\top \bDelta_\beta \right \rangle
	- \left \langle \bDelta_{\talpha}, \bXi_X^\top \bDelta_\beta \right \rangle. 
\end{align}
\end{lemma}

\begin{lemma} \label{lemma:norm.lead}
Let the setup of Theorem~\ref{thm:normality} hold. If Assumption~\ref{assump:subg} is specialized to Gaussian noise, then, conditional on $\Ec$, $G \sim \mathcal{N}(0, \upsilon^2)$. If instead Assumption~\ref{assump:subg} holds for the general subgaussian case along with condition \eqref{eq:clt}, then, conditional on $\Ec$, $G / \upsilon \rightsquigarrow \mathcal{N}(0, 1)$. 
\end{lemma}

\begin{lemma} \label{lemma:norm.remainder}
Let the setup of Theorem~\ref{thm:normality} hold. Then, conditional on $\Ec$, with probability at least $1 - \Oc(|\nARC|^{-10})$, $\left| R \right| \lesssim \Psi$, 
where $\Psi$ is defined as in \eqref{eq:psi}. 
\end{lemma}

\subsection{Completing Proof of Theorem~\ref{thm:normality}}

\begin{proof}[Proof of Theorem~\ref{thm:normality}] 
Condition on $\Ec$. By Lemma~\ref{lemma:norm.decomp}, we have $\hA_{ij} - A_{ij} = G + R$. 
By Lemma~\ref{lemma:norm.lead}, either $G / \upsilon \sim \Nc(0,1)$ or $G / \upsilon \rightsquigarrow \Nc(0,1)$. 
At the same time, Lemma~\ref{lemma:norm.remainder} states that 
\begin{align}
	\Pb\left( \left| R \right| \le C \Psi \mid \Ec \right) \ge 1 - \Oc(|\nARC|^{-10})
\end{align}
for a sufficiently large constant $C > 0$. Therefore, for any $\epsilon > 0$,
\begin{align}
	\Pb \left( \frac{|R|}{\upsilon} > \epsilon \mid \Ec \right) \le C |\nARC|^{-10} + \boldsymbol{1}\left\{ \frac{C \Psi}{\upsilon} > \epsilon \right\}. 
\end{align}
Because $|\nARC|^{-10} \rightarrow 0$ and $\Psi / \upsilon \rightarrow 0$ by assumption, it follows that $R / \upsilon \xrightarrow{p} 0$. Hence, combining the above and applying conditional Slutsky's theorem completes the proof. 
\end{proof}

\subsection{Proofs of Key Lemmas for Theorem~\ref{thm:normality}}

\subsubsection{Proof of Lemma~\ref{lemma:norm.decomp}}

\begin{proof}[Proof of Lemma~\ref{lemma:norm.decomp}]
Condition on $\Ec$. Recall the identity \eqref{eq:new.rep}. 
%
%
Define $\bdelta \coloneqq \bxi_y - \bXi_X^\top \bbeta^*$. Then,
\begin{align}
	\hA_{ij} - A_{ij} = \left \langle \bz, \hbbeta \right \rangle - \left \langle \bbz, \bbeta^* \right \rangle + \left \langle \tbalpha, \by - \bX^\top \hbbeta \right \rangle
	= \left \langle \bbz, \bDelta_\beta \right \rangle + \left \langle \bxi_z, \bbeta^* \right \rangle + \left \langle \bxi_z, \bDelta_\beta \right \rangle. \label{eq:decomp.1}
\end{align}
Now, note that
\begin{align}
	\by - \bX^\top \hbbeta &= \bby + \bxi_y - \left( \bbX + \bXi_X \right)^\top \left( \bbeta^* + \bDelta_\beta \right)
	\\
	&= \bdelta - \bbX^\top \bDelta_\beta - \bXi_X^\top \bDelta_\beta. 
\end{align}
Using $\tbalpha = \balpha^* + \bDelta_{\talpha}$, 
\begin{align}
	\left \langle \tbalpha, \by - \bX^\top \hbbeta \right \rangle
	&= \left \langle \balpha^*, \bdelta \right \rangle - \left \langle \balpha^*, \bbX^\top \bDelta_\beta \right \rangle - \left \langle \balpha^*, \bXi_X^\top \bDelta_\beta \right \rangle + \left \langle \bDelta_{\talpha}, \bdelta \right \rangle 
	\\ &\quad - \left \langle \bDelta_{\talpha}, \bbX^\top \bDelta_\beta \right \rangle - \left \langle \bDelta_{\talpha}, \bXi_X^\top \bDelta_\beta \right \rangle. \label{eq:decomp.2}
\end{align}
By \eqref{eq:riesz.formula} and \eqref{eq:pop.params}, we have 
\begin{align}
	\left \langle \balpha^*, \bbX^\top \bDelta_\beta \right \rangle 
	= \left \langle \bbX \balpha^*, \bDelta_\beta \right \rangle
	= \left \langle \bbz, \bDelta_\beta \right \rangle. \label{eq:decomp.3}
\end{align}
Combining \eqref{eq:decomp.1} and \eqref{eq:decomp.2} and applying the cancellation \eqref{eq:decomp.3} gives our desired result. 
\end{proof}

\subsubsection{Proof of Lemma~\ref{lemma:norm.lead}}

\begin{proof}[Proof of Lemma~\ref{lemma:norm.lead}]
Condition on $\Ec$. Expanding $G$, we have
\begin{align}
	G &= 
	\sum_{\ell \in \AR} \beta^*_\ell  \varepsilon_{\ell j}
	+ \sum_{q \in \AC} \alpha^*_q \varepsilon_{i q}
	- \sum_{\ell \in \AR} \sum_{q \in \AC} \beta^*_\ell \alpha^*_q \varepsilon_{\ell q}. 
\end{align} 
By Assumption~\ref{assump:subg}, 
\begin{align}
	\Var(G \mid \Ec) = 
	\sum_{\ell \in \AR} \left(\beta^*_\ell \right)^2 \sigma^2_{\ell j}
	+ \sum_{q \in \AC} \left(\alpha^*_q \right)^2 \sigma^2_{i q}
	+ \sum_{\ell \in \AR} \sum_{q \in \AC} \left(\beta^*_\ell \alpha^*_q \right)^2 \sigma^2_{\ell q}. 
	\label{eq:varg}
\end{align}
We now establish our desired result under the two different regimes. 

\medskip \noindent {\em (i) Gaussian noise.} 
If Assumption~\ref{assump:subg} is specialized to Gaussian noise, then $G$ is immediately a mean-zero Gaussian random variable with variance \eqref{eq:varg}. 

\medskip \noindent {\em (ii) General subgaussian noise.} 
Now, consider the more general subgaussian noise environment. 
Define the index set 
\begin{align}
	\Sc \coloneqq \left(\AR \times \{j\} \right) \cup \left(\{i\} \times \AC\right) \cup \left(\AR \times \AC \right). 
	\label{eq:index}
\end{align}
In turn, for $s = (u, v) \in \Sc$, define the population weights $\bomega = [\omega_{uv}]$ as 
\begin{align}
	\omega_{uv} = \begin{cases}
		\beta^*_\ell, &(u,v) = (\ell, j), \ell \in \AR,
		\\
		\alpha^*_q, &(u,v) = (i, q), q \in \AC,
		\\
		-\beta^*_\ell \alpha^*_q,  &(u,v) = (\ell, q), \ell \in \AR, q \in \AC,
	\end{cases}
\end{align} 
yielding $G = \sum_{s \in \Sc} \omega_s \varepsilon_s$ and $\upsilon^2 = \sum_{s \in \Sc} \omega^2_s \sigma^2_s$. 
Define the normalized summand $Z_s \coloneqq \omega_s \varepsilon_s / \upsilon$. 
Because the noise terms are conditionally subgaussian with $\| \varepsilon_s \|_{\psi_2} \le C_\varepsilon \sigma_s$ (Assumption~\ref{assump:subg}), there exists a constant $C_4 < \infty$ such that $\Ex[\varepsilon^4_s \mid \Ec] \le C_4 \sigma^4_s$. Thus,
\begin{align}
	\sum_{s \in \Sc} \Ex\left[Z^4_s \mid \Ec\right] &= \frac{1}{\upsilon^4} \sum_{s \in \Sc} \omega_s^4 \Ex\left[\varepsilon^4_s \mid \Ec\right] \le \frac{C_4}{\upsilon^4} \sum_{s \in \Sc} \omega_s^4 \sigma^4_s. 
\end{align}
Moreover, note that 
%
\begin{align}
	\sum_{s \in \Sc} \omega_s^4 \sigma^4_s
	\le \left(\max_{s \in \Sc} \omega_s^2 \sigma^2_s \right) \sum_{s \in \Sc} \omega_s^2 \sigma^2_s
	\le \rho^2 \upsilon^4,
\end{align}
where $\rho \coloneqq \max_{s \in \Sc} |\omega_s|\sigma_s / \upsilon$. 
Based on \eqref{eq:clt}, 
\begin{align}
	\sum_{s \in \Sc} \Ex\left[Z^4_s \mid \Ec \right] \le C_4 \rho^2 \rightarrow 0. 
\end{align}
This corresponds precisely to Lyapnunov's condition with $\delta = 2$. Hence, by the conditional Lyapunov central limit theorem for triangular arrays \citep{Billingsley}, 
\begin{align}
	\frac{G}{\upsilon} \rightsquigarrow \Nc(0,1). 
\end{align}
This completes the proof. 
\end{proof}

\subsubsection{Proof of Lemma~\ref{lemma:norm.remainder}}

\begin{proof}[Proof of Lemma~\ref{lemma:norm.remainder}]
Condition on $\Ec$. 
Define   
\begin{align}
	\eta_x \coloneqq C_\eta \sigma \left( \sqrt{|\nAR|} + \sqrt{|\nAC|} + \sqrt{L} \right),
	\label{eq:noise} 
\end{align}
where $C_\eta > 0$ is an absolute constant. 
Additionally, define the events
\begin{align}
	\Gc_{\PCR, \beta} \coloneqq \left\{ \left\| \bDelta_\beta \right \|_2 \le C_\beta \Lambda_\beta \right\},
	\quad
	\Gc_\noise \coloneqq \left\{ \left \| \bXi_X \right \|_\txtop \le \eta_x \right\},
	\label{eq:events.pcr.noise} 
\end{align}
where $C_\beta > 0$ is an absolute constant and $\Lambda_\beta$ is defined in \eqref{eq:pcr.error}. By \eqref{eq:rank.sep} and Assumption~\ref{assump:spectra},  
\begin{align}
	\eta_x \lesssim s_{\br},
	\quad
	\left \| \bbX \right \|_\txtop \lesssim s_{\br},
	\quad
	s^2_{\br} \asymp \frac{|\nARC|}{\br}. \label{eq:simplify}  
\end{align}
We proceed to bound each term in $R$ based on Lemma~\ref{lemma:norm.decomp}. 

\medskip \noindent {\em Term 1.} On $\Gc_\noise \cap \Gc_{\PCR, \beta}$, we apply \eqref{eq:riesz.bound} to obtain 
\begin{align}
	\left| \left \langle \balpha^*, \bXi_X^\top \bDelta_\beta \right \rangle \right|
	\le \left \| \balpha^* \right \|_2 \cdot \left \| \bXi_X \right \|_\txtop \cdot \left \| \bDelta_\beta \right \|_2
	\lesssim \frac{\eta_x \sqrt{\br} \Lambda_\beta}{\sqrt{|\nAC|}}.
	\label{eq:t1}
\end{align}
By Lemmas~\ref{lemma:subg_matrix} and \ref{lemma:param.est.1}, 
\begin{align}
	\Pb\left( (\Gc_\noise \cap \Gc_{\PCR, \beta})^c \mid \Ec \right) \le 
	\Pb\left(\Gc_\noise^c \mid \Ec\right) + \Pb\left(\Gc_{\PCR, \beta}^c \mid \Ec \right) \lesssim |\nARC|^{-10}. 
	\label{eq:hp1} 
\end{align}

\noindent {\em Term 2.} 
Define the event 
\begin{align}
	\Gc_z \coloneqq \left\{ \left| \left \langle \bxi_z, \bDelta_\beta \right \rangle \right| \le C_z \sigma \sqrt{L} \Lambda_\beta \right\},
	\label{eq:t2}
\end{align}
where $C_z > 0$ is an absolute constant. Define the sigma-field $\Hc_z \coloneqq \sigma(\Ec, \bX, \by)$. Then, $\bDelta_\beta$ is $\Hc_z$-measurable and $\bxi_z$ is independent of $\Hc_z$. By Lemma~\ref{lemma:hoeffding},
\begin{align}
	\Pb\left( \Gc_{\PCR, \beta} \cap \Gc_z^c \mid \Ec \right)
	= \Ex\left[ \boldsymbol{1}\{\Gc_{\PCR, \beta}\} \cdot \Pb\left(\Gc_z^c \mid \Hc_z \right) \mid \Ec \right] \lesssim |\nARC|^{-10}. 
	\label{eq:hp2}
\end{align}

\noindent {\em Term 3.} 
From \eqref{eq:pop.params}, observe that $\balpha^* = \bQ \balpha^*$, where $\bQ = \bbX^\dagger \bbX$. 
Hence, by \eqref{eq:proof.proj}, we have 
\begin{align}
	\left\| \bDelta_{\talpha} \right \|_2
	\le \left\| \bhQ - \bQ \right\|_\txtop \cdot \left \| \balpha^* \right\|_2. 
\end{align}
On $\Gc_\noise$, applying Lemma~\ref{thm:wedin} with \eqref{eq:riesz.bound} yields 
\begin{align}
	\left\| \bDelta_{\talpha} \right \|_2
	\lesssim
	\frac{\eta_x \sqrt{|\nAR|}}{s^2_{\br}}. \label{eq:delta.g.bound}
\end{align}
Define the event
\begin{align}
	\Gc_y \coloneqq \left\{ \left| \left \langle \bDelta_{\talpha}, \bxi_y \right \rangle \right| 
	\le C_y  \frac{\sigma \eta_x \sqrt{L |\nAR|}}{s^2_{\br}} \right\},
	\label{eq:t3}
\end{align}
where $C_y > 0$ is an absolute constant. 
Define the sigma-field $\Hc_y \coloneqq \sigma(\Ec, \bX)$. Then, $\bDelta_{\talpha}$ is $\Hc_y$-measurable and $\bxi_y$ is independent of $\Hc_y$. By Lemma~\ref{lemma:hoeffding},
\begin{align}
	\Pb\left( \Gc_\noise \cap \Gc_y^c \mid \Ec \right)
	= \Ex\left[ \boldsymbol{1}\{\Gc_\noise\} \cdot \Pb\left(\Gc_y^c \mid \Hc_y \right) \mid \Ec \right] \lesssim |\nARC|^{-10}. 
	\label{eq:hp3}
\end{align} 

\noindent {\em Term 4.} On $\Gc_\noise$, we use \eqref{eq:beta.l2.bound.2} and \eqref{eq:delta.g.bound}
\begin{align}
	\left| \left \langle \bDelta_{\talpha}, \bXi_X^\top \bbeta^* \right \rangle \right|
	\le \left\| \bDelta_{\talpha} \right\|_2 \cdot \left \| \bXi_X \right \|_\txtop \cdot \left \| \bbeta^* \right \|_2
	\lesssim \frac{\eta_x^2 \sqrt{|\ARC|}}{s^3_{\br}}. 
	\label{eq:t4}
\end{align}

\noindent {\em Term 5.} On $\Gc_\noise \cap \Gc_{\PCR, \beta}$, we use \eqref{eq:delta.g.bound} and \eqref{eq:simplify}  
\begin{align}
	\left| \left \langle \bDelta_{\talpha}, \bbX^\top \bDelta_\beta \right \rangle \right|
	\le \left\| \bDelta_{\talpha} \right\|_2 \cdot \left\| \bbX \right \|_\txtop \cdot \left\| \bDelta_\beta \right\|_2
	\lesssim \frac{\eta_x \sqrt{|\AR|} \Lambda_\beta}{s_{\br}}. 
	\label{eq:t5}
\end{align}

\noindent {\em Term 6.} On $\Gc_\noise \cap \Gc_{\PCR, \beta}$, we use \eqref{eq:delta.g.bound} 
\begin{align}
	\left| \left \langle \bDelta_{\talpha}, \bXi_X^\top \bDelta_\beta \right \rangle \right|
	\le \left\| \bDelta_{\talpha} \right\|_2 \cdot \left\| \bXi_X \right \|_\txtop \cdot \left\| \bDelta_\beta \right\|_2
	\lesssim \frac{\eta_x^2 \sqrt{|\AR|} \Lambda_\beta}{s^2_{\br}}. 
	\label{eq:t6}
\end{align}

\noindent {\em Collecting terms.} On the master event, $\Gc_\star \coloneqq \Gc_\noise \cap \Gc_{\PCR, \beta} \cap \Gc_y \cap \Gc_z$, we combine \eqref{eq:t1}, \eqref{eq:t2}, \eqref{eq:t3}, \eqref{eq:t4}, \eqref{eq:t5}, and \eqref{eq:t6} to obtain 
\begin{align}
	\left| R \right|
	\lesssim 
	\frac{\eta_x \sqrt{\br} \Lambda_\beta}{\sqrt{|\nAC|}} 
	+ \sigma \sqrt{L} \Lambda_\beta
	+ \frac{\sigma  \eta_x \sqrt{L |\nAR|}}{s^2_{\br}}
	+ \frac{\eta_x^2 \sqrt{|\ARC|}}{s^3_{\br}}
	+ \frac{\eta_x \sqrt{|\AR|} \Lambda_\beta}{s_{\br}}
	+ \frac{\eta_x^2 \sqrt{|\AR|} \Lambda_\beta}{s^2_{\br}}. 
\end{align}
Taking a union bound over \eqref{eq:hp1}, \eqref{eq:hp2}, and \eqref{eq:hp3}, we conclude $\Pb(\Gc_\star^c \mid \Ec) \lesssim |\nARC|^{-10}$. 
Leveraging \eqref{eq:simplify} to simplify the above bound concludes the proof.
\end{proof}

\section{Proof of Theorem~\ref{thm:var.est.asymp}} 

\begin{proof}[Proof of Theorem~\ref{thm:var.est.asymp}] 
Condition on $\Ec$. 
Define the index set $\Sc$ as in \eqref{eq:index}. 
%
%
For $s = (u, v) \in \Sc$, define the population weights $\bomega = [\omega_{uv}]$ as 
\begin{align}
	\omega_{uv} = \begin{cases}
		\left(\beta^*_\ell\right)^2, &(u,v) = (\ell, j), \ell \in \AR,
		\\
		\left(\alpha^*_q\right)^2, &(u,v) = (i, q), q \in \AC,
		\\
		\left(\beta^*_\ell \alpha^*_q \right)^2,  &(u,v) = (\ell, q), \ell \in \AR, q \in \AC. 
	\end{cases}
\end{align}
Analogously, define the empirical weights $\hbomega = [\homega_{uv}]$ as 
\begin{align}
	\homega_{uv} = \begin{cases}
		\hbeta^2_\ell, &(u,v) = (\ell, j), \ell \in \AR,
		\\
		\halpha_q^2, &(u,v) = (i, q), q \in \AC,
		\\
		\hbeta_\ell^2 \halpha_q^2,  &(u,v) = (\ell, q), \ell \in \AR, q \in \AC. 
	\end{cases}
\end{align}
Accordingly, we rewrite the lead asymptotic variance and the plug-in estimator as
\begin{align}
	\upsilon^2 = \sum_{s \in \Sc} \omega_s \sigma^2_s,
	\quad
	\hupsilon^2 = \sum_{s \in \Sc} \homega_s \hsigma^2_s. 
\end{align}
The estimation error then decomposes as 
\begin{align}
	\hupsilon^2 - \upsilon^2 = \sum_{s \in \Sc} \omega_s \left(\hsigma_s^2 - \sigma_s^2\right) + \sum_{s \in \Sc} \left(\homega_s - \omega_s \right) \hsigma^2_s. 
\end{align}
We proceed to bound each term separately.

\medskip \noindent {\em Term 1.} 
Observe that
\begin{align}
	\left| \sum_{s \in \Sc} \omega_s \left(\hsigma_s^2 - \sigma_s^2\right) \right|
	&\le \sum_{s \in \Sc} \omega_s \left| \hsigma_s^2 - \sigma_s^2 \right| 
	\\
	&\le \max_{s \in \Sc} \left| \hsigma_s^2 - \sigma_s^2 \right| \cdot \left(\left\| \bbeta^* \right\|_2^2 + \left\| \balpha^* \right\|_2^2 + \left\| \bbeta^* \right\|_2^2 \cdot \left\| \balpha^* \right\|_2^2\right).
	\label{eq:var.t0}
\end{align}

\noindent {\em Term 2.}
Next, note that
\begin{align}
	\left| \sum_{s \in \Sc} \left(\homega_s - \omega_s \right) \hsigma^2_s \right| 
	\le \sum_{s \in \Sc} \left|\homega_s - \omega_s \right| \hsigma^2_s
	\le \sigma_+^2 \sum_{s \in \Sc} \left|\homega_s - \omega_s \right|. 
\end{align} 
Expanding, we have
\begin{align}
	\sum_{s \in \Sc} \left| \homega_s - \omega_s \right|
	\le \sum_{\ell \in \AR} \left| \hbeta_\ell^2 - \left(\beta^*_\ell\right)^2 \right| 
	+ \sum_{q \in \AC} \left| \halpha_q^2 - \left(\alpha^*_q\right)^2 \right| 
	+ \sum_{\ell \in \AR} \sum_{q \in \AC} \left| \hbeta_\ell^2 \halpha_q^2 - \left(\beta^*_\ell \alpha^*_q \right)^2 \right|. 
\end{align}
To enable progress, we utilize the following algebraic fact: for vectors $\hbx = [\hx_i] \in \Rb^n$ and $\bx = [x_i] \in \Rb^n$, 
\begin{align}
	\sum_{i=1}^n \left| \hx^2_i - x^2_i \right| \le \left\| \hbx - \bx \right\|_2 \cdot \left( 2 \left\| \bx \right\|_2 + \left\| \hbx - \bx \right\|_2 \right). 
	\label{eq:alg.simple}
\end{align} 
Hence, by \eqref{eq:alg.simple}, we have
\begin{align}
	\sum_{\ell \in \AR} \left| \hbeta_\ell^2 - \left(\beta^*_\ell\right)^2 \right| 
	&\le \left \| \bDelta_\beta \right\|_2 \cdot \left( 2 \left\| \bbeta^* \right\|_2 + \left\| \bDelta_\beta \right\| \right),
	\\
	\sum_{q \in \AC} \left| \halpha_q^2 - \left(\alpha^*_q\right)^2 \right|
	&\le \left \| \bDelta_\alpha \right\|_2 \cdot \left( 2 \left\| \balpha^* \right\|_2 + \left\| \bDelta_\alpha \right\| \right). 
	\label{eq:var.t12}
\end{align}
By a similar argument, we obtain 
\begin{align}
	 \sum_{\ell \in \AR} \sum_{q \in \AC} \left| \hbeta_\ell^2 \halpha_q^2 - \left(\beta^*_\ell \alpha^*_q \right)^2 \right|
	 &\le \sum_{\ell \in \AR} \sum_{q \in \AC} \halpha_q^2 \cdot \left| \hbeta_\ell^2 - \left(\beta^*_\ell\right)^2 \right| + \sum_{\ell \in \AR} \sum_{q \in \AC}  \left(\beta^*_\ell\right)^2 \cdot \left| \halpha_q^2 - \left(\alpha^*_q\right)^2 \right|
	 \\
	 &= \left\| \hbalpha \right\|_2^2  \sum_{\ell \in \AR}  \left| \hbeta_\ell^2 - \left(\beta^*_\ell\right)^2 \right| + \left\| \bbeta^* \right\|_2^2 \sum_{q \in \AC} \left| \halpha_q^2 - \left(\alpha^*_q\right)^2 \right|.
	 \label{eq:var.t3}
\end{align}
Leveraging $\| \hbalpha \|_2 \le \| \balpha^* \|_2 + \| \bDelta_\alpha \|_2$ with \eqref{eq:var.t12} and \eqref{eq:var.t3},
\begin{align}
	\sum_{s \in \Sc} \left| \homega_s - \omega_s \right|
	&\lesssim 
		\left\| \bDelta_\beta \right\|_2  \left( \left\| \bbeta^* \right\|_2 + \left\| \bDelta_\beta \right\|_2 \right) \left(1 + \left[ \left\| \balpha^* \right\|_2 + \left\| \bDelta_\alpha \right\|_2 \right]^2 \right) 
		\\
		&\quad + \left\| \bDelta_\alpha \right\|_2  \left(\left\| \balpha^* \right\|_2 + \left\| \bDelta_\alpha \right\|_2 \right) \left( 1 +  \left\| \bbeta^* \right\|_2^2 \right). 
		\label{eq:var.t1}
\end{align}

\noindent {\em Deterministic bound.} 
Collecting \eqref{eq:var.t0} and \eqref{eq:var.t1} yields the deterministic bound
\begin{align}
	\left| \hupsilon^2 - \upsilon^2 \right|
	&\lesssim 
	\max_{s \in \Sc} \left| \hsigma_s^2 - \sigma_s^2 \right| \cdot \left(\left\| \bbeta^* \right\|_2^2 + \left\| \balpha^* \right\|_2^2 + \left\| \bbeta^* \right\|_2^2 \cdot \left\| \balpha^* \right\|_2^2\right)
	\\
	&\quad 
	+ \sigma^2_+ \cdot \left\| \bDelta_\beta \right\|_2  \left( \left\| \bbeta^* \right\|_2 + \left\| \bDelta_\beta \right\|_2 \right) \left(1 + \left[ \left\| \balpha^* \right\|_2 + \left\| \bDelta_\alpha \right\|_2 \right]^2 \right) 
	\\
	&\quad + \sigma^2_+ \cdot \left\| \bDelta_\alpha \right\|_2 \left(\left\| \balpha^* \right\|_2 + \left\| \bDelta_\alpha \right\|_2 \right) \left( 1 +  \left\| \bbeta^* \right\|_2^2 \right). 
\end{align}

\noindent {\em High-probability bound.} 
Define the event
\begin{align}
	\Gc_\sigma \coloneqq \left\{ \max_{s \in \Sc} \left| \hsigma^2_s - \sigma_s^2 \right| \le \Lambda_\sigma \right\}.
\end{align}
By \eqref{eq:var.bound}, we have $\Pb(\Gc_\sigma^c \mid \Ec) \le p_\sigma$. 
%
%
Continuing, define $\Gc_{\PCR, \beta}$ as in \eqref{eq:events.pcr.noise} and recall $\Pb(\Gc_{\PCR, \beta}^c \mid \Ec) \lesssim |\nARC|^{-10}$ by Lemma~\ref{lemma:param.est.1}. 
Analogously, define
\begin{align}
	\Gc_{\PCR, \alpha} \coloneqq \left\{ \left\| \bDelta_\alpha \right \|_2 \le C_0 \Lambda_\alpha \right\},
\end{align}
where $C_0 > 0$ is an absolute constant and 
\begin{align}
	\Lambda_\alpha \coloneqq \frac{\sigma \br} { \sqrt{|\nAC|} \cdot \min\left\{\sqrt{|\nAR|}, \sqrt{|\nAC|}\right\}} 
		+ \frac{\sqrt{\sigma \br} \cdot L^{1/4}}{\sqrt{|\nAC|} \cdot |\nAR|^{1/4}}. 
\end{align}
Following the proof of Lemma~\ref{lemma:param.est.1} applied to the transposed regression problem of regressing $\bz$ on $\hbX$, we conclude that $\Pb(\Gc_{\PCR, \alpha}^c \mid \Ec) \lesssim |\nARC|^{-10}$. 
Therefore, by \eqref{eq:beta.l2.bound.2} and \eqref{eq:riesz.bound}, on $\Gc_\sigma \cap \Gc_{\PCR, \alpha} \cap \Gc_{\PCR, \beta}$, 
\begin{align}
	&\left| \hupsilon^2 - \upsilon^2 \right|
	\lesssim 
	\Lambda_\sigma \left( \frac{\br}{\min\{|\nAR|, |\nAC|\}} + \frac{\br^2}{|\nARC|} \right)
	\\
	&\qquad + \left(\sigma^2_+ + \Lambda_\sigma \right)
	\left[ \Lambda_\beta \left( \frac{\sqrt{\br}}{\sqrt{|\nAR|}} + \Lambda_\beta \right) \left\{ 1 + \left(\frac{\sqrt{\br}}{\sqrt{|\nAC|}} + \Lambda_\alpha \right)^2 \right\}
	+ \Lambda_\alpha \left(\frac{\sqrt{\br}}{\sqrt{|\nAC|}} + \Lambda_\alpha \right) \left(1 + \frac{\br}{|\nAR|}\right) 
	\right].
\end{align}
Taking a union bound, the above inequality holds with probability at least $1 - p_\sigma - \Oc(|\nARC|^{-10})$. 
The variance-perturbation envelope in \eqref{eq:gamma} is then obtained by observing $\Lambda_\alpha \lesssim \Phi / \sqrt{|\nAC|}$ and $\Lambda_\beta \lesssim \Phi / \sqrt{|\nAR|}$. 
Finally, since $\Gamma / \upsilon^2 = o(1)$ by assumption, $\hupsilon^2 / \upsilon^2 \xrightarrow{p} 1$ and $\hupsilon / \upsilon \xrightarrow{p} 1$.  
Combining this with the condition $\Psi / \upsilon = o(1)$ from Theorem~\ref{thm:normality} and applying Slutsky's theorem yields 
\begin{align}
	\frac{\hA_{ij} - A_{ij}}{\hupsilon} &= 
	\frac{\hA_{ij} - A_{ij}}{\upsilon} \cdot \frac{\upsilon}{\hupsilon} 
	\rightsquigarrow \Nc(0,1). 
\end{align}
The proof is complete. 
\end{proof}


\section{Proof of Theorem~\ref{thm:var.est}} \label{sec:proof.var.est} 
%
%
We divide our proof of Theorem~\ref{thm:var.est} into two parts: Appendix~\ref{sec:proof.var.est.bounded} proves \eqref{eq:var.est.bounded} and Appendix~\ref{sec:proof.var.est.general} proves \eqref{eq:var.est.general}.
In both arguments, we will anchor on the proofs in Appendix~\ref{sec:proof.consistency} that established Theorem~\ref{thm:consistency}. 
There are, however, two primary distinctions in the proofs below compared to those in Appendix~\ref{sec:proof.consistency}:
(i) the subgaussian noise terms $\varepsilon_{ij}$ are replaced by the sub-exponential noise terms $\tY_{ij}^2 - B_{ij}$, which satisfy
\begin{align}
	\left\| \tY_{ij}^2 - B_{ij} \right\|_{\psi_1} 
	&\le\left \| \tY_{ij}^2 \right\|_{\psi_1} 
	\\
	&\le 2 \cdot \left\| A^2_{ij} \right\|_{\psi_1} + 2 \cdot \left\| \varepsilon^2_{ij} \right\|_{\psi_1} 
	\\
	&= 2 \cdot \left\| A_{ij} \right\|^2_{\psi_2} + 2 \cdot \left\| \varepsilon_{ij} \right\|^2_{\psi_2} 
	\\
	&\lesssim 1 + \sigma^2_{ij} 
	\\
	&\lesssim 1 + \sigma^2; 
	\label{eq:subexp.noise.psi1}
\end{align}
(ii) $\rank(\bA_{\AR, \AC})$ is replaced by $\rank(\bB_{\AR, \AC})$, which satisfies 
\begin{align}
	\rank(\bB_{\AR, \AC}) \le \rank(\bSigma_{\AR, \AC}) + \rank(\bA_{\AR, \AC})^2
	= \btau + \br^2, \label{eq:rank.B} 
\end{align}
where $\btau \coloneqq \rank(\bSigma_{\AR, \AC})$. 
%
%
With these observations in hand, we are ready to prove Theorem~\ref{thm:var.est}.

\subsection{Bounded Noise} \label{sec:proof.var.est.bounded}
%

\begin{proof}[Proof of \eqref{eq:var.est.bounded}] 
By construction, we have 
\begin{align}
	\left| \hsigma^{2, \clip}_{ij} - \sigma^2_{ij} \right|
	&= 
	\left|\hB^\clip_{ij} - B_{ij}  +  A^2_{ij} - \left(\hA^{\clip}_{ij}\right)^2 \right|
	\le \left| \hB^\clip_{ij} - B_{ij} \right| + \left| \left(\hA^{\clip}_{ij} \right)^2 - A^2_{ij} \right|. \label{eq:var.est.bounded.1}
\end{align} 
To enable progress, we decompose the second term on the RHS of \eqref{eq:var.est.bounded.1} as
\begin{align}
	\left| \left(\hA^{\clip}_{ij}\right)^2 - A^2_{ij} \right|
	&= 
	\left| \left(\hA^{\clip}_{ij}- A_{ij} \right) \cdot \left(\hA^{\clip}_{ij} + A_{ij} \right) \right|
	\\
	&\le 2 \cdot \left| \hA^{\clip}_{ij}- A_{ij} \right|, &&\because \left| \hA^{\clip}_{ij} + A_{ij} \right| \le 2
	\label{eq:var.est.bounded.2}
\end{align} 
which can be further bounded by a direct application of Theorem~\ref{thm:consistency}. 

Accordingly, we turn our attention to the first term on the RHS of \eqref{eq:var.est.bounded.1}. 
Since the entries $\varepsilon_{ij}$ of $\bE$ are assumed to be bounded, it follows that $\tY_{ij}^2 - B_{ij}$ remain subgaussian random variables \citep{vershynin2018high}. 
Hence, we can again directly apply Theorem~\ref{thm:consistency} to obtain a similar bound with the primary changes being (i) $\| \varepsilon_{ij} \|_{\psi_2} \lesssim \sigma$ is replaced with $\| \tY^2_{ij} - B_{ij} \|_{\psi_2} \lesssim 1 + \sigma^2$, as per \eqref{eq:subexp.noise.psi1}; and (ii) $\rank(\bA_{\AR, \AC}) = \br$ is replaced with $\rank(\bB_{\AR, \AC}) \le \br^2 + \btau$, as per \eqref{eq:rank.B}. 
It is not difficult to verify that the resulting bound dominates the bound on \eqref{eq:var.est.bounded.2}. 
This completes the proof. 
\end{proof}

\subsection{General subgaussian Noise} \label{sec:proof.var.est.general}
%

\begin{proof}[Proof of \eqref{eq:var.est.general}] 
We begin with the decomposition
\begin{align}
	\Ex\left[ \hsigma^2_{ij} - \sigma^2_{ij} \right]
	&=
	\Ex\left[ \hB_{ij} - B_{ij}  -  \left(\hA^{\clip}_{ij}\right)^2 + A^2_{ij} \right]
	\\
	&= \Ex\left[ \hB_{ij} - B_{ij}\right] + \Ex\left[A^2_{ij} - \left(\hA^{\clip}_{ij}\right)^2 \right]
	\\
	&= \Ex\left[ \hB_{ij} - B_{ij}\right] + \Ex\left[\left(\hA^{\clip}_{ij} - A_{ij} \right) \cdot \left(\hA^{\clip}_{ij} + A_{ij} \right) \right]. 
	\label{eq:var.est.general.1}
\end{align}
We bound the second term on the RHS of \eqref{eq:var.est.general.1} as 
\begin{align}
	\Ex\left[\left(\hA^{\clip}_{ij} - A_{ij} \right) \cdot \left(\hA^{\clip}_{ij} + A_{ij} \right) \right]
	&\le
	\left| \Ex\left[\left(\hA^{\clip}_{ij} - A_{ij} \right) \cdot \left(\hA^{\clip}_{ij} + A_{ij} \right) \right] \right|
	\\
	&\le
	\Ex\left[ \left| \left(\hA^{\clip}_{ij} - A_{ij} \right) \cdot \left(\hA^{\clip}_{ij} + A_{ij} \right) \right| \right] && \because \text{~Jensen's inequality} 
	\\
	&=
	\Ex\left[ \left|\hA^{\clip}_{ij} - A_{ij} \right| \cdot \left| \hA^{\clip}_{ij} + A_{ij}  \right| \right] 
	\\
	&\le 2 \cdot \Ex\left[ \left| \hA^{\clip}_{ij} - A_{ij} \right| \right]. && \because \left| \hA^{\clip}_{ij} + A_{ij} \right| \le 2
	\label{eq:var.est.general.2}
\end{align}
%
%
Consider any $\epsilon > 0$. 
Define the event $\Gc \coloneqq \{ \mid \hA^{\clip}_{ij} - A_{ij} \mid \le \epsilon \}$ and its complement $\Gc^c \coloneqq \{| \hA^{\clip}_{ij} - A_{ij} \mid > \epsilon \}$. 
%
Then, we have  
\begin{align}
	\Ex\left[ \left| \hA^{\clip}_{ij} - A_{ij} \right| \right]
	&= 
	\Ex\left[ \left| \hA^{\clip}_{ij} - A_{ij} \right| \cdot \boldsymbol{1}\{\Gc\} \right] 
	+ \Ex\left[ \left| \hA^{\clip}_{ij} - A_{ij} \right| \cdot \boldsymbol{1}\{\Gc^c\} \right].  
\end{align}
On the event $\Gc$, we have $| \hA^{\clip}_{ij} - A_{ij} \mid \le \epsilon$. 
On the event $\Gc^c$, we utilize $| \hA^{\clip}_{ij} - A_{ij} \mid \le 2$, which holds by the construction of $\hA^\clip_{ij}$. 
Accordingly, 
\begin{align}
	\Ex\left[ \left| \hA^{\clip}_{ij} - A_{ij} \right| \right]
	&\le
	\epsilon + 2 \cdot \Pb \left( \left|\hA^\clip_{ij} - A_{ij} \right| > \epsilon \right), \label{eq:var.est.general.3}
\end{align}
which can be bounded by calling Theorem~\ref{thm:consistency} and choosing $\epsilon$ to be the RHS of \eqref{eq:thm.hp.special}. 

Next, we bound the first term on the RHS of \eqref{eq:var.est.general.1}.  
To this end, define the event 
\begin{align}
	\Gc \coloneqq \left\{ \left\| \bX \circ \bX - \bB_{\nAR, \nAC} \right\|_{\txtop} 
	\le C \sqrt{|\nAR|} + \left(1+\sigma^2 \right) \sqrt{|\nAC|} L^{3/2} \right\}
\end{align}
for some absolute constnat $C > 0$, 
which holds with probability at least $1-\mathcal{O}( |\nARC|^{-10})$ by Lemma~\ref{lemma:dg}. 
We utilize the decomposition 
\begin{align}
	\Ex \left[ \hB_{ij} - B_{ij} \right]
	&=
	\Ex\left[ \hB_{ij} - B_{ij} \mid \Gc \right] \cdot \Pb(\Gc) + \Ex \left[\hB_{ij} - B_{ij} \mid \Gc^c \right] \cdot \Pb(\Gc^c)
	\\
	&\le \Ex\left[ \hB_{ij} - B_{ij} \mid \Gc \right] + \Ex \left[\hB_{ij} - B_{ij} \mid \Gc^c \right] \cdot \Pb(\Gc^c), && \because \Pb(\Gc) \le 1
	\label{eq:var.est.general.4}
\end{align}
to arrive at the inequality 
\begin{align}
	\Ex \left[ \hB_{ij} - B_{ij} \right]
	&\le
	\left| \Ex \left[ \hB_{ij} - B_{ij} \right] \right|
	\\ & \le \left| \Ex \left[ \hB_{ij} - B_{ij} \mid \Gc \right] \right|
	+ \left| \Ex \left[ \hB_{ij} - B_{ij} \mid \Gc^c \right] \right| \cdot \Pb(\Gc^c). 
	\label{eq:var.est.general.5}
\end{align}
We will bound each term on the RHS of \eqref{eq:var.est.general.5} separately. 
For the second term, we have 
\begin{align}
	\left| \Ex \left[ \hB_{ij} - B_{ij} \mid \Gc^c \right] \right| \cdot \Pb(\Gc^c)
	&\le \Ex \left[ \left| \hB_{ij} - B_{ij} \right| \mid \Gc^c \right] \cdot \Pb(\Gc^c)
	&&\because \text{Jensen's inequality} 
	\\ 
	&\lesssim M \cdot |\nARC|^{-10}.  
	&& \because \left| \hB_{ij} - B_{ij} \right| \le M
	\label{eq:var.est.general.6} 
\end{align}
The first term, on the other hand, can be controlled following the proof of Theorem~\ref{thm:consistency} with the primary changes being 
(i) $\| \varepsilon_{ij} \|_{\psi_2} \lesssim \sigma$ is replaced with $\| \tY^2_{ij} - B_{ij} \|_{\psi_1} \lesssim 1 + \sigma^2$, as per \eqref{eq:subexp.noise.psi1}; 
(ii) Lemma~\ref{lemma:subg_matrix} is replaced with Lemma~\ref{lemma:dg}; 
and (iii) $\rank(\bA_{\AR, \AC}) = \br$ is replaced with $\rank(\bB_{\AR, \AC}) \le \br^2 + \btau$, as per \eqref{eq:rank.B}. 
For these reasons, we do not write out the proof in its entirety to reduce redundancies. 
Nevertheless, we obtain
\begin{align}
	\left|\Ex\left[ \hB_{ij} - B_{ij} \mid \Gc \right] \right|
	 \lesssim 
	\frac{\left(1 + \sigma^4 \right) \left(\br^2 + \btau \right) }{\min \left\{ \sqrt{|\nAR|} \cdot L^{-3/2}, \sqrt{|\nAC|} \right\}}.
\label{eq:var.est.general.7}
\end{align}
Observe that the RHS of \eqref{eq:var.est.general.7} dominates $\epsilon$, as defined in \eqref{eq:var.est.general.3}. Hence, simplifying the resulting inequality completes the proof. 
\end{proof}

\end{document}